\documentclass{amsart}
\usepackage{amssymb}
\usepackage{amsmath}
\usepackage{amsthm}
\usepackage{mathrsfs}
\usepackage{colortbl}
\usepackage{fancybox}

\usepackage{graphicx}
\usepackage{ascmac}

\newtheorem{theorem}{Theorem}
\newtheorem{lemma}[theorem]{Lemma}
\newtheorem{assumption}[theorem]{Assumption}
\newtheorem{remark}[theorem]{Remark}

\theoremstyle{definition}
\newtheorem{definition}[theorem]{Definition}
\newtheorem{convention}[theorem]{Convention}

\newtheorem{proposition}[theorem]{Proposition}

\makeatletter

\@addtoreset{theorem}{section}
\makeatother

\makeatletter

\@addtoreset{equation}{section}
\makeatother

\begin{document}                                                 
\title[On Compressible Fluid Flows on an Evolving Surface]{On Derivation of Compressible Fluid Systems on an Evolving Surface}



\author{Hajime Koba}
\address{ Graduate School of Engineering Science, Osaka University\\
1-3 Machikaneyamacho, Toyonaka, Osaka, 560-8531, Japan, Japan}
\email{iti@sigmath.es.osaka-u.ac.jp}
\thanks{This work was partly supported by the Japan Society for the Promotion of Science (JSPS) KAKENHI Grant Numbers JP25887048 and JP15K17580.}

\date{}
                      

\begin{abstract}
We consider the governing equations for the motion of compressible fluid on an evolving surface from both energetic and thermodynamic points of view. We employ our energetic variational approaches to derive the momentum equation of our compressible fluid systems on the evolving surface. Applying the first law of thermodynamics and the Gibbs equation, we investigate the internal energy, enthalpy, entropy, and free energy of the fluid on the evolving surface. We also study conservative forms and conservation laws of our compressible fluid systems on the evolving surface. Moreover, we derive the generalized heat and diffusion systems on an evolving surface from an energetic point of view. This paper gives a mathematical validity of the surface stress tensor determined by the Boussinesq-Scriven law. Using a flow map on an evolving surface and applying the Riemannian metric induced by the flow map are key ideas to analyze fluid flow on the evolving surface.
\end{abstract}

\maketitle


\section{Introduction}\label{sect1}

\subsection{The Purposes and Key Ideas of this Paper}

Interface flow and surface flow play an important role in fluid dynamics such as soap bubbles in air, the atmosphere and ocean on the earth, and phase transition in complex fluids. One can consider surface flow as fluid flow on an evolving surface. An evolving surface means that the surface is moving or that the shape of the surface is changing with along the time. In this paper we consider the thickness of evolving surfaces as zero. 

The aim of this paper is to derive several governing equations for the motion of compressible fluid on an evolving surface from both energetic and thermodynamic points of view. We apply both our energetic variational approach and the first law of thermodynamics to derive several compressible fluid systems on the evolving surface, and we investigate the entropy and free energy of the fluid on the evolving surface by making use of the second law of thermodynamics. This paper derives the surface stress tensor determined by the Boussinesq-Scriven law (Boussinesq \cite{Bou13}, Scriven \cite{Scr60}) from both the energy dissipation due to the viscosities and the work done by the pressure of the fluid on an evolving surface to give a mathematical validity of the Boussinesq-Scriven law. Of course, this paper provides one possibility of the dominant equations for the motion of compressible fluid on an evolving surface. However, employing a similar technique of this paper, we can derive several compressible Navier-Stokes systems in domains and make a mathematical model of two-phase flow with surface flow and surface tension.

We first introduce fundamental notation. Let $t \geq 0$ be the time variable, $ x (= { }^t ( x_1,x_2,x_3) )$, $\xi (= { }^t ( \xi_1, \xi_2 ,\xi_3)) \in \mathbb{R}^3$ the spatial variables, and $X (= { }^t (X_1 , X_2) ) \in \mathbb{R}^2$ the spatial variables. Let $T \in (0, \infty]$ and $\Gamma (t) ( = \{ \Gamma (t) \}_{0 \leq t < T})$ be a smoothly evolving $2$-dimensional surface in $\mathbb{R}^3$ depending on time $t$. The notation $u = u (x,t) = { }^t (u_1,u_2,u_3)$ and $w = w (x,t) ={ }^t (w_1,w_2,w_3)$ represent the \emph{relative fluid velocity} of a fluid particle at a point $x = { }^t (x_1,x_2,x_3)$ of the evolving surface $\Gamma (t)$ and the \emph{motion velocity} at a point $x$ of $\Gamma (t)$ which determines the motion of the evolving surface $\Gamma (t)$, respectively. We often call $u$ a \emph{surface flow (surface velocity)} on the evolving surface and $w$ the \emph{speed} of the evolving surface. Assume that $u$ is a tangential vector on $\Gamma (t)$. Remark that $w$ is not a necessary tangential vector on $\Gamma ( t )$. We notice that by introducing the surface flow $u$ and the motion velocity $w$, then there is no exchange of particles between the surface and the environment. The velocity
\begin{equation*}
v = v (x,t) = { }^t(v_1,v_2,v_3) := u + w
\end{equation*}
is defined as the \emph{total velocity} of a fluid particle at a point $x$ of $\Gamma (t)$. In this paper we focus on the total velocity $v$.

The notation $\rho = \rho (x,t)$, $\sigma = \sigma (x,t)$, $e = e (x,t)$, $\theta = \theta (x,t)$, $e_A = e_A (x,t)$, $h = h ( x , t)$, $s = s (x,t)$, and $e_F = e_F (x,t)$ represent the \emph{density}, the (\emph{total}) \emph{pressure}, the \emph{internal energy}, the \emph{temperature}, the \emph{total energy}, the \emph{enthalpy}, the \emph{entropy}, and the \emph{free energy} of the fluid on the evolving surface $\Gamma (t)$, respectively. Note that the total pressure $\sigma$ includes surface tension and surface pressure in general. The symbol $C = C (x,t)$ denotes a \emph{concentration} of a substance in the fluid on the evolving surface. Assume that $u , w ,v , \rho , \sigma, e, \theta, e_A, h, s, e_F , C$ are smooth functions.

The symbols $\mu = \mu (x,t), \lambda = \lambda (x,t)$ are two \emph{viscosity coefficients} of the fluid, $\kappa = \kappa (x , t)$ is the \emph{thermal conductivity} of the fluid, $\nu = \nu (x,t)$ is the \emph{diffusion coefficient} of the concentration $C$, $F = F (x,t) = { }^t (F_1 ,F_2,F_3)$ is the \emph{external force} or \emph{gravity vector}, $Q_\theta = Q_\theta (x,t)$ is the \emph{heat source}, and $Q_C = Q_C (x,t)$ is the \emph{source on the concentration} $C$. We often call $\mu$ the \emph{surface share viscosity} and $\mu + \lambda$ the \emph{surface dilatational viscosity}. Suppose that $\mu, \lambda, \kappa, \nu, F, Q_\theta, Q_C$ are smooth functions.

This paper has six purposes. The first one is to derive and study the following \emph{full compressible fluid system} on the evolving surface $\Gamma (t)$:
\begin{equation}\label{eq11}
\begin{cases}
D_t \rho + ({\rm{div}}_\Gamma v) \rho = 0 & \text{ on } \mathcal{S}_T,\\
\rho D_t v = {\rm{div}}_\Gamma S_{\Gamma} (v, \sigma , \mu , \lambda ) + \rho F& \text{ on } \mathcal{S}_T,\\
\rho D_t e + ({\rm{div}}_\Gamma v) \sigma = {\rm{div}}_\Gamma q_\theta + \rho Q_\theta + \tilde{e}_D& \text{ on } \mathcal{S}_T,\\
D_t C + ({\rm{div}}_\Gamma v ) C = {\rm{div}}_\Gamma q_C + Q_C& \text{ on } \mathcal{S}_T,
\end{cases}
\end{equation}
where
\begin{equation*}
\mathcal{S}_T = \left\{ (x,t) \in \mathbb{R}^4; {  \ } (x,t) \in \bigcup_{0<t<T}\{ \Gamma (t) \times \{ t \} \} \right\},
\end{equation*}
$D_t$ is the \emph{material derivative} defined by $D_t f = \partial_t f + (v , \nabla )f$, $q_\theta = \kappa {\rm{grad}}_\Gamma \theta$, $\tilde{e}_D = 2 \mu | D_\Gamma (v) |^2 + \lambda |{\rm{div}}_\Gamma v |^2$, and $q_C = \nu {\rm{grad}}_\Gamma C$. Here $D_\Gamma (v) = P_\Gamma D (v) P_\Gamma$, $| D_\Gamma (v) |^2 = D_\Gamma (v) : D_\Gamma (v)$, $\partial_t = \partial/\partial t$, $\partial_i = \partial/\partial x_i$, $\nabla = {  }^t ( \partial_1, \partial_2, \partial_3 )$, $(v,\nabla )f = v_1 \partial_1 f + v_2 \partial_2 f + v_3 \partial_3 f$, $D (v) = \{ ( \nabla v) + { }^t (\nabla v) \}/2$, and $P_\Gamma = I_{3 \times 3 } - n \otimes n$, where $\otimes $ denotes the \emph{tensor product} and $n = n(x,t)  = { }^t( n_1 , n_2 , n_3)$ is the \emph{unit outer normal vector} at $x \in \Gamma (t)$.
 The symbols $q_\theta$, $\tilde{e}_D$, and $q_C$ are often called a \emph{heat flux} on the evolving surface, the \emph{density for the energy dissipation due to the viscosities} $\mu,\lambda$, and a \emph{surface flux} on the evolving surface, respectively. We call $D_\Gamma (v)$, $D (v)$, and $P_\Gamma$, the \emph{surface strain rate tensor}, the \emph{strain rate tensor}, and an \emph{orthogonal projection to a tangent space}, respectively. The notation ${\rm{div}}_\Gamma$ denotes \emph{surface divergence}, ${\rm{grad}}_\Gamma$ denotes \emph{surface divergence}, and $S_\Gamma (v , \sigma , \mu , \lambda )$ is the \emph{surface stress tensor} defined by
\begin{equation*}
S_\Gamma ( v , \sigma , \mu , \lambda )= 2 \mu D_\Gamma (v) +  \lambda  P_\Gamma ( {\rm{div}}_\Gamma v) - P_\Gamma \sigma .
\end{equation*}
The surface stress tensor $S_\Gamma (v , \sigma , \mu , \lambda )$ was introduced by Scriven \cite{Scr60}. See Section \ref{sect2} for the definitions of ${\rm{div}}_\Gamma$ and ${\rm{grad}}_\Gamma$.

We can write the system \eqref{eq11} as the following conservative form:
\begin{equation}\label{eq12}
\begin{cases}
D_t^N \rho + {\rm{div}}_\Gamma (\rho v) =0& \text{ on } \mathcal{S}_T,\\
D_t^N (\rho v) + {\rm{div}}_\Gamma (\rho v \otimes v - S_\Gamma (v,\sigma , \mu , \lambda ) ) = \rho F& \text{ on } \mathcal{S}_T,\\
D_t^N e_A + {\rm{div}}_\Gamma ( e_A v - q_\theta - S_\Gamma (v, \sigma , \mu , \lambda ) v )  = \rho Q_\theta + \rho F \cdot v& \text{ on } \mathcal{S}_T,\\
D_t^N C + {\rm{div}}_\Gamma (C v - q_C) = Q_C& \text{ on } \mathcal{S}_T. 
\end{cases}
\end{equation}
Here $e_A$ is the \emph{total energy} defined by $e_A = \rho |v|^2/2 + \rho e$ and $D_t^N$ is the \emph{time derivative with the normal derivative} defined by
\begin{equation*}
D_t^N f = \partial_t f + ( v \cdot n ) (n , \nabla ) f.
\end{equation*}
Note that $v \cdot n = w \cdot n$ since $u \cdot n = 0$. Note also that
\begin{align*}
\tilde{e}_D - ({\rm{div}}_\Gamma v ) \sigma = & {\rm{div}}_\Gamma (S_\Gamma (v , \sigma , \mu , \lambda ) v ) - {\rm{div}}_\Gamma (S_\Gamma (v , \sigma , \mu , \lambda )) \cdot v,\\
D_\Gamma (v) : D_\Gamma (v) = & D_\Gamma (v) : \mathbb{D}_\Gamma (v), 
\end{align*}
where $\mathbb{D}_{\Gamma} (v) = \{ (\nabla_{\Gamma} v) + { }^t (\nabla_{\Gamma} v ) \}/2$. We often call $D_\Gamma (v)$ a \emph{projected strain rate} and $\mathbb{D}_\Gamma (v)$ a \emph{tangential strain rate}. See Section \ref{sect2} for the notation $\nabla_\Gamma$.

Now we consider conservation laws of the system \eqref{eq11}. Assume that $\Gamma (t)$ is flowed by the velocity $v$ and that $F \equiv 0$, $Q_\theta \equiv 0$, $Q_C \equiv 0$. Then we observe that \eqref{eq11} satisfies that for $0< t_1 < t_2 < T$,
\begin{align}
\int_{\Gamma (t_2)} \rho (x,t_2) { \ }d \mathcal{H}^2_x = & \int_{\Gamma (t_1)} \rho (x,t_1) { \ }d \mathcal{H}^2_x ,\label{eq13}\\
\int_{\Gamma (t_2)} \rho (x,t_2) v (x,t_2) { \ }d \mathcal{H}^2_x = & \int_{\Gamma (t_1)} \rho (x,t_1) v (x,t_1) { \ }d \mathcal{H}^2_x ,\label{eq14}\\
\int_{\Gamma (t_2)} e_A (x,t_2) { \ }d \mathcal{H}^2_x = & \int_{\Gamma (t_1)} e_A (x,t_1) { \ }d \mathcal{H}^2_x ,\label{eq15}\\
\int_{\Gamma (t_2)} C (x,t_2) { \ }d \mathcal{H}^2_x = & \int_{\Gamma (t_1)} C (x,t_1) { \ }d \mathcal{H}^2_x \label{eq16},\\
\int_{\Gamma (t_2)} x \times (\rho v ) { \ }d \mathcal{H}^2_x = & \int_{\Gamma (t_1)} x \times ( \rho v ) { \ }d \mathcal{H}^2_x \label{eq17}.
\end{align}
Here $d \mathcal{H}^2_x$ denotes the \emph{$2$-dimensional Hausdorff measure}. We often call \eqref{eq13}, \eqref{eq14}, \eqref{eq15}, and \eqref{eq17}, the law of conservation of mass, the law of conservation of momentum, the law of conservation of the total energy, and the law of conservation of angular momentum, respectively. See Theorem \ref{thm18} for details.

The second one is to investigate the enthalpy, entropy, and free energy of the fluid on the evolving surface $\Gamma (t)$. We now assume that $\mu, \lambda , \kappa ,\rho , \theta $ are positive functions. We set the \emph{enthalpy} $h = h (x,t) = e + {\sigma}/{\rho}$. Then
\begin{equation}\label{eq18}
\rho D_t h = {\rm{div}}_\Gamma q_\theta + \rho Q_\theta + \tilde{e}_D + D_t \sigma \text{ on } \mathcal{S}_T. 
\end{equation}
This is equivalent to
\begin{equation*}
D_t^N (\rho h) + {\rm{div}}_\Gamma (\rho h v - q_\theta ) - \rho Q_\theta = \tilde{e}_D + D_t \sigma \text{ on } \mathcal{S}_T.
\end{equation*}
Suppose that the following \emph{Gibbs equation} holds:
\begin{equation*}
D_t e = \theta D_t s - \sigma D_t \left( \frac{1}{\rho } \right) \text{ on } \mathcal{S}_T .
\end{equation*}
Then we check that the \emph{entropy} $s = s ( x , t )$ fulfills
\begin{equation}\label{eq19}
\rho D_t s = \frac{1}{\theta} \{ {\rm{div}}_\Gamma q_\theta+ \rho Q_\theta + \tilde{e}_D \} \text{ on } \mathcal{S}_T.
\end{equation}
Moreover, we obtain the \emph{Clausius-Duhem inequality}:
\begin{equation*}
D_t^N (\rho s ) + {\rm{div}}_\Gamma \left( \rho s v - \frac{ q_\theta}{\theta} \right) - \frac{ \rho Q_\theta }{\theta}= \frac{ \tilde{e}_D}{\theta} + \frac{ q_\theta\cdot {\rm{grad}}_\Gamma \theta }{\theta^2} \geq 0 \text{ on } \mathcal{S}_T.
\end{equation*}
We also observe that
\begin{equation}\label{eq110}
\rho D_t e_F + \rho s D_t \theta - S_\Gamma ( v , \sigma , \mu , \lambda ): D_\Gamma (v)  = - \tilde{e}_D \leq 0 \text{ on } \mathcal{S}_T. 
\end{equation}
Here $e_F$ is the \emph{free energy} defined by $e_F = e - \theta s$. Remark that $\tilde{e}_D = 2 \mu | D_\Gamma (v) |^2 + \lambda | {\rm{div}}_\Gamma v |^2 \geq 0$ and that $q_\theta \cdot {\rm{grad}}_\Gamma \theta = \kappa | {\rm{grad}}_\Gamma \theta |^2 \geq 0$.

The third one is to derive compressible fluid systems on an evolving surface from a variational point of view. We easily have the momentum equation of the system \eqref{eq11} if we assume
\begin{equation*}
\rho D_t v = {\rm{div}}_\Gamma S_\Gamma (v , \sigma , \mu , \lambda ) + \rho F .
\end{equation*}
However, applying our variational methods, we can derive other type of systems of compressible fluid flow on the evolving surface. For example, we use our methods to derive the following \emph{tangential compressible} and \emph{non-canonical compressible fluid systems} on an evolving surface:
\begin{equation}\label{eq111}
\begin{cases}
D_t \rho + ({\rm{div}}_\Gamma v) \rho = 0 & \text{ on } \mathcal{S}_T,\\
P_\Gamma \rho D_t v = P_\Gamma {\rm{div}}_\Gamma S_\Gamma (v , \sigma , \mu , \lambda ) + P_\Gamma \rho F& \text{ on } \mathcal{S}_T, \\
v \cdot n = 0 & \text{ on } \mathcal{S}_T,\\
\rho D_t e + ({\rm{div}}_\Gamma v) \sigma = {\rm{div}}_\Gamma q_\theta+ \rho Q_\theta + \tilde{e}_D& \text{ on } \mathcal{S}_T,\\
D_t C + ({\rm{div}}_\Gamma v ) C = {\rm{div}}_\Gamma q_C + Q_C& \text{ on } \mathcal{S}_T,
\end{cases}
\end{equation}
and
\begin{equation}\label{eq112}
\begin{cases}
D_t \rho + ({\rm{div}}_\Gamma v) \rho = 0 & \text{ on } \mathcal{S}_T,\\
\rho D_t v + {\rm{grad}}_\Gamma \sigma + \sigma H_\Gamma n = P_\Gamma {\rm{div}}_\Gamma S_\Gamma ( u , 0 , \mu , \lambda ) + \rho F& \text{ on } \mathcal{S}_T.
\end{cases}
\end{equation}
Here $H_\Gamma = H_\Gamma (x,t)$ denotes the mean curvature of $\Gamma (t)$ in the direction of the unit normal outer vector $n$. See Section \ref{sect5} for our energetic variational approaches for the two systems \eqref{eq111} and \eqref{eq112}.

The fourth one is to give a mathematical validity of the \emph{surface stress tensor determined by the Boussinesq-Scriven law} (Boussinesq \cite{Bou13}, Scriven \cite{Scr60}):
\begin{equation*}
S_\Gamma (v , \sigma , \mu , \lambda ) = 2 \mu D_\Gamma (v) + \lambda P_\Gamma ({\rm{div}}_\Gamma v) - P_\Gamma \sigma .
\end{equation*}
This paper derives the surface stress tensor from both the energy dissipation due to the viscosities $\mu, \lambda$ and the work done by the pressure $\sigma$ of the fluid on an evolving surface. More precisely, from Theorem \ref{thm15}, Propositions \ref{prop312}, and \cite[Theorem 2.1]{KS17}, we see a mathematical validity of the Boussinesq-Scriven law.

The fifth one is to derive the following two \emph{barotropic compressible fluid systems} on an evolving surface:
\begin{equation}\label{eq113}
\begin{cases}
D_t \rho + ({\rm{div}}_\Gamma v) \rho = 0& \text{ on } \mathcal{S}_T,\\
\rho D_t v + {\rm{grad}}_\Gamma \mathfrak{p} + \mathfrak{p} H_\Gamma n =0& \text{ on } \mathcal{S}_T,\\
\mathfrak{p} = \mathfrak{p} (\rho )= \rho p' (\rho ) - p (\rho ), 
\end{cases}
\end{equation}
and
\begin{equation}\label{eq114}
\begin{cases}
D_t \rho + ({\rm{div}}_\Gamma v) \rho = 0& \text{ on } \mathcal{S}_T,\\
P_\Gamma \rho D_t v + {\rm{grad}}_\Gamma \mathfrak{p} =0& \text{ on } \mathcal{S}_T,\\
\mathfrak{p} = \mathfrak{p} (\rho ) = \rho p' (\rho ) - p ( \rho ), 
\end{cases}
\end{equation}
where $p ( \cdot )$ is a smooth function. We often call $p ( \rho )$ the \emph{chemical potential}. See Theorem \ref{thm19} for details.

The sixth one is to derive the following \emph{generalized heat and diffusion systems} on an evolving surface:
\begin{equation}\label{eq115}
\begin{cases}
D_t \rho + ({\rm{div}}_\Gamma v ) \rho = 0& \text{ on } \mathcal{S}_T,\\
\rho D_t ( C_\theta \theta ) = {\rm{div}}_\Gamma q_{\mathcal{J}_1} + \rho Q_\theta + \mathcal{F}_1 & \text{ on } \mathcal{S}_T,
\end{cases}
\end{equation}
and
\begin{equation}\label{eq116}
D_t C + ({\rm{div}}_\Gamma v ) C = {\rm{div}}_\Gamma q_{\mathcal{J}_2} + Q_C + \mathcal{F}_2 { \ }\text{ on } \mathcal{S}_T.
\end{equation}
Here $q_{\mathcal{J}_1}$ is the \emph{generalized heat flux} defined by $q_{ \mathcal{J}_1} = e_{\mathcal{J}_1}' ( | {\rm{grad}}_\Gamma \theta |^2 ) {\rm{grad}}_\Gamma \theta$, $q_{\mathcal{J}_2}$ is the \emph{generalized surface flux} defined by $q_{ \mathcal{J}_2} = e_{\mathcal{J}_2}' ( | {\rm{grad}}_\Gamma C |^2 ){\rm{grad}}_\Gamma C$, where $e_{\mathcal{J}_1}$ and $e_{\mathcal{J}_2}$ are $C^1$-functions, $\mathcal{F}_1 = \mathcal{F}_1 (x,t), \mathcal{F}_2 = \mathcal{F}_2 (x,t)$ are two smooth functions depending on the situation and environment, and $C_\theta = C_\theta (x,t)$ is a smooth function which is called the \emph{specific heat} of the fluid. See Section \ref{sect5} for our derivation of the two systems \eqref{eq115} and \eqref{eq116}.

\begin{remark}\label{rem11}
$(\rm{i})$ We do not use the assumption that $\sigma = \sigma (\rho, e)$ when we derive the systems \eqref{eq11} and \eqref{eq111} by applying our methods. While, we assume the Gibbs equation $D_t e = \theta D_t s - \sigma D_t \left( 1/\rho \right)$ when we study the entropy and free energy of the fluid on an evolving surface.\\
$(\rm{ii})$ If we assume that $\sigma = \sigma (\rho , e)$ and $e = e (\rho , \theta )$, then in general the system \eqref{eq11} is an overdetermined system for its initial value problem when the motion of $\Gamma (t)$ is given and we consider $(\rho, v , \theta , C)$ as unknown functions, while the system \eqref{eq111} is not an overdetermined system for its initial value problem when the motion of $\Gamma (t)$ is given. Because the second expression of the system \eqref{eq11} has six equations including the tangential and normal parts of the total velocity. For the same reason, the system \eqref{eq113} is an overdetermined system for its initial value problem when the motion of $\Gamma (t)$ is given, and the system \eqref{eq114} is not an overdetermined system for its initial value problem when the motion of $\Gamma (t)$ is given.\\
$(\rm{iii})$ When we consider $( \rho , \theta )$ as unknown functions, the system \eqref{eq115} is not an overdetermined system for its initial value problem when the motion of $\Gamma (t)$, surface flow $u$, and specific heat $C_\theta$ are prescribed. The system \eqref{eq116} is not an overdetermined system for its initial valued problem when the motion of $\Gamma (t)$ and surface flow $u$ are prescribed and we consider $C$ as an unknown function. It is easy to check that
\begin{equation*}
{\rm{div}}_\Gamma \{ \kappa {\rm{grad}}_\Gamma \theta \}= \kappa \Delta_\Gamma \theta \text{ and } {\rm{div}}_\Gamma \{ \nu {\rm{grad}}_\Gamma C \} = \nu \Delta_\Gamma C
\end{equation*}
if $\kappa, \nu $ are constants, where $\Delta_\Gamma$ is the Laplace-Beltrami operator.\\
$(\rm{iv})$ There exists at least three mathematical methods for deriving the pressure of the fluid on an evolving surface. This paper provides two of them. Koba-Liu-Giga \cite{KLG} gave the last one. The paper \cite{KLG} studied incompressible fluid flow on an evolving surface from  an energetic point of view. They applied the potential theory (the Helmholtz-Weyl decomposition of a function on an evolving surface) to derive the pressure of incompressible fluid on the evolving surface. While, this paper derives the pressure of compressible fluid on an evolving surface from the power density for the work done by the pressure or a chemical potential.
\end{remark}

Let us explain four key ideas of deriving the compressible fluid systems \eqref{eq11} and \eqref{eq111}. The first point is to focus our attention on the following energy densities for compressible fluid on an evolving surface:
\begin{assumption}[Energy densities for compressible fluid]\label{ass12}
\begin{multline*}
e_K = \frac{1}{2} \rho |v|^2,{ \ }e_D = \frac{1}{2}\{2 \mu |D_\Gamma (v)|^2 + \lambda |{\rm{div}}_\Gamma v |^2 \},{ \ } e_{W} = ( {\rm{div}}_\Gamma v ) \sigma + \rho F \cdot v,\\
e_A = \frac{1}{2} \rho |v|^2 + \rho e, { \ }e_{TD} = \frac{1}{2} \kappa | {\rm{grad}}_\Gamma \theta |^2, { \ }e_{SD} = \frac{1}{2} \nu |{\rm{grad}}_\Gamma C|^2 .
\end{multline*}
We call $e_K$ the \emph{kinetic energy}, $e_D$, the \emph{energy density for the energy dissipation due to the viscosities $\mu,\lambda$}, $e_W$ the \emph{power density for the work done by both the pressure $\sigma$ and exterior force $F$}, $e_A$ the \emph{total energy}, $e_{TD}$ the \emph{energy density for the energy dissipation due to thermal diffusion}, and $e_{SD}$ the \emph{energy density for the energy dissipation due to surface diffusion}, of compressible fluid on an evolving surface. Remark that $e_D \neq \tilde{e}_D$ in general, where $\tilde{e}_D$ is the density for the energy dissipation due to the viscosities $\mu, \lambda$. See subsection \ref{subsec51} and \cite{KS17} for the reason.
\end{assumption}
\noindent Combining Propositions \ref{prop311}, \ref{prop312} and \cite[Theorem 2.1 and Section 3]{KS17} gives a mathematical validity of our energy densities for compressible fluid on the evolving surface. Making use of these energy densities, we derive our compressible fluid systems on an evolving surface. Note that, from the system \eqref{eq11}, we have the following energy equality:
\begin{multline*}
\int_{\Gamma (t_2)} \left\{ \frac{1}{2} \rho |v|^2 \right\} (x,t_2) { \ }d \mathcal{H}_x^2 + \int_{t_1}^{t_2} \int_{\Gamma (\tau)} \{ 2 \mu | D_\Gamma (v) |^2 + \lambda |{\rm{div}}_\Gamma v |^2 \} (x,\tau ){ \ }d \mathcal{H}_x^2 d \tau\\
= \int_{\Gamma (t_1)} \left\{ \frac{1}{2} \rho |v|^2 \right\} (x,t_1) { \ }d \mathcal{H}_x^2 + \int_{t_1}^{t_2} \int_{\Gamma (\tau)} \{ ({\rm{div}}_\Gamma v) \sigma + \rho F \cdot v \} (x,\tau ) { \ } d\mathcal{H}^2_x d \tau.
\end{multline*}
Note also that we set the total energy $e_B$ of the barotropic compressible fluid on the evolving surface as follows:
\begin{equation*}
e_B = \frac{1}{2} \rho |v|^2 - p (\rho).
\end{equation*}
See Theorem \ref{thm19} for the barotropic compressible fluid systems on the evolving surface.

The second point is to apply an energetic variational approach. In order to derive our compressible fluid systems, this paper uses forces derived from variations of the action integral determined by the kinetic energy, dissipation energies determined by our energy densities, and work for compressible fluid on an evolving surface. For example, we obtain forces from the following variations: 
\begin{multline*}
\frac{d}{d \varepsilon} \int_0^T \int_{\Gamma^\varepsilon (t)} \frac{1}{2} \rho^\varepsilon |v^\varepsilon |^2 { \ } d \mathcal{H}^2_x d t, \frac{d}{d \varepsilon} \int_{\Gamma (t)} \frac{1}{2}(2 \mu | D_\Gamma (\tilde{v}^\varepsilon ) |^2 + \lambda |{\rm{div}}_\Gamma \tilde{v}^\varepsilon |^2) d \mathcal{H}_x^2,\\
\text{ and } \frac{d}{d \varepsilon} \int_{\Gamma (t)} \{ ({\rm{div}}_\Gamma \tilde{v}^\varepsilon) \sigma +  \rho F \cdot \tilde{v}^\varepsilon \} d \mathcal{H}_x^2,
\end{multline*}
where $\Gamma^\varepsilon (t)$, $\rho^\varepsilon$, $v^\varepsilon$, and $\tilde{v}^\varepsilon$ are variations. See Sections \ref{sect2}-\ref{sect4} for details.

The third point is to make use of the first and second laws of thermodynamics. To derive the dominant equations for the internal energy, entropy, and free energy of the fluid on an evolving surface, we assume the \emph{first law of thermodynamics}: for $\Omega (t) \subset \Gamma (t)$
\begin{equation*}
\frac{d}{d t} \int_{\Omega (t)} \{ \rho e \} (x,t) { \ }d \mathcal{H}_x^2 = \int_{\Omega (t)} \{ {\rm{div}}_\Gamma q_\theta+ \rho Q_\theta + \tilde{e}_D - ({\rm{div}}_\Gamma v) \sigma \} (x,t)  { \ } d \mathcal{H}_x^2,
\end{equation*}
and the Gibbs equation:
\begin{equation*}
D_t e = \theta D_t s - \sigma D_t \left( \frac{1}{\rho } \right) \text{ on } \mathcal{S}_T.
\end{equation*}

The fourth point is to apply a flow map on an evolving surface and the Riemannian metric induced by the flow map. Using the flow map and the Riemannian metric, we analyze the fluid on the evolving surface in this paper. Next subsection describes the flow maps in detail. See Section \ref{sect3} for the Riemannian metric induced by a flow map.

Finally we refer the reader to some references for basic notation and terminologies for fluid dynamics and thermodynamics in this paper. The basic notation and technical terms for physics in this paper are based on Serrin \cite{Ser59} and  Gurtin-Fried-Anand \cite{GFA10}. Serrin \cite{Ser59} studied mathematical derivations of fluid systems. They applied variational principles and thermodynamics to derive compressible and incompressible fluid systems in a domain. Gurtin-Fried-Anand \cite{GFA10} is a textbook for students and graduate students in physics and mathematics. The book \cite{GFA10} showed the fundamental rules of the mechanics and thermodynamics of continua. For fluid interfaces and interface models, we refer the reader to Gatignol-Prud'homme \cite{GP01} and Slattery-Sagis-Oh \cite{SSO07}.

The outline of this paper is as follows: In subsection \ref{subsec12}, we first introduce a flow map on an evolving surface, the velocity determined by the flow map, and variations of both the flow map and the evolving surface. Then we state the main results of this paper. In Section \ref{sect2} we prepare useful tools to analyze fluid flow on an evolving surface. We give the definitions of evolving surfaces and differential operators on an evolving surface, and we study basic properties of surface divergence, surface gradient, an orthogonal projection to a tangent space, and surface strain rate tensor. In Section \ref{sect3} we introduce the Riemannian metric induced by a flow map. Using the Riemannian metric, we investigate the representation of the kinetic energy, the dissipation energies, and the work for the fluid on an evolving surface to provide a mathematical validity of these energies. In Section \ref{sect4} we calculate variations of the action integral determined by the kinetic energy, the dissipation energies, and the work for the fluid on an evolving surface by using a flow map on the evolving surface and the Riemannian metric induced by the flow map. In Section \ref{sect5} we derive various compressible fluid systems on an evolving surface by applying our energetic variational approaches and the thermodynamic theory. Moreover, we derive the generalized heat and diffusion systems on an evolving surface.

\subsection{Main Results}\label{subsec12}

Let us first introduce a flow map on an evolving surface and the velocity determined by the flow map. Next we describe variations of both the flow map and the evolving surface. Then we state the main results of this paper.

Let $T \in (0 , \infty]$, and let $\Gamma (t) ( = \{ \Gamma (t) \}_{0 \leq t < T})$ be a smoothly evolving surface in $\mathbb{R}^3$ depending on time $t$. Assume that $\Gamma (t)$ is a closed Riemannian $2$-dimensional manifold for each $t \in [0, T)$. 

We say that $\Omega (t) \subset \Gamma ( t )$ is \emph{flowed by the velocity field} $V = V (x,t) = $ \\$ { }^t ( V_1 (x,t) , V_2 (x,t) , V_3 (x,t) )$ if there exists a smooth function $\hat{x} = \hat{x}( \xi , t ) $ \\$= { }^t (\hat{x}_1 ( \xi , t ), \hat{x}_2 ( \xi , t ), \hat{x}_3 ( \xi , t ) ) $ such that for every $\xi \in \Gamma (0)$,
\begin{equation*}
\begin{cases}
\frac{d \hat{x}}{d t} (\xi , t) = V ( \hat{x} (\xi , t), t ), { \ \ \ }t \in (0,T),\\
\hat{x} (\xi , 0) = \xi ,
\end{cases}
\end{equation*}
and $\Omega ( t )$ is expressed by
\begin{equation*}
\Omega ( t )= \{ x= { }^t (x_1,x_2,x_3) \in \mathbb{R}^3;{ \ }x = \hat{x}( \xi , t ), { \ }\xi \in \Omega_0, { \ } \Omega_0 \subset \Gamma (0) \}.
\end{equation*}
The mapping $\xi \mapsto \hat{x} ( \xi , t )$ is called a \emph{flow map} on $\Gamma (t)$, the mapping $t \mapsto \hat{x} ( \xi , t )$ is called an \emph{orbit} starting from $\xi$, and $V = V (x,t)$ is called the \emph{velocity determined by the flow map} $\hat{x} ( \xi , t)$. For simplicity we call $\hat{x} ( \xi , t)$ a flow map. We assume that $\hat{x} (\cdot , t) : \Gamma (0) \to \Gamma (t)$ is bijective for each $0 < t <T$ and that $V$ is the total velocity $v$.

Next we introduce a variation $\hat{x}^\varepsilon (\xi,t)$ of a flow map $\hat{x} ( \xi , t )$ and the velocity $v^\varepsilon$ determined by the flow map $\hat{x}^\varepsilon$. Let $\hat{x} ( \xi , t )$ be a flow map on $\Gamma (t)$, and let $v$ be the velocity determined by the flow map $\hat{x} ( \xi , t )$ on $\Gamma ( t )$, i.e. for $\xi \in \Gamma (0)$ and $0<t<T$,
\begin{equation*}
\begin{cases}
v = v (x,t)= { }^t ( v_1 (x,t) , v_2 ( x,t) , v_3 (x,t) ),\\ 
\hat{x} = \hat{x} ( \xi , t ) = { }^t ( \hat{x}_1 ( \xi , t ) , \hat{x}_2 ( \xi , t ) , \hat{x}_3 ( \xi , t )),\\
\frac{d \hat{x}}{d t} (\xi , t) = v (\hat{x} ( \xi , t ) , t ),\\
\hat{x} (\xi , 0) = \xi .
\end{cases}
\end{equation*}
Write
\begin{align*}
\Gamma (t) & := \{ x = { }^t (x_1 , x_2 , x_3 ) \in \mathbb{R}^3; { \ } x = \hat{x}( \xi , t ), { \ }\xi  \in \Gamma (0) \} ,\\
\mathcal{S}_T & := \bigg\{ (x,t) \in \mathbb{R}^4;{ \ } (x,t ) \in \bigcup_{0<t <T} \{ \Gamma (t) \times \{ t \} \} \bigg\}.
\end{align*}
For $-1 < \varepsilon <1$ let $\Gamma^\varepsilon (t) ( = \{ \Gamma^\varepsilon (t) \}_{0 \leq t < T})$ be a smoothly evolving $2$-dimensional surface in $\mathbb{R}^3$ depending on time $t$. We say that $\Gamma^\varepsilon (t)$ is a variation of $\Gamma (t)$ if $\Gamma^\varepsilon (0) = \Gamma (0)$ and $\Gamma^{\varepsilon}(t)|_{\varepsilon = 0} = \Gamma (t) $. Set
\begin{equation*}
\mathcal{S}_T^\varepsilon := \bigg\{ (x,t) \in \mathbb{R}^4;{ \ } (x,t ) \in \bigcup_{0<t <T} \{ \Gamma^\varepsilon (t) \times \{ t \} \} \bigg\}.
\end{equation*}
Let $\hat{x}^\varepsilon (\xi,t)$ be a flow map on $\Gamma^\varepsilon (t)$, and $v^\varepsilon$ be the velocity determined by the flow map $\hat{x}^\varepsilon$, i.e. for $\xi \in \Gamma (0)$ and $0<t<T$,
\begin{equation*}
\begin{cases}
v^\varepsilon = v^\varepsilon (x,t)= { }^t ( v^\varepsilon_1 (x,t) , v^\varepsilon_2 ( x,t) , v^\varepsilon_3 (x,t) ),\\ 
\hat{x}^\varepsilon = \hat{x}^\varepsilon ( \xi ,t ) = { }^t ( \hat{x}^\varepsilon_1 (\xi,t ) , \hat{x}^\varepsilon_2 (\xi,t ) , \hat{x}^\varepsilon_3 (\xi,t ) ),\\
\frac{d \hat{x}^\varepsilon}{d t} (\xi , t) = v^\varepsilon (\hat{x}^\varepsilon(\xi , t ) , t ),\\
\hat{x}^\varepsilon (\xi , 0) = \xi .
\end{cases}
\end{equation*}
We say that $( \hat{x}^\varepsilon (\xi ,t) , \mathcal{S}_T^\varepsilon)$ is a variation of $(\hat{x} ( \xi , t ) , \mathcal{S}_T)$ if $\hat{x}^\varepsilon (\xi , t)$ is smooth as a function of $( \varepsilon , \xi , t ) \in (-1 ,1) \times \Gamma (0) \times [0,T)$ and $\hat{x}^\varepsilon ( \xi ,t ) |_{ \varepsilon = 0} = \hat{x} ( \xi , t )$. Assume that $\Gamma^\varepsilon (t)$ is expressed by
\begin{equation*}
\Gamma^\varepsilon (t) = \{ x = { }^t (x_1 , x_2 , x_3 ) \in \mathbb{R}^3; { \ } x = \hat{x}^\varepsilon (\xi , t ) , { \ }\xi  \in \Gamma (0) \} .
\end{equation*}
For smooth functions $f = f (x,t)$, $D_t f := \partial_t f + ( v , \nabla ) f $ and $D_t^\varepsilon f := \partial_t f + (v^\varepsilon , \nabla  )f $.

To consider compressible fluid flow on an evolving surface, we first study the density of the fluid on the evolving surface. To this end, we need the continuity equation of the fluid on the evolving surface. We also need the continuity equation of the fluid on a variation of the evolving surface. For $- 1 < \varepsilon <1$, assume that $\rho = \rho (x,t)$ and $\rho^\varepsilon= \rho^\varepsilon (x,t)$ are smooth functions.
\begin{proposition}[Continuity equation of fluid on evolving surfaces]\label{prop13}{ \ }\\
\noindent $(\mathrm{i})$ For each $0<t<T$ and $\Omega (t) \subset \Gamma (t)$ flowed by the velocity $v$, assume that
\begin{equation*}
\frac{d }{d t}\int_{\Omega (t)} \rho (x,t) { \ }d \mathcal{H}_{x}^2 = 0.
\end{equation*}
Then $\rho$ satisfies
\begin{equation*}
D_t \rho + ({\rm{div}}_\Gamma v) \rho = 0 \text{ on }\mathcal{S}_T.
\end{equation*}
\noindent $(\mathrm{ii})$ For each $0<t<T$ and $\Omega^\varepsilon (t) \subset \Gamma^\varepsilon (t)$ flowed by the velocity $v^\varepsilon$, assume that
\begin{equation*}
\frac{d }{d t}\int_{\Omega^\varepsilon (t)} \rho^\varepsilon (x,t) { \ }d \mathcal{H}_{x}^2 = 0.
\end{equation*}
Then $\rho^\varepsilon$ satisfies
\begin{equation*}
D_t^\varepsilon \rho^\varepsilon + ({\rm{div}}_{\Gamma^\varepsilon} v^\varepsilon) \rho^\varepsilon = 0 \text{ on } \mathcal{S}^\varepsilon_T. 
\end{equation*}
\end{proposition}
\noindent We often call Proposition \ref{prop13} the \emph{surface transport theorem}. The proof of Proposition \ref{prop13} can be found in \cite{Bet86}, \cite{GSW89}, \cite{DE07} and \cite{KLG} (see also \cite{SSO07}). For our purposes, we give the proof of Proposition \ref{prop13} in Section \ref{sect3}. See Section \ref{sect2} for the definitions of surface divergences ${\rm{div}}_\Gamma$ and ${\rm{div}}_{ \Gamma^\varepsilon}$.

Next we study variations of the flow map to the action integral determined by the kinetic energy. Let $\rho_0 = \rho_0 (x)$ be a smooth function. We call $\rho$ the density of the fluid on $\Gamma (t)$ if $\rho$ satisfies
\begin{equation*}
\begin{cases}
D_t \rho + ({\rm{div}}_\Gamma v) \rho =0 & \text{ on } \mathcal{S}_T,\\
\rho |_{t = 0} = \rho_0 & \text{ on }\Gamma (0) .
\end{cases}
\end{equation*}
We call $\rho^\varepsilon$ the density of the fluid on $\Gamma^\varepsilon (t)$ if $\rho^\varepsilon$ satisfies
\begin{equation*}
\begin{cases}
D_t^\varepsilon \rho^\varepsilon + ({\rm{div}}_{\Gamma^\varepsilon} v^\varepsilon ) \rho^\varepsilon =0 & \text{ on } \mathcal{S}^\varepsilon_T,\\
\rho^\varepsilon |_{t = 0} = \rho_0 & \text{ on } \Gamma (0).
\end{cases}
\end{equation*}
For each variation $\hat{x}^\varepsilon$ we define the \emph{action integral} as
\begin{equation*}
A [ \hat{x}^\varepsilon ] = - \int_0^T \int_{\Gamma^\varepsilon (t)} \frac{1}{2} \rho^\varepsilon (x,t) | v^\varepsilon (x ,t)|^2 { \ } d \mathcal{H}^2_x d t,
\end{equation*}
where $\rho^\varepsilon$ is the density of the fluid on $\Gamma^\varepsilon (t)$ and $v^\varepsilon$ is the velocity determined by the flow map $\hat{x}^\varepsilon = \hat{x}^\varepsilon (\xi ,t)$. Note that $\hat{x}^\varepsilon (\xi ,t)$ is a flow map on $\Gamma^\varepsilon (t)$.

We now assume that there are $\hat{y} \in [ C^\infty ( \mathbb{R}^3 \times [0,T))]^3$ and $z \in [C^\infty (\mathcal{S}_T)]^3$ such that for $\xi \in \Gamma (0)$ and $0 \leq t < T$,
\begin{align*} 
  \hat{x}^\varepsilon ( \xi ,t) \bigg|_{\varepsilon = 0} &= \hat{x} ( \xi , t ),\\
 v^\varepsilon (\hat{x}^\varepsilon ( \xi ,t ) ,t )\bigg|_{\varepsilon = 0} & = v (\hat{x} ( \xi , t ),t),\\
 \frac{d}{d \varepsilon} \bigg|_{\varepsilon = 0} \hat{x}^\varepsilon ( \xi ,t) & = \hat{y} ( \xi , t ),\\
 z (\hat{x} ( \xi , t ) , t) & = \hat{y} ( \xi, t ).
\end{align*}
Here $z$ is the variation vector field. See Section \ref{sect2} for function spaces on an evolving surface. We also assume that for every $\xi \in \Gamma (0)$ and $0 \leq t <T$,
\begin{equation*}
\rho^\varepsilon (\hat{x}^\varepsilon ( \xi ,t ) ,t )\bigg|_{\varepsilon = 0} = \rho (\hat{x} ( \xi , t ),t).
\end{equation*}

\begin{theorem}[Variation of the flow map to the action integral]\label{thm14}{ \ }\\
Assume that $(\hat{x}^\varepsilon (\xi ,t ), \mathcal{S}_T^\varepsilon)$ is a variation of $(\hat{x} ( \xi , t ), \mathcal{S}_T)$ with $\Gamma^\varepsilon (0) = \Gamma (0)$. Suppose that $\rho$ is the density of the fluid on $\mathcal{S}_T$ and that $\rho^\varepsilon$ is the density of the fluid on $\mathcal{S}_T^\varepsilon$. Then
\begin{equation*}
\frac{d}{d \varepsilon}\bigg|_{\varepsilon = 0}A [\hat{x}^\varepsilon ] = \int_0^T \int_{\Gamma (t)} \{ \rho D_t v \} (x,t) \cdot z ( x ,  t ) { \ } d \mathcal{H}^2_x d t,
\end{equation*}
where $D_t v = \partial_t v + ( v , \nabla ) v$. Moreover, assume in addition that $z \cdot n =0$. Then
\begin{equation*}
\frac{d}{d \varepsilon}\bigg|_{\varepsilon = 0}A [\hat{x}^\varepsilon ] = \int_0^T \int_{\Gamma (t)} \{ P_\Gamma \rho D_t v \} (x,t) \cdot z ( x ,  t ) { \ } d \mathcal{H}^2_x d t.
\end{equation*}
\end{theorem}
Remark: In the former part of Theorem \ref{thm14} we consider variations with respect to the flow map, including the motion of $\Gamma (t)$, while we consider variations with respect to the tangential part of the flow map on $\Gamma (t)$ in the latter part.

We now study variations of the dissipation energy $E_D [ V]$ and the work $E_W [V]$ for the velocity field $V = { }^t ( V_1 (x,t) , V_2 (x,t) , V_3 (x,t) )$ at each fixed time $t$. Let $\mu = \mu (x,t)$, $\lambda = \lambda (x,t)$, $\sigma = \sigma (x,t)$, and $F = F (t,x) ={ }^t (F_1 , F_2, F_3)$ be smooth functions. For each smooth function $V = V (x , t ) = { }^t ( V_1, V_2, V_3)$, let
\begin{align*}
E_D [V] (t) := & - \int_{\Gamma (t)} \bigg\{  \frac{1}{2} \left( 2 \mu | D_\Gamma (V) |^2  + \lambda |{\rm{div}}_\Gamma V|^2 \right)  \bigg\} (x,t) { \ }d \mathcal{H}^2_x,\\
E_{W_1} [V] (t) := & \int_{\Gamma (t)} \{ ({\rm{div}}_\Gamma V ) \sigma \} (x,t) { \ }d \mathcal{H}^2_x,\\
E_{W_2} [V]  (t):= & \int_{\Gamma (t)} \{ \rho F \cdot V \} (x,t) { \ }d \mathcal{H}^2_x.
\end{align*}
Here $D_\Gamma (V) = P_\Gamma D (V) P_\Gamma$, where $P_\Gamma$ is the orthogonal projection to a tangent space defined by \eqref{eq22} in Section \ref{sect2} and $D (V) = \{ (\nabla V) + { }^t (\nabla V) \} /2$. We often call $E_D$, $E_{W_1}$, and $E_{W_2}$ the dissipation energy determined by the energy density $e_D$, the work done by the pressure $\sigma$, and the work done by the exterior force $F$, of the fluid on an evolving surface, respectively. Moreover, we set
\begin{equation*}
E_{D + W}[V] = E_{D}[V] + E_{W_1} [V] + E_{W_2}[V].
\end{equation*}
We shall study their variations.
\begin{theorem}[Variation of the velocity to dissipation energy/work]\label{thm15}{ \ }\\
Fix $t \in (0,T)$. Then for every vector field $\varphi \in [ C_0^\infty (\Gamma (t))]^3$,
\begin{multline*}
\frac{d}{d \varepsilon} \bigg|_{\varepsilon = 0} E_{D + W} [v + \varepsilon \varphi ] (t) \\
= \int_{\Gamma (t)} \bigg\{ {\rm{div}}_\Gamma \bigg( 2 \mu D_\Gamma (v) + \lambda P_\Gamma ({\rm{div}}_\Gamma v) - P_\Gamma \sigma \bigg) + \rho F \bigg\}(x,t) \cdot \varphi (x) { \ }d \mathcal{H}^2_x.
\end{multline*}
Moreover, if $\varphi \cdot n =0$, then
\begin{multline*}
\frac{d}{d \varepsilon} \bigg|_{\varepsilon = 0} E_{D + W} [v + \varepsilon \varphi ] (t) \\
= \int_{\Gamma (t)} \bigg\{  P_\Gamma {\rm{div}}_\Gamma \bigg( 2 \mu D_\Gamma (v) + \lambda P_\Gamma ({\rm{div}}_\Gamma v) - P_\Gamma \sigma \bigg) + P_\Gamma \rho F \bigg\} \cdot \varphi (x) { \ }d \mathcal{H}^2_x.
\end{multline*}
\end{theorem}
\noindent Remark: In the former part of Theorem \ref{thm15} we consider variations with respect to the total velocity, including the motion of $\Gamma (t)$, while we consider variations with respect to the tangential part of the total velocity on $\Gamma (t)$ in the latter part.

Combining Theorem \ref{thm15}, Proposition \ref{prop312}, and \cite[Theorem 2.1 and Section 3]{KS17} gives a mathematical validity of the Boussinesq-Scriven law. Indeed, we obtain the surface stress tensor determined by the Boussinesq-Scriven law from Theorem \ref{thm15} when $F \equiv 0$.

In the next step, we study Fourier's and Fick's laws of surface diffusion. Let $\kappa = \kappa (x,t)$, $\nu = \nu (x,t)$, $\theta = \theta (x,t)$, and $C = C (x,t)$ be four smooth functions. Fix $t$. For each smooth function $f = f (x,t)$, let
\begin{align*}
E_{TD} [f](t) := & - \int_{\Gamma (t)} \frac{1}{2} \kappa (x,t) |{\rm{grad}}_\Gamma f (x,t)|^2  { \ }d \mathcal{H}^2_x,\\
E_{SD} [f](t) := & - \int_{\Gamma (t)} \frac{1}{2} \nu (x,t) |{\rm{grad}}_\Gamma f (x,t)|^2  { \ }d \mathcal{H}^2_x .
\end{align*}
We often call $E_{TD}$ and $E_{SD}$ the dissipation energies determined by $e_{TD}$ and $e_{SD}$, respectively. See Proposition \ref{prop312} for the representation of these energies.

\begin{theorem}[Fourier' laws and Fick's laws of surface diffusion]\label{thm16}{ \ }\\
$(\mathrm{i})$ Fix $t \in (0,T)$. Then for every $\varphi \in C_0^\infty (\Gamma (t))$,
\begin{align*}
\frac{d}{d \varepsilon} \bigg|_{\varepsilon = 0} E_{TD} [\theta + \varepsilon \varphi ](t) = \int_{\Gamma (t)} {\rm{div}}_\Gamma \bigg( \kappa (x,t) {\rm{grad}}_\Gamma \theta (x,t)  \bigg) \varphi (x) { \ }d \mathcal{H}^2_x,\\
\frac{d}{d \varepsilon} \bigg|_{\varepsilon = 0} E_{SD} [C + \varepsilon \varphi ](t) = \int_{\Gamma (t)} {\rm{div}}_\Gamma \bigg( \nu (x,t) {\rm{grad}}_\Gamma C (x,t)  \bigg) \varphi (x) { \ }d \mathcal{H}^2_x.
\end{align*}
$(\mathrm{ii})$ Fix $t\in (0,T)$ and $x \in \Gamma (t)$. For each $\vartheta_1 , \vartheta_2 , \vartheta_3 \in \mathbb{R}$,
\begin{align*}
\mathcal{E}_{TD} [ \vartheta_1 , \vartheta_2 , \vartheta_3 ] : = - \frac{\kappa (x,t)}{2} (\vartheta_1^2 + \vartheta_2^2 + \vartheta_3^2 ),\\
\mathcal{E}_{SD} [ \vartheta_1 , \vartheta_2 , \vartheta_3 ] : =  -\frac{\nu (x,t)}{2} (\vartheta_1^2 + \vartheta_2^2 + \vartheta_3^2 ).
\end{align*}
Then
\begin{align*}{ }^t
\begin{pmatrix}
\frac{\partial \mathcal{E}_{TD} }{\partial \vartheta_1}, \frac{\partial \mathcal{E}_{TD}}{\partial \vartheta_2}, \frac{\partial \mathcal{E}_{TD}}{\partial \vartheta_3}
\end{pmatrix}\Bigg|_{(\vartheta_1, \vartheta_2 , \vartheta_3 ) = (\partial_1^\Gamma \theta , \partial_2^\Gamma \theta , \partial_3^\Gamma \theta)} = - \kappa {\rm{grad}}_\Gamma \theta,\\
{ }^t
\begin{pmatrix}
\frac{\partial \mathcal{E}_{SD} }{\partial \vartheta_1}, \frac{\partial \mathcal{E}_{SD}}{\partial \vartheta_2}, \frac{\partial \mathcal{E}_{SD}}{\partial \vartheta_3}
\end{pmatrix}\Bigg|_{(\vartheta_1, \vartheta_2 , \vartheta_3 ) = (\partial_1^\Gamma C , \partial_2^\Gamma C , \partial_3^\Gamma C)} = - \nu {\rm{grad}}_\Gamma C .
\end{align*}
\end{theorem}

Next we give the generalized surface flux on the evolving surface $\Gamma (t)$.
\begin{theorem}[Variation of dissipation energy/Fluxes on surfaces]\label{thm17} { \ }\\
Let $e_{\mathcal{J}} \in C^1 ([0, \infty))$ or $e_{\mathcal{J}} \in C^1 ((0, \infty))$. Suppose that $e_{\mathcal{J}}$ is a non-negative function. Then the following two assertions hold:\\
\noindent $(\mathrm{i})$ Fix $t$. For each smooth function $f  = f (x,t)$,
\begin{equation*}
E_{GD}[f](t)  := - \int_{\Gamma (t)} \frac{1}{2} e_{\mathcal{J}}( | {\rm{grad}}_\Gamma f |^2) { \ }d \mathcal{H}^2_x.
\end{equation*}
Then for every $f \in C^{2,0} ( \mathcal{S}_T)$ with $| {\rm{grad}}_\Gamma f| \neq 0$ and $\varphi \in C_0^\infty (\Gamma (t))$,
\begin{equation*}
\frac{d}{d \varepsilon} \bigg|_{\varepsilon = 0} E_{GD}[f + \varepsilon \varphi ] (t)= \int_{\Gamma (t)} \{ {\rm{div}}_\Gamma ( e_{\mathcal{J}}' ( | {\rm{grad}}_\Gamma f |^2) {\rm{grad}}_\Gamma f ) \} (x,t) \varphi (x) { \ }d \mathcal{H}^2_x.
\end{equation*}
$(\mathrm{ii})$ Let $f = f (x,t)$ be a smooth function with $| {\rm{grad}}_\Gamma f | \neq 0$. Fix $0 < t <T$ and $x \in \Gamma (t)$. For each $\vartheta_1 , \vartheta_2 , \vartheta_3 \in \mathbb{R}$,
\begin{equation*}
\mathcal{E}_{GD} := - \frac{1}{2} e_{\mathcal{J}}( \vartheta_1^2 + \vartheta_2^2 + \vartheta_3^2) .
\end{equation*}
Then
\begin{equation*}{ }^t
\begin{pmatrix}
\frac{\partial \mathcal{E}_{GD}}{\partial \vartheta_1 } ,\frac{\partial \mathcal{E}_{GD}}{\partial \vartheta_2 }, \frac{\partial \mathcal{E}_{GD}}{\partial \vartheta_3 }
\end{pmatrix}\bigg|_{ ( \vartheta_1 = \partial_1^\Gamma f, \vartheta_2 = \partial_2^\Gamma f, \vartheta_3 = \partial_3^\Gamma f )} = - e_{\mathcal{J}}' (| {\rm{grad}}_\Gamma f |^2) {\rm{grad}}_\Gamma f.
\end{equation*}
\end{theorem}
Proposition \ref{prop313} and \cite[Theorem 2.1]{KS17} give a mathematical validity of the representation of the energy density $e_{\mathcal{J}} (|{\rm{grad}}_\Gamma f |^2)$ for the energy dissipation due to general diffusion.

Applying Proposition \ref{prop13}, Theorems \ref{thm14}-\ref{thm17}, an energetic variational approach (Least Action Principle and Maximum/Minimum Dissipation Principle), the first law of thermodynamics, we can derive our compressible fluid systems \eqref{eq11} and \eqref{eq111}. Under some conditions, we can obtain the enthalpy \eqref{eq18},  the entropy \eqref{eq19}, and the free energy \eqref{eq110}. We easily check that the system \eqref{eq11} satisfies the conservative form \eqref{eq12}. See Section \ref{sect5} for more detailed derivation of our compressible fluid systems.

Let us investigate conservation laws of the compressible fluid system \eqref{eq11}.
\begin{theorem}[Conservation laws]\label{thm18}{ \ }\\
Assume that $\Gamma (t)$ is flowed by the total velocity $v$. Then the system \eqref{eq11} satisfies that for $0 < t_1 < t_2 <T$,
\begin{align}
& \int_{\Gamma (t_2)} \rho (x,t_2) { \ }d \mathcal{H}^2_x = \int_{\Gamma (t_1)} \rho (x,t_1) { \ }d \mathcal{H}^2_x, \label{eq117}\\
& \int_{\Gamma (t_2)} \rho v { \ }d \mathcal{H}^2_x = \int_{\Gamma (t_1)} \rho v { \ }d \mathcal{H}^2_x + \int_{t_1}^{t_2} \int_{\Gamma (\tau)} \rho F { \ }d \mathcal{H}^2_x d \tau,\label{eq118}\\
& \int_{\Gamma (t_2)} e_A{ \ }d \mathcal{H}^2_x = \int_{\Gamma (t_1)} e_A { \ }d \mathcal{H}^2_x + \int_{t_1}^{t_2} \int_{\Gamma (\tau)} \{ \rho F \cdot v + \rho Q_\theta \}  { \ }d \mathcal{H}^2_x d \tau, \label{eq119}\\
& \int_{\Gamma (t_2)} C (x,t_2) { \ }d \mathcal{H}^2_x = \int_{\Gamma (t_1)} C (x,t_1) { \ }d \mathcal{H}^2_x + \int_{t_1}^{t_2} \int_{\Gamma (\tau)} Q_C  { \ }d \mathcal{H}^2_x d \tau , \label{eq120}
\end{align}
and
\begin{multline}\label{eq121}
\int_{\Gamma (t_2)} \{ x \times ( \rho v ) \} ( x , t ) { \ }d \mathcal{H}^2_x \\
= \int_{\Gamma (t_1)} \{ x \times ( \rho v) \} ( x , t) { \ }d \mathcal{H}^2_x + \int_{t_1}^{t_2} \int_{\Gamma (\tau)} \{ x \times (\rho F)\} (x , \tau )  { \ }d \mathcal{H}^2_x d \tau.
\end{multline}
\end{theorem}

Finally, we state barotropic compressible fluid systems. Using a chemical potential, we derive the pressure of barotropic compressible fluid on an evolving surface.
\begin{theorem}[Barotropic compressible fluid]\label{thm19}
Let $p \in C^1 ( ( 0, \infty ) )$. Under the hypotheses of Theorem \ref{thm14}, for each variation $\hat{x}^\varepsilon$,
\begin{equation*}
A_{B} [\hat{x}^\varepsilon ] := - \int_0^T \int_{\Gamma^\varepsilon (t)} \left\{ \frac{1}{2} \rho^\varepsilon (x,t) |v^\varepsilon (x,t)|^2 - p (\rho^\varepsilon (x,t) ) \right\} { \ }d \mathcal{H}_2^x d t. 
\end{equation*}
Then
\begin{equation*}
\frac{d}{d \varepsilon } \bigg|_{\varepsilon = 0} A_{B} [\hat{x}^\varepsilon ] = \int_0^T \int_{\Gamma (t)} \left\{ \rho D_t v + {\rm{grad}}_\Gamma \mathfrak{p} + \mathfrak{p} H_\Gamma n \right\} (x,t) \cdot z (x,t) { \ }d \mathcal{H}_2^x d t,
\end{equation*}
where $\mathfrak{p} = \mathfrak{p} (\rho ) = \rho p'(\rho) - p (\rho)$. Moreover, the two assertions hold:\\
\noindent $(\mathrm{i})$ For every $z \in [C_0^\infty (\mathcal{S}_T)]^3$, assume that 
\begin{equation*}
\int_0^T \int_{\Gamma (t)} \left\{ \rho D_t v + {\rm{grad}}_\Gamma \mathfrak{p} + \mathfrak{p} H_\Gamma n \right\} (x,t) \cdot z (x,t) { \ }d \mathcal{H}_x^2 d t= 0 . 
\end{equation*}
Then $( \rho , v , \mathfrak{p} )$ fulfill
\begin{equation*}
\rho D_t v + {\rm{grad}}_\Gamma \mathfrak{p} + \mathfrak{p} H_\Gamma n = 0 { \ }\text{ on } \mathcal{S}_T .
\end{equation*}
\noindent $(\mathrm{ii})$ For every $z \in [C_0^\infty (\mathcal{S}_T)]^3$ satisfying $z \cdot n =0$, assume that 
\begin{equation*}
\int_0^T \int_{\Gamma (t)} \left\{ \rho D_t v + {\rm{grad}}_\Gamma \mathfrak{p} + \mathfrak{p} H_\Gamma n \right\} (x,t) \cdot z (x,t){ \ }d \mathcal{H}_x^2 d t = 0. 
\end{equation*}
Then $( \rho , v , \mathfrak{p} )$ fulfill
\begin{equation*}
P_\Gamma \rho D_t v + {\rm{grad}}_\Gamma \mathfrak{p} = 0 { \ }\text{ on } \mathcal{S}_T. 
\end{equation*}
\end{theorem}
Remark: We prove Theorems \ref{thm14} and \ref{thm19} in Section \ref{sect3}, Theorems \ref{thm15}-\ref{thm17} in Section \ref{sect4}, and Theorem \ref{thm18} in subsection \ref{subsec53}.

Let us state three difficulties in the derivation of our compressible fluid systems on an evolving surface and the ideas to overcome these difficulties. The first difficultly is to drive the pressure of compressible fluid on an evolving surface. In order to derive the pressure term of our compressible fluid systems, we focus our attention on the power density for the work done by the pressure of the fluid and a chemical potential. The second difficulty is to derive viscous term of our compressible fluid systems on an evolving surface. To overcome the difficult point, we calculate variations of both the energy dissipation due to the viscosities and the work done by both the pressure and exterior force. The third difficulty is to investigate the internal energy, enthalpy, entropy, and free energy of the fluid on an evolving surface. To solve the problem, we introduce the notation $D_t^N$ and apply the first and second laws of thermodynamics, where $D_t^N f = \partial_t f + ( v \cdot n ) (n , \nabla ) f $.

Let us explain two essential strategies to analyze fluid flow on an evolving surface. The first one is to apply both a flow map on the evolving surface and the Riemannian metric induced by the flow map. By using them, we deal with functions on the evolving surface and give a mathematical validity of several energies of the fluid on the evolving surface. The second one is to use both an energetic variational approach and the first law of thermodynamics. Combining an energetic variational principle and the thermodynamic theory, we derive our compressible fluid systems on the evolving surface. The energetic variational method of this paper improves the one from Koba-Liu-Giga \cite{KLG}. The paper \cite{KLG} improved the energetic variational approach, which had been studied by Strutt \cite{Str73} and Onsager (\cite{Ons31a}, \cite{Ons31b}), to derive incompressible fluid systems on an evolving surface. Koba-Sato \cite{KS17} applied their energetic variational approaches to derive their non-Newtonian fluid systems in domains.

Let us state the history of the surface stress tensor determined by the \\Boussinesq-Scriven law. Boussinesq \cite{Bou13} first considered the existence of surface fluid. Scriven \cite{Scr60} introduced the surface stress tensor to apply it to arbitrary surfaces. Slattery \cite{Sla64} studied some properties of the surface stress tensor determined by the Boussinesq-Scriven law. After that many researchers have made models of two-phase flow system with interfacial phenomena such as surface tension, surface flow, and phase transition. See Gatignol-Prud'homme \cite{GP01}, Slattery-Sagis-Oh \cite{SSO07}, and the references given there. Bothe and Pr\"{u}ss \cite{BP12} used the Boussinesq-Scriven law to make a two-phase flow system with surface viscosity and surface tension. This paper gives a mathematical validity of the Boussinesq-Scriven law (see Theorem \ref{thm15}, Section \ref{sect3}, and \cite[Theorem 2.1 and Section 3]{KS17}). 

We next explain some mathematical derivations of incompressible fluid systems on a manifold. Arnol'd \cite{Arn66}, \cite{Arn97} applied the Lie group of diffeomorphisms to derive an inviscid incompressible fluid system on a manifold. See also Ebin-Marsden \cite{EM70}. Taylor \cite{Tay92} introduced a viscous incompressible fluid system on a manifold from their physical sense. Mitsumatsu and Yano \cite{MY02} used their energetic approach to derive a viscous incompressible fluid system on a manifold. Arnaudon and Cruzeiro \cite{AC12} applied stochastic variational approach to derive a viscous incompressible fluid system on a manifold. Remark that Taylor \cite{Tay92}, Mitsumatsu-Yano \cite{MY02}, and Arnaudon-Cruzeiro \cite{AC12} used Taylor's strain rate $\{ (\nabla_\mathcal{M} u ) + { }^t (\nabla_{\mathcal{M}} u) \}/2$ with surface divergence-free to derive their systems, where $\nabla_{\mathcal{M}}$ is the covariant derivative.

Finally we state mathematical derivations of fluid systems on an evolving surface. Dziuk and Elliott \cite{DE07} applied the transport theorem (Leibniz formula) on an evolving surface and their surface flux to make several fluid systems on the evolving surface. Koba-Liu-Giga \cite{KLG} derived incompressible fluid systems on an evolving surface by their energetic variational approach. They applied the potential theory (the Helmholtz-Weyl decomposition) to derive the pressure of incompressible fluid flow on an evolving surface.

\section{Preliminaries}\label{sect2}

Let us prepare useful tools to analyze fluid flow on an evolving surface. We first describe evolving surfaces, and then we introduce function spaces and notation such as ${\rm{div}}_\Gamma$, ${\rm{grad}}_\Gamma$, $P_\Gamma$, and $H_\Gamma$. Finally we investigate fundamental properties of surface gradient, surface divergence, an orthogonal projection to a tangent space, surface strain rate tensors, surface stress tensor, material derivatives, and integration by parts on evolving surfaces.

\subsection{Evolving Surfaces}\label{subsec21}
Let us recall evolving surfaces. 
\begin{definition}[$2$-dimensional $C^2$-surfaces in $\mathbb{R}^3$]\label{def21}{ \ }
A set $\Gamma_0$ in $\mathbb{R}^3$ is called a $C^2$-surface in $\mathbb{R}^3$ if for each point $x_0 \in \Gamma_0$ there are $r >0$ and $\phi \in C^2 (B_{r}(x_0))$ such that
\begin{equation*}
\Gamma_0 \cap B_r (x_0) = \{ x = { }^t (x_1,x_2,x_3) \in B_{r}(x_0);{ \ }\phi (x) = 0 \}
\end{equation*}
and that 
\begin{equation*}
\nabla_x \phi = { }^t \left( \frac{\partial \phi}{\partial x_1},  \frac{\partial \phi}{\partial x_2}, \frac{\partial \phi}{\partial x_3} \right) \neq (0,0,0) { \ }\text{ on }B_r (x_0).
\end{equation*}
Here $B_r (x_0) := \{ x \in \mathbb{R}^3; { \ }|x - x_0| < r \}$. In this paper we call a $2$-dimensional $C^2$-surface in $\mathbb{R}^3$ \emph{a $2$-dimensional surface in $\mathbb{R}^3$}.
\end{definition}
\begin{definition}[Evolving $2$-dimensional $C^{2,1}$-surfaces in $\mathbb{R}^3$]\label{def22}{ \ }\\
Let $\Gamma (t) \{ = \{ \Gamma (t) \}_{0 \leq t < T} \}$ be a set in $\mathbb{R}^3$ depending on time $t \in [0,T)$ for some $T \in (0 , \infty]$. A family $\{ \Gamma (t) \}_{0 \leq t < T}$ is called an evolving $2$-dimensional $C^{2,1}$-surface in $\mathbb{R}^3$ on $[0,T)$ if the following two properties hold:
\begin{enumerate}\setlength{\leftskip}{2mm}
\renewcommand{\labelenumi}{(\roman{enumi})}
\item $\Gamma (0)$ is a $2$-dimensional surface in $\mathbb{R}^3$.
\item For each $t_0 \in (0,T)$ and $x_0 \in \Gamma (t_0)$, there are $r_1 ,r_2 >0$ and\\ $\psi \in C^{2,1} ( B_{r_1}(x_0) \times B_{r_2} (t_0))$ such that
\begin{equation*}
\Gamma (t_0) \cap B_{r_1} (x_0) = \{ x = { }^t ( x_1 , x_2 , x_3 ) \in B_{r_1}(x_0);{ \ }\psi (x,t_0) = 0 \}
\end{equation*}
and that 
\begin{equation*}
\nabla_x \psi = { }^t \left( \frac{\partial \psi}{\partial x_1}, \frac{\partial \psi}{\partial x_2}, \frac{\partial \psi}{\partial x_3} \right) \neq (0,0,0) { \ }\text{ on }B_{r_1}(x_0) \times B_{r_2}(t_0) .
\end{equation*}
Here $B_{r_1} (x_0) := \{ x \in \mathbb{R}^3; { \ }|x - x_0| < r_1 \}$, $B_{r_2} (t_0) := \{ t \in \mathbb{R}_+;{ \ } |t - t_0| < r_2 \}$, and
\begin{multline*}
C^{2,1} ( B_{r_1} (x_0) \times B_{r_2} (t_0) ) := \{ f \in C ( B_{r_1} (x_0) \times B_{r_2} (t_0) ) ; \\{ \ } \partial_i f, \partial_j \partial_i f, \partial_t f, \partial_i \partial_t f, \partial_j \partial_i \partial_t f \in C ( B_{r_1} (x_0) \times B_{r_2} (t_0) ) \text{ for }i,j=1,2,3 \} .
\end{multline*}
\end{enumerate}
Throughout this paper we write $\Gamma (t)$ instead of $\{ \Gamma (t) \}_{ 0 \leq t < T}$.
\end{definition}

\begin{definition}[Evolving surfaces]\label{def23}{ \ \ \ \ \ }
Let $\{ \Gamma (t) \}_{0 \leq t < T}$ be an evolving $2$-dimensional $C^{2,1}$-surface in $\mathbb{R}^3$ on $[0,T)$ for some $T \in (0, \infty ]$. We simply call $\Gamma ( t )$ an \emph{evolving $2$-dimensional surface in $\mathbb{R}^3$} on $[0,T)$ if $\Gamma (t)$ is a closed Riemannian $2$-dimensional manifold for each fixed $t \in [0, T )$. 
\end{definition}

\begin{definition}[Variations of an evolving surface]\label{def24}{ \ }\\
Let $\{ \Gamma (t) \}_{0 \leq t < T}$ be an evolving $2$-dimensional surface in $\mathbb{R}^3$ on $[0,T)$ for some $T \in (0,\infty ]$. For each $-1 < \varepsilon < 1$, let $\{ \Gamma^\varepsilon (t) \}_{0 \leq t < T}$ be an evolving $2$-dimensional surface in $\mathbb{R}^3$ on $[0,T)$. We call $\Gamma^\varepsilon ( t )$ a \emph{variation of $\Gamma(t)$} if the following two properties hold:
\begin{enumerate}\setlength{\leftskip}{2mm}
\renewcommand{\labelenumi}{(\roman{enumi})}
\item For each $0 \leq t <T$,
\begin{equation*}
\left| \int_{\Gamma (t) \setminus \Gamma^\varepsilon (t)} 1 {  \ } d\mathcal{H}_x^2 \right| + \left| \int_{\Gamma^\varepsilon (t) \setminus \Gamma (t)} 1 { \ }d \mathcal{H}_x^2 \right| \leq |\varepsilon| .
\end{equation*}
\item For each $0 \leq t <T$ and $-1 < \varepsilon_0 <1$,
\begin{equation*}
\lim_{\varepsilon \to \varepsilon_0}\left| \int_{\Gamma^{\varepsilon_0} (t) \setminus \Gamma^\varepsilon (t)} 1 {  \ } d\mathcal{H}_x^2 \right| + \left| \int_{\Gamma^\varepsilon (t) \setminus \Gamma^{\varepsilon_0} (t)} 1 { \ }d \mathcal{H}_x^2 \right| = 0.
\end{equation*}
\end{enumerate}
Here $d \mathcal{H}_x^2$ denotes the $2$-dimensional Hausdorff measure, that is, $\int_{\Gamma (t)}1 { \ }d \mathcal{H}_x^2$ is the surface area of $\Gamma (t)$. 
\end{definition}
\noindent Note that $\Gamma^\varepsilon (t)|_{\varepsilon = 0} = \Gamma (t)$ by definition.

\subsection{Function Spaces and Notation}\label{subsec22}
We introduce functions on surfaces and evolving surfaces. Let $\Gamma_0$ be a $2$-dimensional surface in $\mathbb{R}^3$, and let $\Gamma (t)$ be an evolving $2$-dimensional $C^{2,1}$-surface in $\mathbb{R}^3$ on $[0,T)$ for some $T \in (0 , \infty ]$. Set
\begin{equation*}
\mathcal{S}_T \equiv \mathcal{S}_{T,\Gamma (t)} := \left\{ (x , t )= { }^t (x_1,x_2,x_3 , t) \in \mathbb{R}^4; { \ } (x , t) \in \bigcup_{0<t <T} \{ \Gamma (t) \times \{ t \} \}  \right\}.
\end{equation*}
For each $m \in \mathbb{N} \cup \{ 0 , \infty \}$ we define
\begin{align*}
C^m ( \Gamma_0 ) := & \{ f:\Gamma_0 \to \mathbb{R};{ \ } g|_{\Gamma_0} =f \text{ for some }g \in C^m ( \mathbb{R}^3 ) \}, \\
C^m_0 (\Gamma_0 ) := & \{ f \in C^m (\Gamma_0 ) ; { \ }\text{supp}f \text{ does not intersect the boundary of } \Gamma_0 \},\\
C ( \mathcal{S}_T ) := &\{ f: \mathcal{S}_T \to \mathbb{R} ; { \ } g|_{\mathcal{S}_T} =f \text{ for some }g\in C(\mathbb{R}^3 \times \mathbb{R})  \} ,\\
C_0 ( \mathcal{S}_T ) :=& \{ f \in C (\mathcal{S}_T ) ; { \ }\text{supp}f \text{ includes } \mathcal{S}_T \text{ and }\\
&\text{supp}f( \cdot , t) \text{ does not intersect the geometric boundary of } \Gamma (t)  \},
\end{align*}
where $C^0 (\mathbb{R}^3) : = C (\mathbb{R}^3)$. Moreover, we write
\begin{align*}
& C^{1,0} ( \mathcal{S}_T ) := \{ f \in C ( \mathcal{S}_T ) ; { \ } \partial_i f \in C ( \mathcal{S}_T ) \text{ for each }i=1,2,3 \} ,\\
& C^{2,1} ( \mathcal{S}_T ) := \{ f \in C^{1,0} ( \mathcal{S}_T ) ; { \ } \partial_j \partial_i f, \partial_t f, \partial_i \partial_t f, \partial_j \partial_i \partial_t f \in C ( \mathcal{S}_T ), i,j=1,2,3 \} ,\\
&C_0^{2,1} ( \mathcal{S}_T ) := C^{2,1} (\mathcal{S}_T) \cap C_0 (\mathcal{S}_T).
\end{align*}
Similarly, we define $C^{m,n} (\mathcal{S}_T)$, $C^{m,n}_0 (\mathcal{S}_T)$, and $C_0^\infty (\mathcal{S}_T) := C^\infty (\mathcal{S}_T) \cap C_0 (\mathcal{S}_T)$, where
\begin{equation*}
C^\infty ( \mathcal{S}_T ) := \{ f: \mathcal{S}_T \to \mathbb{R};{ \ } g|_{\mathcal{S}_T} =f \text{ for some }g \in C^\infty ( \mathbb{R}^3 \times \mathbb{R} ) \}.
\end{equation*}

For $- 1< \varepsilon <1 $, let $\Gamma^\varepsilon (t)$ be a variation of $\Gamma ( t )$. Set
\begin{equation*}
\mathcal{S}^\varepsilon_T \equiv \mathcal{S}_{T,\Gamma^\varepsilon (t)} := \left\{ (x , t )= { }^t (x_1,x_2,x_3 , t) \in \mathbb{R}^4; (x , t) \in \bigcup_{0<t <T} \{ \Gamma^\varepsilon (t) \times \{ t \} \}  \right\}.
\end{equation*}

Let us explain some conventions used in this paper. We use italic characters $i,j,k,\ell,i' j'$ as $1,2,3$, and use Greek characters $\alpha, \beta, \zeta, \eta,\alpha',\beta'$ as $1,2$. Moreover, we often use the following Einstein summation convention:
\begin{equation*}
a_{i j}b_j = \sum_{j=1}^3 a_{i j} b_j, { \ }a_{i j} b_{i j \ell} = \sum_{i, j=1}^3 a_{i j}b_{i j \ell},{ \ }a_{i j} b_{ i \alpha} c_{\alpha \beta} = \sum_{i =1}^3 \sum_{\alpha = 1}^2 a_{i j} b_{i \alpha} c_{\alpha \beta}.
\end{equation*}

Let $\mathcal{X}$ be a set. The symbol $M_{p \times q} ( \mathcal{X}) $ denotes the set of all $p \times q$ matrices whose component belonging to $\mathcal{X}$, that is, $M \in M_{p \times q} ( \mathcal{X}) $ if and only if
\begin{equation*}M=
\begin{pmatrix}
[M]_{11}& [M]_{12}& \cdots & [M]_{1 q}\\
[M]_{21}& [M]_{22} &\cdots & [M]_{2 q}\\
\vdots& \vdots &  & \vdots\\ 
[M]_{p 1}& [M]_{ p 2} & \cdots & [M]_{p q}
\end{pmatrix}
\end{equation*}
and $[M]_{\mathfrak{i} \mathfrak{j}} \in \mathcal{X}$ $(\mathfrak{i}=1,2,\ldots, p, { \ }\mathfrak{j}=1,2,\ldots, q)$, where $[M]_{\mathfrak{i} \mathfrak{j}}$ denotes the $(\mathfrak{i},\mathfrak{j})$-th component of the matrix $M$. Remark that we can write $M = ([M]_{\mathfrak{i} \mathfrak{j}})_{p \times q}$.

Next we introduce important notation. By $n = n (x, t) = { }^t ( n_1 , n_2 , n_3)$ we mean the \emph{unit outer normal vector} of $\Gamma (t)$ at $x \in \Gamma (t)$ for each fixed $t \in [0,T )$. In this paper, we use the following notation:
\begin{align*}
&\partial_i^{\Gamma}:= ( \delta_{i j } - n_i n_j )\partial_j { \ \ }\left( = \sum_{j=1}^3 (\delta_{i j}- n_i n_j )\partial_j \right) , \\
&\nabla_{\Gamma} := { }^t (\partial_1^{\Gamma}, \partial_2^{\Gamma} , \partial_3^{\Gamma}) ,\\
& \Delta_{\Gamma} := (\partial_1^{\Gamma})^2 + (\partial_2^{\Gamma})^2 + (\partial_3^{\Gamma})^2  .
\end{align*}
Here $\delta_{ij}$ is Kronecker's delta. Moreover, for $f = { }^t ( f_1 , f_2 ,f_3 ) \in [ C^{1}( \Gamma ( t ) )]^3$ and $f_{i j}, g \in C^1 ( \Gamma (t))$,
\begin{align*}
{\rm{div}}_\Gamma f & := \partial_1^{\Gamma} f_1 + \partial_2^{\Gamma} f_2 + \partial_3^{\Gamma}f_3,\\
{\rm{div}}_\Gamma (f_{i j})_{3 \times 3} & := { }^t ( \partial_j^{\Gamma} f_{1 j},  \partial_j^{\Gamma} f_{2 j} , \partial_j^{\Gamma}f_{3 j}),\\
{\rm{grad}}_\Gamma g & := \nabla_{\Gamma} g ={ }^t (\partial_1^\Gamma g , \partial_2^\Gamma g , \partial_3^\Gamma g ).
\end{align*}
Let $H_\Gamma$ and $P_\Gamma$ be the mean curvature of $\Gamma (t)$ and the orthogonal projection to a tangent space of $\Gamma (t)$ defined by
\begin{align}
H_\Gamma &= H_\Gamma (x,t) :=  - {\rm{div}}_\Gamma n, \label{eq21} \\
[P_\Gamma]_{i j} &= [P_\Gamma ( x , t )]_{i j} : = \delta_{i j} - n_i n_j{ \ }(i ,j = 1,2,3),\label{eq22}
\end{align}
respectively. By definition, we easily check that $P_\Gamma n = { }^t (0,0,0)$ and $P_\Gamma^2 = P_\Gamma$. Note that $P_\Gamma = I_{3 \times 3} - n \otimes n$ and that $n_1^2 + n_2^2 + n_3^2 = 1$. Note also that $P_\Gamma (\nabla g) = \nabla_\Gamma g$ and that $f = P_\Gamma f + (f \cdot n) n$.

For $- 1< \varepsilon <1 $, let $\Gamma^\varepsilon (t)$ be a variation of $\Gamma ( t )$.  By $n^\varepsilon = n^\varepsilon (x, t) = { }^t ( n^\varepsilon_1 , n_2^\varepsilon , n^\varepsilon_3)$ we mean the \emph{unit outer normal vector} of $\Gamma^\varepsilon (t)$ at $x \in \Gamma^\varepsilon (t)$ for each fixed $t \in [0,T )$. By definition, we see that for $g \in C^1 ( \Gamma^\varepsilon (t) )$,
\begin{equation*}
\partial_i^{\Gamma^\varepsilon}g = ( \delta_{i j } - n_i^\varepsilon n_j^\varepsilon )\partial_j g { \ \ }\left( = \sum_{j=1}^3 (\delta_{i j}- n_i^\varepsilon n_j^\varepsilon )\partial_j g \right).
\end{equation*}
Moreover, for $f = { }^t ( f_1 , f_2 ,f_3 ) \in [ C^{1}( \Gamma^\varepsilon ( t ) )]^3$ and $f_{i j} , g \in C^1 ( \Gamma^\varepsilon (t))$,
\begin{align*}
{\rm{div}}_{\Gamma^\varepsilon} f =& \partial_1^{\Gamma^\varepsilon} f_1 + \partial_2^{\Gamma^\varepsilon} f_2 + \partial_3^{\Gamma^\varepsilon}f_3,\\
{\rm{div}}_{\Gamma^\varepsilon } (f_{i j})_{3 \times 3} = & { }^t ( \partial_j^{\Gamma^\varepsilon} f_{1 j},  \partial_j^{\Gamma^\varepsilon} f_{2 j} , \partial_j^{\Gamma^\varepsilon}f_{3 j}),\\
{\rm{grad}}_{\Gamma^\varepsilon} g = & \nabla_{\Gamma^\varepsilon} g = { }^t (\partial_1^{\Gamma^\varepsilon} g , \partial_2^{\Gamma^\varepsilon} g , \partial_3^{\Gamma^\varepsilon} g ).
\end{align*}
Let $H_{\Gamma^\varepsilon}$ and $P_{\Gamma^\varepsilon}$ be the mean curvature of $\Gamma^\varepsilon (t)$ and the orthogonal projection to a tangent space of $\Gamma^\varepsilon (t)$ defined by
\begin{align}
&H_{\Gamma^\varepsilon} = H_{\Gamma^\varepsilon} (x,t) :=  - {\rm{div}}_{\Gamma^\varepsilon} n^\varepsilon , \label{eq23} \\
&[P_{\Gamma^\varepsilon}]_{i j} = [P_{\Gamma^\varepsilon} ( x , t )]_{i j} : = \delta_{i j} - n_i^\varepsilon n_j^\varepsilon{ \ }(i ,j = 1,2,3).\label{eq24}
\end{align}

Let us introduce three strain rate tensors and one surface stress tensor. For every $f = { }^t (f_1,f_2,f_3) \in [C^1 (\Gamma (t))]^3$ and $g , \mu, \lambda \in C(\Gamma (t))$,
\begin{align*}
D (f) &:= \frac{1}{2} \{ (\nabla f ) + { }^t (\nabla f ) \} ,\\
\mathbb{D}_{\Gamma} (f) &:= \frac{1}{2} \{ (\nabla_{\Gamma} f ) + { }^t (\nabla_{\Gamma} f ) \} ,\\
D_{\Gamma} (f) &:= \frac{1}{2} P_\Gamma \{ (\nabla f ) + { }^t (\nabla f ) \} P_\Gamma ,\\
S_\Gamma (f , g , \mu , \lambda ) & := 2 \mu D_\Gamma (f) + \lambda P_\Gamma ({\rm{div}}_\Gamma f) - P_\Gamma g. 
\end{align*}
We call $D (f)$ a strain rate tensor, $D_\Gamma (f)$ a surface strain rate tensor, and \\
$S_\Gamma (f, g, \mu, \lambda )$ a surface stress tensor. We also call $\mathbb{D}_\Gamma (f)$ a tangential strain rate and $D_\Gamma (f)$ a projected strain rate. By the definition of $P_\Gamma$ and $\nabla_\Gamma$, we find that
\begin{equation*}
D_\Gamma (f) = \frac{1}{2} \{ (P_\Gamma (\nabla_\Gamma f )) + { }^t (P_\Gamma ( \nabla_\Gamma f )) \}.
\end{equation*}
See Slattery-Sagis-Oh \cite{SSO07} and Lemma \ref{lem26}.

\subsection{Properties of Operators for Evolving Surfaces}\label{subsec23}

Let us now study several operators such as ${\rm{div}}_\Gamma$, ${\rm{grad}}_\Gamma$, and $P_\Gamma$, and the strain rate tensors $D (f)$, $\mathbb{D}_\Gamma (f)$, $D_\Gamma (f)$, and the surface stress tensor $S_\Gamma$. Let $\Gamma (t)$ be an evolving $2$-dimensional $C^{2,1}$-surface in $\mathbb{R}^3$ on $[0,T)$ for some $T \in (0,\infty]$, and let $\Gamma^\varepsilon (t)$ is a variation of $\Gamma (t)$.
\begin{lemma}[Properties of surface gradient and divergence]\label{lem25}{ \ }\\
$(\mathrm{i})$ For every $f \in C^1 ( \Gamma (t)) $ and $v = { }^t (v_1,v_2,v_3) \in [ C^1 (\Gamma (t)) ]^3$,
\begin{align}
P_\Gamma {\rm{grad}}_\Gamma f & = {\rm{grad}}_\Gamma f \notag,\\
n \cdot {\rm{grad}}_\Gamma f  & = 0, \notag \\
{\rm{div}}_\Gamma (P_\Gamma f) & = {\rm{grad}}_\Gamma f  + f H_\Gamma n, \label{eq27}\\
P_\Gamma {\rm{div}}_\Gamma ( P_\Gamma f) & = {\rm{grad}}_\Gamma f, \notag \\
{\rm{div}}_\Gamma ((P_\Gamma f) v ) & = ({\rm{grad}}_\Gamma f ) \cdot v + f H_\Gamma (n \cdot v) + f ({\rm{div}}_\Gamma v)  \label{eq29},\\
(v , \nabla ) f & = (v , \nabla_\Gamma ) f + (v \cdot n) (n , \nabla ) f \label{eq210}.
\end{align}
$(\mathrm{ii})$ For every $f \in C^1 ( \Gamma^\varepsilon (t)) $ and $v = { }^t (v_1,v_2,v_3) \in [ C^1 (\Gamma^\varepsilon (t)) ]^3$,
\begin{align*}
P_{\Gamma^\varepsilon} {\rm{grad}}_{\Gamma^\varepsilon} f & = {\rm{grad}}_{\Gamma^\varepsilon} f,\\
n^\varepsilon \cdot {\rm{grad}}_{\Gamma^\varepsilon} f  & = 0,\\
{\rm{div}}_{\Gamma^\varepsilon} (P_{\Gamma^\varepsilon} f ) & = {\rm{grad}}_{\Gamma^\varepsilon} f  + f H_{\Gamma^\varepsilon} n^\varepsilon, \\
P_{\Gamma^\varepsilon} {\rm{div}}_{\Gamma^\varepsilon} (P_{\Gamma^\varepsilon} f) & = {\rm{grad}}_{\Gamma^\varepsilon} f , \\
{\rm{div}}_{\Gamma^\varepsilon} ((P_{\Gamma^\varepsilon} f) v ) & = ({\rm{grad}}_{\Gamma^\varepsilon} f ) \cdot v + f H_{\Gamma^\varepsilon} (n^\varepsilon \cdot v) + f ({\rm{div}}_{\Gamma^\varepsilon} v),\\
(v , \nabla ) f & = (v , \nabla_{\Gamma^\varepsilon} ) f + (v \cdot n^\varepsilon) (n^\varepsilon , \nabla ) f.
\end{align*}
\end{lemma}
Since the proof of Lemma \ref{lem25} is not difficult, the proof is left to the reader.

The following lemma is useful to study the fluid flow on an evolving surface.
\begin{lemma}[Properties of strain and stress tensors]\label{lem26}{ \ }\\
\noindent $(\mathrm{i})$ For all $v = { }^t (v_1,v_2,v_3)$, $\varphi = { }^t ( \varphi_1 , \varphi_2 , \varphi_3 ) \in [C^1 (\Gamma (t))]^3$,
\begin{align}
P_\Gamma D (v) P_\Gamma = & P_\Gamma \mathbb{D}_\Gamma (v) P_\Gamma, \label{eq211}\\
D_\Gamma (v) : D_\Gamma (\varphi ) = & D_\Gamma (v) : \mathbb{D}_\Gamma ( \varphi ) \label{eq212}.
\end{align}
\noindent $(\mathrm{ii})$ For all $v = { }^t (v_1,v_2,v_3) \in [C^2 (\Gamma (t))]^3$ and $g , \mu , \lambda \in C^1 (\Gamma(t))$,
\begin{align}
{\rm{div}}_\Gamma \{ \mu D_\Gamma (v) v \} = & {\rm{div}}_\Gamma \{ \mu D_\Gamma (v) \} \cdot v + \mu D_\Gamma (v): D_{\Gamma}(v)\label{eq213},\\
{\rm{div}}_\Gamma \{ \lambda P_\Gamma ({\rm{div}}_\Gamma v ) v \} = & {\rm{div}}_\Gamma \{ \lambda P_\Gamma ({\rm{div}}_\Gamma v ) \} \cdot v + \lambda |{\rm{div}}_\Gamma v |^2 ,\label{eq214}\\
{\rm{div}}_\Gamma \{ S_\Gamma (v,g, \mu ,\lambda ) v \} & - {\rm{div}}_\Gamma \{ S_\Gamma (v,g,\mu,\lambda )\} \cdot v \notag \\
= & 2 \mu D_\Gamma (v) : D_\Gamma (v) + \lambda |{\rm{div}}_\Gamma v |^2 - g ({\rm{div}}_\Gamma v) .\label{eq215}
\end{align}
\noindent $(\mathrm{iii})$ For all $v = { }^t (v_1,v_2,v_3) \in [C^1 (\Gamma (t))]^3$ and $g , \mu , \lambda \in C (\Gamma(t))$,
\begin{equation}
S_\Gamma (v ,g , \mu , \lambda ) : D_\Gamma (v) = 2 \mu D_\Gamma (v): D_\Gamma (v) + \lambda |{\rm{div}}_\Gamma v |^2 - g ({\rm{div}}_\Gamma v ).\label{eq216}
\end{equation}

\end{lemma}

\begin{proof}[Proof of Lemma \ref{lem26}] We first prove $(\mathrm{i})$. Fix $v = { }^t (v_1,v_2,v_3)$, $\varphi = { }^t ( \varphi_1 , \varphi_2 , \varphi_3 ) \in [C^1 (\Gamma (t))]^3$. By definition, we check that for each $i,j=1,2,3,$
\begin{align*}
[2 P_\Gamma D (v)]_{i j} & = \partial_i^\Gamma v_j + \partial_j v_i - n_i (n \cdot \partial_j v),\\
[2 P_\Gamma \mathbb{D}_\Gamma (v)]_{i j} &= \partial_i^\Gamma v_j + \partial_j^\Gamma v_i - n_i (n \cdot \partial_j^\Gamma v),\\
[2 P_\Gamma D (v) P_\Gamma ]_{i j} & = \partial_i^{\Gamma} v_j + \partial_j^{\Gamma} v_i - n_i (n \cdot \partial_j^{\Gamma} v ) - n_j ( n \cdot \partial_i^{\Gamma} v ),\\
[2 P_\Gamma \mathbb{D}_\Gamma (v) P_\Gamma ]_{i j} & = \partial_i^{\Gamma} v_j + \partial_j^{\Gamma} v_i - n_i (n \cdot \partial_j^{\Gamma} v ) - n_j ( n \cdot \partial_i^{\Gamma} v ).
\end{align*}
This shows \eqref{eq211}. Next we shall show that
\begin{equation}
D_\Gamma (v) : \{ D_\Gamma (\varphi ) - \mathbb{D}_\Gamma ( \varphi ) \} = 0. \label{eq217}
\end{equation}
We now use the Einstein summation convention. From
\begin{align*}
2 [D_\Gamma (v )]_{i j} = & \partial_i^{\Gamma} v_j + \partial_j^{\Gamma} v_i - n_i (n \cdot \partial_j^{\Gamma} v ) - n_j ( n \cdot \partial_i^{\Gamma} v ),\\
2 [D_\Gamma (\varphi )]_{i j} = & \partial_i^{\Gamma} \varphi_j + \partial_j^{\Gamma} \varphi_i - n_i (n \cdot \partial_j^{\Gamma} \varphi ) - n_j ( n \cdot \partial_i^{\Gamma} \varphi ),\\
2 [\mathbb{D}_\Gamma (\varphi )]_{i j} = & \partial_j^{\Gamma} \varphi_i + \partial_i^{\Gamma} \varphi_j = \partial_i^{\Gamma} \varphi_j + \partial_j^{\Gamma} \varphi_i,
\end{align*}
we find that
\begin{multline}\label{eq218}
4 D_\Gamma (v) : \{ D_\Gamma (\varphi ) - \mathbb{D}_\Gamma ( \varphi ) \} \\
= \{ \partial_i^{\Gamma} v_j + \partial_j^{\Gamma} v_i - n_i (n \cdot \partial_j^{\Gamma} v ) - n_j ( n \cdot \partial_i^{\Gamma} v )  \} \{ - n_i (n \cdot \partial_j^{\Gamma} \varphi ) - n_j ( n \cdot \partial_i^{\Gamma} \varphi ) \}\\
= - ( \partial_i^{\Gamma} v_j + \partial_j^{\Gamma} v_i) \{ n_i (n \cdot \partial_j^{\Gamma} \varphi ) + n_j ( n \cdot \partial_i^{\Gamma} \varphi ) \}\\
+ \{ n_i (n \cdot \partial_j^{\Gamma} v ) + n_j ( n \cdot \partial_i^{\Gamma} v )  \} \{ n_i (n \cdot \partial_j^{\Gamma} \varphi ) + n_j ( n \cdot \partial_i^{\Gamma} \varphi ) \} \\
=: - A_1 + A_2.
\end{multline}
A direct calculation gives
\begin{multline*}
A_1 = (\partial_i^{\Gamma} v_j + \partial_j^{\Gamma} v_i) \{ n_i (n \cdot \partial_j^{\Gamma} \varphi ) + n_j ( n \cdot \partial_i^{\Gamma} \varphi ) \}\\
=\partial_i^{\Gamma} v_j \{ n_i (n \cdot \partial_j^{\Gamma} \varphi ) \} + \partial_j^{\Gamma} v_i \{ n_j (n \cdot \partial_i^{\Gamma} \varphi ) \}\\
+ \partial_i^{\Gamma} v_j \{ n_j (n \cdot \partial_i^{\Gamma} \varphi ) \} + \partial_j^{\Gamma} v_i \{ n_i (n \cdot \partial_j^{\Gamma} \varphi ) \}.
\end{multline*}
Since $n_ j \partial_j^{\Gamma} = n \cdot {\rm{grad}}_\Gamma = 0$, we see that
\begin{equation*}
\partial_i^{\Gamma} v_j \{ n_i (n \cdot \partial_j^{\Gamma} \varphi ) \} + \partial_j^{\Gamma} v_i \{ n_j (n \cdot \partial_i^{\Gamma} \varphi ) \} = n_i \partial_i^{\Gamma} v_j (n \cdot \partial_j^{\Gamma} \varphi )  + n_j \partial_j^{\Gamma} v_i (n \cdot \partial_i^{\Gamma} \varphi ) = 0 .
\end{equation*}
It is easy to check that
\begin{multline*}
 \partial_i^{\Gamma} v_j \{ n_j (n \cdot \partial_i^{\Gamma} \varphi ) \} + \partial_j^{\Gamma} v_i \{ n_i (n \cdot \partial_j^{\Gamma} \varphi ) \} \\
= 2 \{ (n \cdot \partial_1^{\Gamma} v) (n \cdot \partial_1^{\Gamma} \varphi) + (n \cdot \partial_2^{\Gamma} v) (n \cdot \partial_2^{\Gamma} \varphi) + (n \cdot \partial_3^{\Gamma} v) (n \cdot \partial_3^{\Gamma} \varphi) \}.
\end{multline*}
Consequently, we have
\begin{equation}\label{eq219}
A_1 = 2 (n \cdot \partial_j^{\Gamma} v) (n \cdot \partial_j^{\Gamma} \varphi) .
\end{equation}
Next we consider
\begin{multline*}
A_2 = \{ n_i (n \cdot \partial_j^{\Gamma} v) + n_j ( n \cdot \partial_i^{\Gamma} v) \} \{ n_i (n \cdot \partial_j^{\Gamma} \varphi ) + n_j ( n \cdot \partial_i^{\Gamma} \varphi) \}\\
= n_i n_j (n \cdot \partial_j^{\Gamma} v) (n \cdot \partial_i^{\Gamma} \varphi) + n_i n_j (n \cdot \partial_i^{\Gamma} v) (n \cdot \partial_j^{\Gamma} \varphi)\\
+ n_i n_i (n \cdot \partial_j^{\Gamma} v) (n \cdot \partial_j^{\Gamma} \varphi) + n_j n_j (n \cdot \partial_i^{\Gamma} v) (n \cdot \partial_i^{\Gamma} \varphi) .
\end{multline*}
Since $n_j \partial_j^{\Gamma } = 0$, we check that
\begin{equation*}
n_i n_j (n \cdot \partial_j^{\Gamma} v) (n \cdot \partial_i^{\Gamma} \varphi) = n_i (n \cdot n_j \partial_j^{\Gamma} v) (n \cdot \partial_i^{\Gamma} \varphi) = 0.
\end{equation*}
This shows that
\begin{equation*}
n_i n_j (n \cdot \partial_j^{\Gamma} v) (n \cdot \partial_i^{\Gamma} \varphi) + n_i n_j (n \cdot \partial_i^{\Gamma} v) (n \cdot \partial_j^{\Gamma} \varphi) =0.
\end{equation*}
By $n_1^2 + n_2^2 + n_3^2 =1$, we have
\begin{align*}
n_i n_i (n \cdot \partial_j^{\Gamma} v) (n \cdot \partial_j^{\Gamma} \varphi) = (n \cdot \partial_j^{\Gamma} v) (n \cdot \partial_j^{\Gamma} \varphi) .
\end{align*}
As a result, we obtain
\begin{equation}\label{eq220}
A_2 = 2 (n \cdot \partial_j^{\Gamma} v) (n \cdot \partial_j^{\Gamma} \varphi) .
\end{equation}
From \eqref{eq218}, \eqref{eq219}, and \eqref{eq220}, we see \eqref{eq217}. Therefore we have \eqref{eq212}.

Next we show the assertion $(\mathrm{ii})$. Fix $v = { }^t (v_1,v_2,v_3) \in [C^2 (\Gamma (t))]^3$ and $g, \mu , \lambda \in C^1 (\Gamma(t))$. Now we derive \eqref{eq213}. Since
\begin{equation*}{\rm{div}}_\Gamma \{ \mu D_\Gamma (v) v \} ={\rm{div}}_\Gamma
\begin{pmatrix}
\mu [ D_\Gamma (v) ]_{11} v_1  + \mu [ D_\Gamma (v) ]_{12} v_2  + \mu [ D_\Gamma (v) ]_{13} v_3\\
\mu [ D_\Gamma (v) ]_{21} v_1 + \mu [ D_\Gamma (v) ]_{22} v_2 + \mu [ D_\Gamma (v) ]_{23} v_3\\
\mu [ D_\Gamma (v) ]_{31} v_1 + \mu [ D_\Gamma (v) ]_{32} v_2 + \mu [ D_\Gamma (v) ]_{33} v_3
\end{pmatrix}
\end{equation*}
and
\begin{multline*} {\rm{div}}_\Gamma \{ \mu D_\Gamma (v) \} \cdot v = {\rm{div}}_\Gamma
\begin{pmatrix}
\mu [ D_\Gamma (v) ]_{11} & \mu [ D_\Gamma (v) ]_{12} &  \mu [ D_\Gamma (v) ]_{13}\\
\mu [ D_\Gamma (v) ]_{21} & \mu [ D_\Gamma (v) ]_{22} &  \mu [ D_\Gamma (v) ]_{23}\\
\mu [ D_\Gamma (v) ]_{31} & \mu [ D_\Gamma (v) ]_{32} &  \mu [ D_\Gamma (v) ]_{33}
\end{pmatrix} \cdot \begin{pmatrix}
v_1\\
v_2\\
v_3
\end{pmatrix}\\
=
\begin{pmatrix}
\partial_1^\Gamma (\mu [ D_\Gamma (v) ]_{11}) + \partial_2^\Gamma (\mu [ D_\Gamma (v) ]_{12}) + \partial_3^\Gamma (\mu [ D_\Gamma (v) ]_{13})\\
\partial_1^\Gamma (\mu [ D_\Gamma (v) ]_{21}) + \partial_2^\Gamma (\mu [ D_\Gamma (v) ]_{22}) + \partial_3^\Gamma (\mu [ D_\Gamma (v) ]_{23})\\
\partial_1^\Gamma (\mu [ D_\Gamma (v) ]_{31}) + \partial_2^\Gamma (\mu [ D_\Gamma (v) ]_{32}) + \partial_3^\Gamma (\mu [ D_\Gamma (v) ]_{33})
\end{pmatrix} \cdot
\begin{pmatrix}
v_1\\
v_2\\
v_3
\end{pmatrix},
\end{multline*}
we find that
\begin{equation*}
{\rm{div}}_\Gamma \{ \mu D_\Gamma (v) v \} - {\rm{div}}_\Gamma \{ \mu D_\Gamma (v) \} \cdot v = \mu D_\Gamma (v): \mathbb{D}_\Gamma (v).
\end{equation*}
Note that $[D_\Gamma (v) ]_{ji}=[D_\Gamma (v)]_{ij}$ and $[\mathbb{D}_\Gamma (v) ]_{i j} = ( \partial_j^\Gamma v_i + \partial_i^\Gamma v_j )/2$. Using \eqref{eq212}, we have \eqref{eq213}. 

We now attack \eqref{eq214} and \eqref{eq215}. Set $M = \lambda P_\Gamma ({\rm{div}}_\Gamma v)$. It is clear that
\begin{multline*}
({\rm{div}}_\Gamma M) \cdot v = (\partial_1^\Gamma [M]_{11} + \partial_2^\Gamma [M]_{12} + \partial_3^\Gamma [M]_{13}) v_1 \\
+ (\partial_1^\Gamma [M]_{21} + \partial_2^\Gamma [M]_{22} + \partial_3^\Gamma [M]_{23}) v_2 +  (\partial_1^\Gamma [M]_{31} + \partial_2^\Gamma [M]_{32} + \partial_3^\Gamma [M]_{33}) v_3. 
\end{multline*}
Since $[M]_{i j} = [\lambda P_\Gamma ({\rm{div}}_\Gamma v)]_{ij} =\lambda ({\rm{div}}_\Gamma v) (\delta_{i j} - n_i n_j) = [M]_{j i}$ and $n_j \partial_j^ \Gamma =0$, we observe that
\begin{multline*}
{\rm{div}}_\Gamma \{ M v \} = {\rm{div}}_\Gamma 
\begin{pmatrix}
[ M ]_{11} v_1 + [ M ]_{12} v_2 + [ M ]_{13} v_3\\
[ M ]_{21} v_1 + [ M ]_{22} v_2 + [ M ]_{23} v_3\\
[ M ]_{31} v_1 + [ M ]_{32} v_2 + [ M ]_{33} v_3
\end{pmatrix}\\
= ({\rm{div}}_\Gamma M ) \cdot v + \lambda ({\rm{div}}_\Gamma v ) (\partial_1^\Gamma v_1 + \partial_2^\Gamma v_2 + \partial_3^\Gamma v).
\end{multline*}
Therefore we have \eqref{eq214}. Note that $[M]_{ij} \partial^\Gamma_j v_i = \lambda ({\rm{div}}_\Gamma v) \delta_{i j} \partial^\Gamma_i v_i$ for each fixed $i$. From \eqref{eq27} and \eqref{eq29}, we have
\begin{equation}
{\rm{div}}_\Gamma \{ (P_\Gamma g ) v \} - {\rm{div}}_\Gamma \{ P_\Gamma g \} \cdot v = g ({\rm{div}}_\Gamma v ).\label{eq221}
\end{equation}
Using \eqref{eq213}, \eqref{eq214}, and \eqref{eq221}, we derive \eqref{eq215}.

Finally, we show $(\mathrm{iii})$. Fix $v = { }^t (v_1,v_2,v_3) \in [C^1 (\Gamma (t))]^3$ and $g , \mu , \lambda \in C (\Gamma(t))$. Since $n_j \partial^\Gamma_j =0$, we see that
\begin{multline*}
\lambda P_\Gamma ({\rm{div}}_\Gamma v ) : D_\Gamma (v) \\
= \lambda ({\rm{div}}_\Gamma v ) (\delta_{i j} -n_i n_j) \{ \partial_i^{\Gamma} v_j + \partial_j^{\Gamma} v_i - n_i (n \cdot \partial_j^{\Gamma} v ) - n_j ( n \cdot \partial_i^{\Gamma} v )\} /2 \\
= \lambda ({\rm{div}}_\Gamma v ) \delta_{i j} \{ \partial_i^{\Gamma} v_j + \partial_j^{\Gamma} v_i - n_i (n \cdot \partial_j^{\Gamma} v ) - n_j ( n \cdot \partial_i^{\Gamma} v )\} /2 = \lambda |{\rm{div}}_\Gamma v |^2
\end{multline*}
and that
\begin{multline*}
- P_\Gamma g : D_\Gamma (v) = - g (\delta_{i j} - n_i n_j ) \{ \partial_i^{\Gamma} v_j + \partial_j^{\Gamma} v_i - n_i (n \cdot \partial_j^{\Gamma} v ) - n_j ( n \cdot \partial_i^{\Gamma} v )\} /2 \\
= - g ({\rm{div}}_\Gamma v ).
\end{multline*}
Therefore we find that
\begin{align*}
S_\Gamma (v,g,\mu , \lambda ) : D_\Gamma (v) = & ( 2 \mu D_\Gamma (v) + \lambda P_\Gamma ({\rm{div}}_\Gamma v ) - P_\Gamma g ) : D_\Gamma (v)\\
= & 2 \mu D_\Gamma (v) : D_\Gamma (v) + \lambda |{\rm{div}}_\Gamma v |^2 - g ({\rm{div}}_\Gamma v ).
\end{align*}
Therefore the lemma follows.
 \end{proof}

Next we study material derivatives on an evolving surface.
\begin{lemma}[Material derivatives on evolving surface]\label{lem27}{ \ }\\
Let $v = { }^t (v_1,v_2,v_3) \in [C^{1,1} (\mathcal{S}_T)]^3$. For every $f \in C^{1,1} (\mathcal{S}_T)$,
\begin{align*}
D_t f &:= \partial_t f + (v , \nabla ) f ,\\
D_t^N f & := \partial_t f + (v \cdot n )(n, \nabla ) f ,\\
D_t^\Gamma f & := \partial_t f + (v , \nabla_{\Gamma}) f . 
\end{align*}
Then for every $f \in C^{1,1}(\mathcal{S}_T)$,
\begin{align}
D_t^N f + {\rm{div}}_\Gamma (f v) & = D_t f + ({\rm{div}}_\Gamma v ) f, \label{eq222}\\
D_t^N ( f v) + {\rm{div}}_\Gamma (f v \otimes v) & = \{ D_t f + ({\rm{div}}_\Gamma v ) f \} v + f D_t v \label{eq223}.
\end{align}
\end{lemma}

\begin{proof}[Proof of Lemma \ref{lem27}]
From \eqref{eq210}, we see that
\begin{align*}
D_t^N f + {\rm{div}}_\Gamma ( f v ) = & \{ \partial_t f + (v ,\nabla - \nabla_{\Gamma}) f \} + \{ ( v , \nabla_{\Gamma}) f + ({\rm{div}}_\Gamma v) f \}\\
= & D_t f + ({\rm{div}}_\Gamma v ) f
\end{align*}
and that
\begin{multline*}
D_t^N (f v ) + {\rm{div}}_\Gamma ( f v \otimes v) \\
= \{ ( D_t^N f ) v + f (D_t^N v ) \} + ( \{ (v,\nabla_{\Gamma} ) f \} v + f (v ,\nabla_{\Gamma})v + f ({\rm{div}}_\Gamma v) v )\\
= \{ D_t^N f + (v ,\nabla_{\Gamma} ) f +({\rm{div}}_\Gamma v) f \} v + f \{ D_t^N v + (v , \nabla_{\Gamma} ) v\}\\
=\{ D_t f + ({\rm{div}}_\Gamma v ) f \} v + f D_t v  .
\end{multline*}
Therefore the lemma follows. Note that $[f v \otimes v]_{i j} = f v_i v_j$.
 \end{proof}

Finally we state integration by parts on surfaces.
\begin{lemma}[Integration by parts on surfaces]\label{lem28}{ \ }\\
Let $\Gamma_0$ be $2$-dimensional surface in $\mathbb{R}^3$. Then two assertions hold;\\
$(\mathrm{i})$ For each $f \in C^1 (\Gamma_0 )$, $g \in C^1_0 (\Gamma_0)$, $\varphi \in [ C^1_0 ( \Gamma_0 )]^3$, and $m \in \{ 1, 2, 3 \}$,
\begin{align*}
\int_{\Gamma_0} \{ (\partial_m^{\Gamma_0} f )g \}(x) { \ }d  \mathcal{H}_x^2 & = - \int_{\Gamma_0} \{ f (\partial_m^{\Gamma_0} g)\} (x){ \ }d \mathcal{H}_x^2 - \int_{\Gamma_0} \{ H_{\Gamma_0} n_m f g\} (x) { \ }d \mathcal{H}_x^2,\\
\int_{\Gamma_0} \{ f ({\rm{div}}_\Gamma \varphi )  \}(x) { \ }d  \mathcal{H}_x^2 & = - \int_{\Gamma_0} \{ {\rm{grad}}_\Gamma f + f H_{\Gamma_0} n \} (x) \cdot \varphi (x){ \ }d \mathcal{H}_x^2.
\end{align*}
$(\mathrm{ii})$ Assume that $\Gamma_0$ is a closed surface. Then  for each $f , g \in C^1 (\Gamma_0 )$, $\varphi \in [ C^1 ( \Gamma_0 )]^3$, and $m \in \{ 1, 2, 3 \}$,
\begin{align*}
\int_{\Gamma_0} \{ (\partial_m^{\Gamma_0} f )g \}(x) { \ }d  \mathcal{H}_x^2 & = - \int_{\Gamma_0} \{ f (\partial_m^{\Gamma_0} g)\} (x){ \ }d \mathcal{H}_x^2 - \int_{\Gamma_0} \{ H_{\Gamma_0} n_m f g\} (x) { \ }d \mathcal{H}_x^2,\\
\int_{\Gamma_0} \{ f ( {\rm{div}}_\Gamma \varphi ) \} (x) { \ }d  \mathcal{H}_x^2 & = - \int_{\Gamma_0} \{ {\rm{grad}}_\Gamma f + f H_{\Gamma_0} n \} (x) \cdot \varphi (x){ \ }d \mathcal{H}_x^2.
\end{align*}
Here $n = n (x) = { }^t (n_1 , n_2 ,n_3)$ denotes the unit outer normal vector at $x \in \Gamma_0$, and $H_{\Gamma_0}$ is the mean curvature defined by $H_{\Gamma_0} = - {\rm{div}}_{\Gamma_0} n$.
\end{lemma}
\noindent The proof of Lemma \ref{lem28} is found in \cite[Chapter 2]{Sim83} and \cite{KLG}.

\section{Flow Maps and Riemannian Metrics}\label{sect3}

This section has two purposes. The first one is to derive the continuity equation for the fluid on an evolving surface. The second one is to investigate a mathematical validity of our energy densities for compressible fluid on the evolving surface. To achieve these purposes, we make use of a flow map on the evolving surface and the Riemannian metric induced by the flow map.

In subsection \ref{subsec31} we introduce a flow map on an evolving surface and the Riemannian metric defined by the flow map. In subsection \ref{subsec32}, by using a flow map and the Riemannian metric, we investigate the orthogonal projection $P_\Gamma$ and surface divergence ${\rm{div}}_\Gamma$. In subsection \ref{subsec33} we derive the continuity equation for the fluid on an evolving surface to prove Proposition \ref{prop13}. In subsection \ref{subsec34} we study the representation of the kinetic energy, dissipation energies, and work for compressible fluid on an evolving surface.

\subsection{Flow Maps and Riemannian Metrics on Evolving Surfaces}\label{subsec31}

We first introduce a flow map on an evolving surface and the Riemannian metric defined by the flow map. Then we study surface area integral by applying the flow map and the Riemannian metric. We also consider a flow map on a variation of the evolving surface and the Riemannian metric induced by the flow map.

\begin{definition}[Flow map on an evolving surface]{ \ }\\
Let $\Gamma (t)$ be an evolving $2$-dimensional surface in $\mathbb{R}^3$ on $[0,T)$ for some $T \in (0, \infty ]$. Let $x = { }^t (x_1, x_2 , x_3) \in [ C^\infty ( \mathbb{R}^4) ]^3$. We call $x = \hat{x} ( \xi , t )$ a \emph{flow map} on $\Gamma ( t )$ if the three properties hold:\\
$( \mathrm{i} )$ for every $\xi \in \Gamma ( 0 )$
\begin{equation*}
\hat{x} ( \xi , 0 ) = \xi,
\end{equation*}
$( \mathrm{ii})$ for all $\xi \in \Gamma (0)$ and $0 \leq t < T$
\begin{equation*}
\hat{x} ( \xi , t ) \in \Gamma ( t ),
\end{equation*}
$(\mathrm{iii})$ for each $0 \leq t < T$
\begin{equation*}
\hat{x} ( \cdot , t): \Gamma (0) \to \Gamma (t) \text{ is bijective}.
\end{equation*}
The mapping $\xi \mapsto \hat{x} ( \xi , t )$ is called a \emph{flow map} on $\Gamma (t)$, while the mapping $t \mapsto \hat{x} ( \xi , t )$ is called an \emph{orbit} starting from $\xi$.
\end{definition}

\begin{definition}[Velocity determined by a flow map]\label{def32}{ \ \ \ }
Let $\Gamma (t)$ be an evolving $2$-dimensional surface in $\mathbb{R}^3$ on $[0,T)$ for some $T \in (0, \infty ]$. Let $\hat{x} = \hat{x} ( \xi , t )$ be a flow map on $\Gamma ( t )$. Suppose that there is a smooth function $v = v (x , t) = { }^t (v_1 , v_2 , v_3 )$ such that for $\xi \in \Gamma (0)$ and $0<t <T$,
\begin{equation*}
\frac{d \hat{x}}{d t} = \hat{x}_t ( \xi , t ) = v ( \hat{x} ( \xi , t ) ,t).
\end{equation*}
We call the vector-valued function $v$ the \emph{velocity determined by the flow map} $\hat{x} ( \xi , t )$.
\end{definition}

Let us now study fundamental properties of a flow map on an evolving surface and the velocity determined by the flow map. Let $\Gamma (t)$ be a given evolving $2$-dimensional surface in $\mathbb{R}^3$ on $[0,T)$ for some $T \in (0, \infty ]$. Let $\hat{x} = \hat{x} ( \xi , t )$ be a flow map on $\Gamma ( t )$, and let $v = v (x,t)$ be the velocity determined by the flow map $\hat{x}$, i.e. for every $\xi \in \Gamma (0)$ and $0<t <T$,
\begin{equation*}
\begin{cases}
\frac{d \hat{x}}{d t} (\xi, t) = v (\hat{x} ( \xi , t ) , t),\\
\hat{x} (\xi, 0) = \xi .
\end{cases}
\end{equation*}
We assume that $v$ is the total velocity. From now on we write $\Gamma_0 = \Gamma ( 0 )$. By the bijection of the flow map, we see that $\Gamma (t)$ is expressed as follows:
\begin{equation*}
\Gamma (t) = \{ x = { }^t (x_1, x_2 , x_3 ) \in \mathbb{R}^3;{ \ } x = \hat{x}( \xi , t ), {  \ } \xi \in \Gamma_0 \}.
\end{equation*}
\noindent Since $\Gamma_0 (= \Gamma (0))$ is a closed Riemannian $2$-dimensional manifold, there are $\Gamma_m \subset \Gamma_0$, $\Phi_m \in C^\infty (\mathbb{R}^2)$, $U_m \subset \mathbb{R}^2$, $\Psi_m \in C^\infty (\mathbb{R}^3)$, $(m=1,2,\cdots,N)$ such that
\begin{align*}
\bigcup_{m=1}^N \Gamma_m = \Gamma_0,\\
\Gamma_m = \Phi_m (U_m),\\
{\rm{supp}} \Psi_m \subset \Gamma_m,\\
\| \Psi_m \|_{L^\infty} = 1,\\
\sum_{m=1}^N \Psi_m = 1 \text{ on } \Gamma_0 .
\end{align*}
This is a partition of unity. Fix $\xi \in \Gamma_0$. Assume that $\xi \in \Gamma_m$ for some $m \in \{1,2, \cdots, N \}$. Since we can write $\xi = \Phi_m (X)$ for some $X = { }^t (X_1,X_2) \in U_m \subset \mathbb{R}^2$, we set
\begin{equation*}
\tilde{x} = \tilde{x} ( X , t ) = \hat{x} (\Phi_m (X), t) (= \hat{x} ( \xi , t ) ).
\end{equation*}
Then
\begin{equation*}
\begin{cases}
\frac{d \tilde{x}}{d t} = \tilde{x}_t ( X , t) = v ( \tilde{x} (X,t) ,t),\\
\tilde{ x }|_{t=0} = \Phi_m ( X ) (= \xi).
\end{cases}
\end{equation*}
Now we write
\begin{align*}
\Phi := \Phi_m \text{ if } \xi \in \Gamma_m . 
\end{align*}
Then for each $\xi \in \Gamma_0$ and $0 < t < T$,
\begin{equation*}
\begin{cases}
\frac{d \tilde{x}}{d t} =\tilde{x}_t ( X , t) = v ( \tilde{x} (X,t) ,t),\\
\tilde{ x }|_{t=0} = \Phi ( X ) (= \xi ).
\end{cases}
\end{equation*}
We also call $\tilde{x} (X,t)$ a \emph{flow map} on $\Gamma (t)$. For the map $\tilde{x} = \tilde{x} (X,t )$, we define
\begin{equation*}
J = J (X,t) = {\rm{det}} \{ { }^t (\nabla_X \tilde{x}) (\nabla_X \tilde{x}) \} ,
\end{equation*}
where
\begin{equation*}
\nabla_X \tilde{x} = 
\begin{pmatrix}
\frac{\partial \tilde{x}_1}{\partial X_1} & \frac{\partial \tilde{x}_1}{\partial X_2}\\
\frac{\partial \tilde{x}_2}{\partial X_1} & \frac{\partial \tilde{x}_2}{\partial X_2}\\
\frac{\partial \tilde{x}_3}{\partial X_1} & \frac{\partial \tilde{x}_3}{\partial X_2}
\end{pmatrix}.
\end{equation*}
We assume that $J >0$. Indeed, we can choose a coordinate such that $J> 0$ from the definition of evolving surfaces (see Definitions \ref{def21}-\ref{def23}).

Next we study surface area integral by applying the flow maps $\hat{x} ( \xi , t )$ and $\tilde{x}(X,t)$ on $\Gamma ( t )$. For each $f (\cdot,\cdot) \in C ( \mathbb{R}^3 \times \mathbb{R})$, we find that
\begin{equation}\label{eq31}
\int_{\Gamma (t)} f (x,t) { \ }d \mathcal{H}_{x}^2 = \int_U \tilde{\Psi} (X) f ( \tilde{x} (X,t) , t) \sqrt{ J (X,t)  } { \ }d X.
\end{equation}
Here
\begin{multline*}
 \int_U \tilde{\Psi} (X) f ( \tilde{x} (X,t), t ) \sqrt{ J ( X , t )  } { \ } d X\\
 := \sum_{m=1}^N \int_{U_m} \Psi_m (\Phi_m (X)) f ( \tilde{x} (X,t) , t ) \sqrt{ J (X , t )} { \ } d X .
\end{multline*}
Since
\begin{equation*}
\Gamma (t) = \{ x \in \mathbb{R}^3; { \ }x = \hat{x} ( \xi , t ) , { \ }\xi \in \Gamma_0 \},
\end{equation*}
we use the change of variables and usual surface area integral to check that for $0<t<T$
\begin{multline*}
\int_{\Gamma (t)} f (x,t) { \ } d \mathcal{H}_{x}^2 = \int_{\Gamma_0} f (\hat{x} ( \xi , t ),t )  \det ({\nabla_\xi \hat{x}}) { \ }d \mathcal{H}_\xi^2 \\
= \sum_{m = 1}^N \int_{\Gamma_m} \Psi_m ( \xi ) f ( \hat{x} ( \xi , t ) , t ) \det (\nabla_{\xi} \hat{x}) { \ } d \mathcal{H}_{\xi}^2\\
= \sum_{m=1}^N \int_{U_m} \Psi_m (\Phi_m (X)) f (x(\Phi_m (X) , t ) , t ) \det (\nabla_{\xi} \hat{x}) \sqrt{\det ({ }^t \nabla_{X} \Phi_m \nabla_{X} \Phi_m )} { \ } d X\\
= \sum_{m=1}^N \int_{U_m} \Psi_m (\Phi_m (X)) f ( \tilde{x} (X,t) , t ) \sqrt{ J ( X , t ) } { \ } d X.
\end{multline*}
Here we used the fact that
\begin{equation*}
\left|
\begin{pmatrix}
\frac{ \partial \tilde{x}_1}{\partial X_1}\\
\frac{ \partial \tilde{x}_2}{\partial X_1}\\
\frac{ \partial \tilde{x}_3}{\partial X_1}
\end{pmatrix} \times
\begin{pmatrix}
\frac{ \partial \tilde{x}_1}{\partial X_2}\\
\frac{ \partial \tilde{x}_2}{\partial X_2}\\
\frac{ \partial \tilde{x}_3}{\partial X_2}
\end{pmatrix}
\right| = \sqrt{ {\rm{det}}\{ { }^t (\nabla_X \tilde{x} )(\nabla_X \tilde{x})\} }. 
\end{equation*}
Note that
\begin{equation*}
\begin{pmatrix}
\frac{ \partial \tilde{x}_1}{\partial X_1}\\
\frac{ \partial \tilde{x}_2}{\partial X_1}\\
\frac{ \partial \tilde{x}_3}{\partial X_1}
\end{pmatrix} \times
\begin{pmatrix}
\frac{ \partial \tilde{x}_1}{\partial X_2}\\
\frac{ \partial \tilde{x}_2}{\partial X_2}\\
\frac{ \partial \tilde{x}_3}{\partial X_2}
\end{pmatrix} =
\begin{pmatrix}
\frac{ \partial \tilde{x}_2}{\partial X_1} \frac{ \partial \tilde{x}_3}{\partial X_2} -  \frac{ \partial \tilde{x}_2}{\partial X_2} \frac{ \partial \tilde{x}_3}{\partial X_1} \\
\frac{ \partial \tilde{x}_3}{\partial X_1} \frac{ \partial \tilde{x}_1}{\partial X_2} -  \frac{ \partial \tilde{x}_3}{\partial X_2} \frac{ \partial \tilde{x}_1}{\partial X_1} \\
\frac{ \partial \tilde{x}_1}{\partial X_1} \frac{ \partial \tilde{x}_2}{\partial X_2} -  \frac{ \partial \tilde{x}_1}{\partial X_2} \frac{ \partial \tilde{x}_2}{\partial X_1} 
\end{pmatrix}
\end{equation*}
and that
\begin{multline*}
{ }^t (\nabla_X \tilde{x}) (\nabla_X \tilde{x}) \\=
\begin{pmatrix}
\frac{ \partial \tilde{x}_1}{\partial X_1}  \frac{ \partial \tilde{x}_1}{\partial X_1} + \frac{ \partial \tilde{x}_2}{\partial X_1} \frac{ \partial \tilde{x}_2}{\partial X_1} + \frac{ \partial \tilde{x}_3}{\partial X_1}  \frac{ \partial \tilde{x}_3}{\partial X_1} & \frac{ \partial \tilde{x}_1}{\partial X_1}  \frac{ \partial \tilde{x}_1}{\partial X_2} + \frac{ \partial \tilde{x}_2}{\partial X_1} \frac{ \partial \tilde{x}_2}{\partial X_2} + \frac{ \partial \tilde{x}_3}{\partial X_1}  \frac{ \partial \tilde{x}_3}{\partial X_2}   \\
\frac{ \partial \tilde{x}_1}{\partial X_2}  \frac{ \partial \tilde{x}_1}{\partial X_1} + \frac{ \partial \tilde{x}_2}{\partial X_2} \frac{ \partial \tilde{x}_2}{\partial X_1} + \frac{ \partial \tilde{x}_3}{\partial X_2}  \frac{ \partial \tilde{x}_3}{\partial X_1} & \frac{ \partial \tilde{x}_1}{\partial X_2}  \frac{ \partial \tilde{x}_1}{\partial X_2} + \frac{ \partial \tilde{x}_2}{\partial X_2} \frac{ \partial \tilde{x}_2}{\partial X_2} + \frac{ \partial \tilde{x}_3}{\partial X_2}  \frac{ \partial \tilde{x}_3}{\partial X_2}   
\end{pmatrix} .
\end{multline*}
Thus, we have \eqref{eq31}. Remark that $U_m$, $\Psi_m$, $\Phi_m$, $J$ are independent of $f$.

Next we introduce the Riemannian metric induced by the flow map $\tilde{x} (X,t)$. For the flow map $\tilde{x} = \tilde{x} (X , t)$ on $\Gamma (t)$,
\begin{equation*}
g_\alpha =g_\alpha (X,t) := { }^t \left(\frac{\partial \tilde{x}_1}{\partial X_\alpha},\frac{\partial \tilde{x}_2}{\partial X_\alpha},\frac{\partial \tilde{x}_3}{\partial X_\alpha} \right).
\end{equation*}
Write
\begin{equation*}
g_{\alpha \beta}=g_{\alpha \beta} (X,t) := g_\alpha \cdot g_\beta = \frac{\partial \tilde{x}_i}{\partial X_\alpha} \frac{\partial \tilde{x}_i}{\partial X_\beta}= \sum_{i = 1}^3 \frac{\partial \tilde{x}_i}{\partial X_\alpha} \frac{\partial \tilde{x}_i}{\partial X_\beta}.
\end{equation*}
Set
\begin{align*}
&( g^{\alpha \beta} )_{2 \times 2} := ((g_{\alpha \beta})_{2 \times 2} )^{-1}, \text{that is, } \begin{pmatrix}
g^{11} & g^{12}\\
g^{21} & g^{22}
\end{pmatrix} := \begin{pmatrix}
g_{11} & g_{12}\\
g_{21} & g_{22}
\end{pmatrix}^{-1},\\
&g^\alpha := g^{\alpha \beta}g_\beta = g^{\alpha 1}g_1 + g^{\alpha 2}g_2,\\
&\acute{g}_\alpha := \frac{d}{d t} g_\alpha  = \frac{\partial v}{\partial X_\alpha} = { }^t \left( \frac{\partial v_1}{\partial X_\alpha} , \frac{\partial v_2}{\partial X_\alpha} ,  \frac{\partial v_3}{\partial X_\alpha} \right).
\end{align*}
It is easy to check that $g_{\beta \alpha} = g_{ \alpha \beta}$, $g^{\beta \alpha} = g^{\alpha \beta}$, $g^{\alpha \beta} = g^\alpha \cdot g^\beta$, $g^\alpha \cdot g_\beta = \delta_{\alpha \beta}$,
\begin{align*}
&g_\alpha = g_{\alpha \beta}g^\beta = g_{\alpha 1}g^1 + g_{\alpha 2}g^2,\\
&\acute{g}_{\alpha \beta} = \acute{g}_\alpha \cdot g_\beta + g_\alpha \cdot \acute{g}_\beta,\\
&\acute{g}_\alpha = \frac{\partial v}{\partial X_\alpha} = \frac{\partial \tilde{x}_i}{\partial X_\alpha}\frac{\partial v}{\partial \tilde{x}_i}= \sum_{i= 1}^3 \frac{\partial \tilde{x}_i}{\partial X_\alpha}\frac{\partial v}{\partial \tilde{x}_i},\\
&\acute{g}_\alpha \cdot g_\beta = \frac{\partial \tilde{x}_i}{\partial X_\alpha}\frac{\partial v_j}{\partial \tilde{x}_i} \frac{\partial \tilde{x}_j}{\partial X_\beta}= \sum_{i , j = 1}^3 \frac{\partial \tilde{x}_i}{\partial X_\alpha}\frac{\partial v_j}{\partial \tilde{x}_i}\frac{\partial \tilde{x}_j}{\partial X_\beta},\\
&\acute{g}_{\alpha \beta} = \frac{\partial \tilde{x}_i}{\partial X_\alpha}\left( \frac{\partial v_j}{\partial \tilde{x}_i} + \frac{\partial v_i}{\partial \tilde{x}_j} \right)\frac{\partial \tilde{x}_j}{\partial X_\beta} = 2\sum_{i, j =1}^3 \frac{\partial \tilde{x}_i}{\partial X_\alpha} [D (v) ]_{i j} \frac{\partial \tilde{x}_j}{\partial X_\beta} ,
\end{align*}
where $\delta_{\alpha \beta}$ is Kronecker's delta. Indeed, we see at once that
\begin{multline*}
\acute{g}_\alpha \cdot g_\beta = { }^t \left( \frac{\partial v_1}{\partial X_\alpha} , \frac{\partial v_2}{\partial X_\alpha} ,  \frac{\partial v_3}{\partial X_\alpha} \right) \cdot { }^t \left(\frac{\partial \tilde{x}_1}{\partial X_\beta},\frac{\partial \tilde{x}_2}{\partial X_\beta},\frac{\partial \tilde{x}_3}{\partial X_\beta} \right)\\
= { }^t \left( \sum_{i=1}^3 \frac{\partial \tilde{x}_i}{\partial X_\alpha}\frac{\partial v_1}{\partial \tilde{x}_i} , \sum_{i=1}^3 \frac{\partial \tilde{x}_i}{\partial X_\alpha}\frac{\partial v_2}{\partial \tilde{x}_i}  ,  \sum_{i=1}^3 \frac{\partial \tilde{x}_i}{\partial X_\alpha}\frac{\partial v_3}{\partial \tilde{x}_i}  \right) \cdot { }^t \left(\frac{\partial \tilde{x}_1}{\partial X_\beta},\frac{\partial \tilde{x}_2}{\partial X_\beta},\frac{\partial \tilde{x}_3}{\partial X_\beta} \right)\\
= \sum_{i , j = 1}^3 \frac{\partial \tilde{x}_i}{\partial X_\alpha}\frac{\partial v_j}{\partial \tilde{x}_i}\frac{\partial \tilde{x}_j}{\partial X_\beta} = \frac{\partial \tilde{x}_i}{\partial X_\alpha}\frac{\partial v_j}{\partial \tilde{x}_i}\frac{\partial \tilde{x}_j}{\partial X_\beta}
\end{multline*}
and that $ g^\alpha \cdot g_\beta = (g^{\alpha 1} g_1 + g^{\alpha 2} g_2 ) \cdot g_\beta = g^{\alpha 1 } g_{1 \beta} + g^{\alpha 2} g_{2 \beta} = \delta_{\alpha \beta}$.
Remark that
\begin{equation*}
\text{(R.H.S.) of }\eqref{eq31} = \int_U \tilde{\Psi} (X) f ( \tilde{x} (X,t) , t) \sqrt{ {\rm{det}} \begin{pmatrix} g_{11} & g_{12}\\ g_{21} & g_{22} \end{pmatrix} } { \ }d X.
\end{equation*}
See \cite{Cia05} and \cite{Jos11} for differential geometry and the Riemannian manifold.

Next we introduce a flow map on a variation of the evolving surface $\Gamma (t)$.
\begin{definition}[Flow map on a variation of an evolving surface]\label{def33}{ \ }\\
Let $\Gamma (t)$ be an evolving $2$-dimensional surface in $\mathbb{R}^3$ on $[0,T)$ for some $T \in (0, \infty ]$. For $- 1 < \varepsilon <1$, let $\Gamma^\varepsilon (t)$ be a variation of $\Gamma (t)$. Let $\hat{x}^\varepsilon = { }^t ( x_1^\varepsilon , x_2^\varepsilon , x_3^\varepsilon ) \in [ C^\infty ( \mathbb{R}^4) ]^3$. We call $\hat{x}^\varepsilon = \hat{x}^\varepsilon ( \xi , t )$ a \emph{flow map} on $\Gamma^\varepsilon ( t )$ if the three properties hold:\\
$( \mathrm{i} )$ for every $\xi \in \Gamma_0 (= \Gamma ( 0 ))$
\begin{equation*}
\hat{x}^\varepsilon ( \xi , 0 ) = \xi,
\end{equation*}
$( \mathrm{ii})$ for all $\xi \in \Gamma_0$ and $0 \leq t < T$
\begin{equation*}
\hat{x}^\varepsilon ( \xi , t) \in \Gamma^\varepsilon ( t ),
\end{equation*}
$(\mathrm{iii})$ for each $0 \leq t < T$
\begin{equation*}
\hat{x}^\varepsilon ( \cdot , t): \Gamma_0 \to \Gamma^\varepsilon (t) \text{ is bijective}.
\end{equation*}
\end{definition}
Note that from the property $(\mathrm{iii})$ we can write
\begin{equation*}
\Gamma^\varepsilon (t) = \{  x \in \mathbb{R}^3;{ \ }x = \hat{x}^\varepsilon (\xi , t ) , { \ }\xi \in \Gamma_0 \} .
\end{equation*}

\begin{definition}[Velocity determined by a flow map on $\Gamma^\varepsilon (t)$]\label{def34}{ \ }\\
Let $\Gamma (t)$ be a given evolving $2$-dimensional surface in $\mathbb{R}^3$ on $[0,T)$ for some $T \in (0, \infty ]$. For $- 1 < \varepsilon <1$, let $\Gamma^\varepsilon (t)$ be a variation of $\Gamma (t)$. Let $\hat{x}^\varepsilon = \hat{x}^\varepsilon (\xi ,t)$ be a flow map on $\Gamma^\varepsilon ( t )$. Suppose that there is a smooth function $v^\varepsilon = v^\varepsilon (x , t) = { }^t (v_1^\varepsilon , v_2^\varepsilon , v_3^\varepsilon )$ such that for $\xi \in \Gamma (0)$ and $0<t <T$,
\begin{equation*}
\frac{d \hat{x}^\varepsilon}{d t} = \hat{x}^\varepsilon_t ( \xi , t ) = v^\varepsilon ( \hat{x}^\varepsilon ( \xi ,t ) ,t) .
\end{equation*}
We call the vector-valued function $v^\varepsilon$ the \emph{velocity} determined by the flow map $\hat{x}^\varepsilon (\xi ,t)$.
\end{definition}

For $- 1< \varepsilon <1 $, let $\Gamma^\varepsilon (t)$ be a variation of $\Gamma ( t )$. Let $\hat{x}^\varepsilon = \hat{x}^\varepsilon (\xi ,t)$ be a flow map on $\Gamma^\varepsilon ( t )$, and let $v^\varepsilon = v^\varepsilon (x,t)$ be the velocity determined by the flow map $\hat{x}^\varepsilon$, i.e. for $\xi \in \Gamma (0)$ and $0<t <T$,
\begin{equation*}
\begin{cases}
\frac{d \hat{x}^\varepsilon}{d t} (\xi, t) = v^\varepsilon (\hat{x}^\varepsilon ( \xi, t ) , t),\\
\hat{x}^\varepsilon (\xi, 0) = \xi .
\end{cases}
\end{equation*}
By the same way as in the previous argument, we write
\begin{equation*}
\tilde{x}^\varepsilon ( X, t) := \hat{x}^\varepsilon (\Phi_m (X),t).
\end{equation*}
We also call $\tilde{x}^\varepsilon (X,t)$ a flow map on $\Gamma^\varepsilon (t)$.
Set
\begin{equation*}
J^\varepsilon = J^\varepsilon (X,t) = {\rm{det}} \{ { }^t (\nabla_X \tilde{x}^\varepsilon ) ( \nabla_X \tilde{x}^\varepsilon ) \},
\end{equation*}
where
\begin{equation*}
\nabla_X \tilde{x}^\varepsilon = 
\begin{pmatrix}
\frac{\partial \tilde{x}^\varepsilon_1}{\partial X_1} & \frac{\partial \tilde{x}^\varepsilon_1}{\partial X_2}\\
\frac{\partial \tilde{x}^\varepsilon_2}{\partial X_1} & \frac{\partial \tilde{x}^\varepsilon_2}{\partial X_2}\\
\frac{\partial \tilde{x}^\varepsilon_3}{\partial X_1} & \frac{\partial \tilde{x}^\varepsilon_3}{\partial X_2}
\end{pmatrix}.
\end{equation*}
We assume that $J^\varepsilon >0$. We see at once that for $f \in C (\mathbb{R}^4)$
\begin{equation}\label{eq32}
\int_{\Gamma^\varepsilon (t)} f (x ,t) { \ } d \mathcal{H}_{x}^2 = \int_U \tilde{\Psi} (X) f ( \tilde{x}^\varepsilon (X,t) , t) \sqrt{ J^\varepsilon (X,t) } { \ }d X.
\end{equation}
Here
\begin{multline*}
\int_U \tilde{\Psi} (X) f ( \tilde{x}^\varepsilon (X,t), t ) \sqrt{J^\varepsilon (X,t)} { \ } d X\\
 := \sum_{m=1}^N \int_{U_m} \Psi_m (\Phi_m (X)) f ( \tilde{x}^\varepsilon (X,t) , t ) \sqrt{J^\varepsilon (X,t)} { \ } d X.
\end{multline*}

For each flow map $\tilde{x}^\varepsilon = \tilde{x}^\varepsilon (X , t)$ on $\Gamma^\varepsilon (t)$,
\begin{equation*}
g_\alpha^\varepsilon := { }^t \left(\frac{\partial \tilde{x}^\varepsilon_1}{\partial X_\alpha},\frac{\partial \tilde{x}^\varepsilon_2}{\partial X_\alpha},\frac{\partial \tilde{x}^\varepsilon_3}{\partial X_\alpha} \right).
\end{equation*}
Write
\begin{equation*}
g^\varepsilon_{\alpha \beta}:= g_\alpha^\varepsilon \cdot g_\beta^\varepsilon = \frac{\partial \tilde{x}^\varepsilon_i}{\partial X_\alpha} \frac{\partial \tilde{x}^\varepsilon_i}{\partial X_\beta}= \sum_{i=1}^3 \frac{\partial \tilde{x}^\varepsilon_i}{\partial X_\alpha} \frac{\partial \tilde{x}^\varepsilon_i}{\partial X_\beta}.
\end{equation*}
Set
\begin{align*}
&( g^{\alpha \beta}_\varepsilon )_{2 \times 2} := ((g_{\alpha \beta}^\varepsilon)_{2 \times 2} )^{-1}, \text{that is, } \begin{pmatrix}
g^{11}_\varepsilon & g^{12}_\varepsilon\\
g^{21}_\varepsilon & g^{22}_\varepsilon
\end{pmatrix} := \begin{pmatrix}
g_{11}^\varepsilon & g_{12}^\varepsilon\\
g_{21}^\varepsilon & g_{22}^\varepsilon
\end{pmatrix}^{-1},\\
&g^\alpha_\varepsilon := g_\varepsilon^{\alpha \beta}g_\beta^\varepsilon = g_\varepsilon^{\alpha 1}g_1^\varepsilon + g_\varepsilon^{\alpha 2}g_2^\varepsilon ,\\
&\acute{g}^\varepsilon_\alpha := \frac{d}{d t} g^\varepsilon_\alpha  = \frac{\partial v^\varepsilon}{\partial X_\alpha} = { }^t \left( \frac{\partial v^\varepsilon_1}{\partial X_\alpha} , \frac{\partial v^\varepsilon_2}{\partial X_\alpha} ,  \frac{\partial v^\varepsilon_3}{\partial X_\alpha} \right).
\end{align*}
It is clear that $g^\varepsilon_{\beta \alpha} = g^\varepsilon_{\alpha \beta}$, $g_\varepsilon^{\beta \alpha} = g_\varepsilon^{\alpha \beta}$, $g^{\alpha \beta}_\varepsilon = g^\alpha_\varepsilon \cdot g^\beta_\varepsilon$, $g^\alpha_\varepsilon \cdot g_\beta^\varepsilon = \delta_{\alpha \beta}$,
\begin{align*}
&g_\alpha^\varepsilon = g_{\alpha \beta}^\varepsilon g^\beta_\varepsilon = g_{\alpha 1}^\varepsilon g^1_\varepsilon + g_{\alpha 2}^\varepsilon g^2_\varepsilon ,\\
&\acute{g}_{\alpha \beta}^\varepsilon = \acute{g}_\alpha^\varepsilon \cdot g_\beta^\varepsilon + g_\alpha^\varepsilon \cdot \acute{g}_\beta^\varepsilon,\\
&\acute{g}_\alpha^\varepsilon = \frac{\partial v^\varepsilon}{\partial X_\alpha} = \frac{\partial \tilde{x}^\varepsilon_i}{\partial X_\alpha}\frac{\partial v^\varepsilon}{\partial \tilde{x}^\varepsilon_i}= \sum_{i=1}^3 \frac{\partial \tilde{x}^\varepsilon_i}{\partial X_\alpha}\frac{\partial v^\varepsilon}{\partial \tilde{x}^\varepsilon_i},\\
&\acute{g}_\alpha^\varepsilon \cdot g_\beta^\varepsilon = \frac{\partial \tilde{x}^\varepsilon_i}{\partial X_\alpha}\frac{\partial v^\varepsilon_j}{\partial \tilde{x}^\varepsilon_i}\frac{\partial \tilde{x}^\varepsilon_j}{\partial X_\beta}= \sum_{i , j = 1}^3 \frac{\partial \tilde{x}^\varepsilon_i}{\partial X_\alpha}\frac{\partial v^\varepsilon_j}{\partial \tilde{x}^\varepsilon_i}\frac{\partial \tilde{x}^\varepsilon_j}{\partial X_\beta},\\
&\acute{g}_{\alpha \beta}^\varepsilon = \frac{\partial \tilde{x}^\varepsilon_i}{\partial X_\alpha}\left( \frac{\partial v^\varepsilon_j}{\partial \tilde{x}^\varepsilon_i} + \frac{\partial v^\varepsilon_i}{\partial \tilde{x}^\varepsilon_j} \right)\frac{\partial \tilde{x}^\varepsilon_j}{\partial X_\beta} = 2\sum_{i, j = 1}^3 \frac{\partial \tilde{x}^\varepsilon_i}{\partial X_\alpha} [D (v^\varepsilon) ]_{i j} \frac{\partial \tilde{x}^\varepsilon_j}{\partial X_\beta},
\end{align*}
where $\delta_{\alpha \beta}$ is Kronecker's delta. Remark that
\begin{equation*}
\text{(R.H.S.) of }\eqref{eq32} = \int_U \tilde{\Psi} (X) f ( \tilde{x}^\varepsilon (X,t) , t) \sqrt{ {\rm{det}} \begin{pmatrix} g_{11}^\varepsilon & g_{12}^\varepsilon\\ g_{21}^\varepsilon & g_{22}^\varepsilon \end{pmatrix} } { \ }d X.
\end{equation*}

Throughout Section \ref{sect3} we follow the convention:
\begin{convention}\label{conv35}
Let $\Gamma (t)$ be a given evolving $2$-dimensional surface in $\mathbb{R}^3$ on $[0,T)$ for some $T \in (0, \infty ]$. Let $\hat{x} = \hat{x} ( \xi , t )$ be a flow map on $\Gamma ( t )$, and let $v = v (x,t)$ be the velocity determined by the flow map $\hat{x} ( \xi , t )$. For $- 1< \varepsilon <1 $, let $\Gamma^\varepsilon (t)$ be a variation of $\Gamma ( t )$. Let $\hat{x}^\varepsilon = \hat{x}^\varepsilon (\xi ,t)$ be a flow map on $\Gamma^\varepsilon ( t )$, and let $v^\varepsilon = v^\varepsilon (x,t)$ be the velocity determined by the flow map $\hat{x}^\varepsilon (\xi , t)$. The symbols $g^\alpha, g_\alpha, g^{\alpha \beta}, g_{\alpha \beta}$ and $g^\alpha_\varepsilon, g_\alpha^\varepsilon, g^{\alpha \beta}_\varepsilon, g_{\alpha \beta}^\varepsilon$ are the components of the Riemannian metrics determined by the flow maps $\hat{x} ( \xi , t )$ and $\hat{x}^\varepsilon (\xi , t )$, respectively. The symbols $U$, $U_m$, $\tilde{\Psi}$, $\Psi_m$, $\Phi_m$, $J$, $J^\varepsilon$ represent the notation appearing in the argument in subsection \ref{subsec31}.
\end{convention}

\subsection{Orthogonal Projection and Surface Divergence}\label{subsec32}

Let us study an orthogonal projection and surface divergence by using the Riemannian metrics determined by flow maps. By definition, we see that
\begin{align}\label{eq33}
\begin{pmatrix}
g^{11}_\varepsilon & g^{12}_\varepsilon\\
g^{21}_\varepsilon & g^{22}_\varepsilon
\end{pmatrix} & =
\frac{1}{g_{11}^\varepsilon g_{22}^\varepsilon - g_{12}^\varepsilon g_{21}^\varepsilon}
\begin{pmatrix}
g_{22}^\varepsilon & - g_{12}^\varepsilon\\
- g_{21}^\varepsilon & g_{11}^\varepsilon
\end{pmatrix},
\end{align}
\begin{align}
& g_{11}^\varepsilon = \frac{\partial \tilde{x}^\varepsilon_1}{\partial X_1}\frac{\partial \tilde{x}^\varepsilon_1}{\partial X_1}
+ \frac{\partial \tilde{x}^\varepsilon_2}{\partial X_1}\frac{\partial \tilde{x}^\varepsilon_2}{\partial X_1}
+ \frac{\partial \tilde{x}^\varepsilon_3}{\partial X_1}\frac{\partial \tilde{x}^\varepsilon_3}{\partial X_1},\label{eq34}\\
& g_{12}^\varepsilon = g_{21}^\varepsilon = \frac{\partial \tilde{x}^\varepsilon_1}{\partial X_1}\frac{\partial \tilde{x}^\varepsilon_1}{\partial X_2}
+ \frac{\partial \tilde{x}^\varepsilon_2}{\partial X_1}\frac{\partial \tilde{x}^\varepsilon_2}{\partial X_2}
+ \frac{\partial \tilde{x}^\varepsilon_3}{\partial X_1}\frac{\partial \tilde{x}^\varepsilon_3}{\partial X_2},\label{eq35}\\
& g_{22}^\varepsilon = \frac{\partial \tilde{x}^\varepsilon_1}{\partial X_2}\frac{\partial \tilde{x}^\varepsilon_1}{\partial X_2}
+ \frac{\partial \tilde{x}^\varepsilon_2}{\partial X_2}\frac{\partial \tilde{x}^\varepsilon_2}{\partial X_2}
+ \frac{\partial \tilde{x}^\varepsilon_3}{\partial X_2}\frac{\partial \tilde{x}^\varepsilon_3}{\partial X_2}\label{eq36}.
\end{align}
The above equalities are still valid without $\varepsilon$. We also see that for each $i,j=1,2,3,$
\begin{multline}\label{eq37}
\frac{\partial \tilde{x}^\varepsilon_i}{\partial X_\alpha} \frac{\partial \tilde{x}^\varepsilon_j}{\partial X_\beta} g^{\alpha \beta}_\varepsilon \\
= \frac{\partial \tilde{x}^\varepsilon_i}{\partial X_1} \frac{\partial \tilde{x}^\varepsilon_j}{\partial X_1} g^{11}_\varepsilon + \frac{\partial \tilde{x}^\varepsilon_i}{\partial X_1} \frac{\partial \tilde{x}^\varepsilon_j}{\partial X_2} g^{12}_\varepsilon +\frac{\partial \tilde{x}^\varepsilon_i}{\partial X_2} \frac{\partial \tilde{x}^\varepsilon_j}{\partial X_1} g^{21}_\varepsilon + \frac{\partial \tilde{x}^\varepsilon_i}{\partial X_2} \frac{\partial \tilde{x}^\varepsilon_j}{\partial X_2} g^{22}_\varepsilon .
\end{multline}

Let us now study the projections $P_\Gamma$ and $P_{\Gamma^\varepsilon}$. Let us first recall that $[ P_\Gamma ]_{i j}= \delta_{i j} -n_i n_j$ and $[P_{\Gamma^\varepsilon}]_{i j} = \delta_{i j} - n_i^\varepsilon n_j^\varepsilon$.
\begin{lemma}[Representation of $P_\Gamma$ and $P_{\Gamma^\varepsilon}$]\label{lem36}
For each $i,j=1,2,3$,
\begin{align*}
\int_{\Gamma (t)} \{ \delta_{i j} - n_i  n_j\} (x,t) { \ }d \mathcal{H}^2_x = & \int_U \tilde{\Psi} (X) \left\{ \frac{\partial \tilde{x}_i}{\partial X_\alpha} \frac{\partial \tilde{x}_j}{\partial X_\beta} g^{\alpha \beta} \right\} (X,t ) \sqrt{J (X,t)} d X,\\
\int_{\Gamma^\varepsilon (t)} \{ \delta_{i j} - n_i^\varepsilon n_j^\varepsilon \} (x,t) { \ }d \mathcal{H}^2_x = & \int_U \tilde{\Psi} (X) \left\{ \frac{\partial \tilde{x}^\varepsilon_i}{\partial X_\alpha} \frac{\partial \tilde{x}^\varepsilon_j}{\partial X_\beta} g^{\alpha \beta}_\varepsilon \right\} (X,t) \sqrt{J^\varepsilon (X,t)} d X.
\end{align*}
Here
\begin{equation*}
 \int_U \tilde{\Psi} (X) f ( X ,t  ){ \ } d X = \sum_{m=1}^N \int_{U_m} \Psi_m (\Phi_m (X)) f ( X ,t  ) { \ } d X .
\end{equation*}
\end{lemma}

\begin{proof}[Proof of Lemma \ref{lem36}]
Since $n^\varepsilon$ is the unit outer normal vector of $\Gamma^\varepsilon(t)$, we see that
\begin{align}
n^\varepsilon = \begin{pmatrix} 
n_1^\varepsilon\\
n_2^\varepsilon\\
n_3^\varepsilon 
\end{pmatrix}
=& \pm \frac{ g_1^\varepsilon \times g_2^\varepsilon }{|g_1^\varepsilon \times g_2^\varepsilon |} \notag \\
=& \pm \frac{1}{\sqrt{ g_{11}^\varepsilon g_{22}^\varepsilon - g_{12}^\varepsilon g_{21}^\varepsilon }} \begin{pmatrix}
\frac{ \partial \tilde{x}^\varepsilon_2}{\partial X_1} \frac{ \partial \tilde{x}^\varepsilon_3}{\partial X_2} -  \frac{ \partial \tilde{x}^\varepsilon_2}{\partial X_2} \frac{ \partial \tilde{x}^\varepsilon_3}{\partial X_1} \\
\frac{ \partial \tilde{x}^\varepsilon_3}{\partial X_1} \frac{ \partial \tilde{x}^\varepsilon_1}{\partial X_2} -  \frac{ \partial \tilde{x}^\varepsilon_3}{\partial X_2} \frac{ \partial \tilde{x}^\varepsilon_1}{\partial X_1} \\
\frac{ \partial \tilde{x}^\varepsilon_1}{\partial X_1} \frac{ \partial \tilde{x}^\varepsilon_2}{\partial X_2} -  \frac{ \partial \tilde{x}^\varepsilon_1}{\partial X_2} \frac{ \partial \tilde{x}^\varepsilon_2}{\partial X_1} 
\end{pmatrix}\label{eq38}.
\end{align}
Note that
\begin{equation*}
\int_{\Gamma (t)} n (x,t) { \ }d \mathcal{H}^2_x = \int_U \tilde{\Psi} (X)  \left\{ \pm \frac{ g_1^\varepsilon \times g_2^\varepsilon }{|g_1^\varepsilon \times g_2^\varepsilon |} \right\}  (X,t ) \sqrt{J (X,t)} d X.
\end{equation*}

Here we only show that 
\begin{align}
1 - n_1^\varepsilon n_1^\varepsilon = \frac{\partial \tilde{x}^\varepsilon_1}{\partial X_\alpha} \frac{\partial \tilde{x}^\varepsilon_1}{\partial X_\beta} g^{\alpha \beta}_\varepsilon,\label{eq39}\\
 - n_1^\varepsilon n_2^\varepsilon = \frac{\partial \tilde{x}^\varepsilon_1}{\partial X_\alpha} \frac{\partial \tilde{x}^\varepsilon_2}{\partial X_\beta} g^{\alpha \beta}_\varepsilon \label{eq310}.
\end{align}

We first attack \eqref{eq39}. From \eqref{eq38}, we have
\begin{equation}\label{eq311}
n_1^\varepsilon n_1^\varepsilon = \frac{1}{g_{11}^\varepsilon g_{22}^\varepsilon - g_{12}^\varepsilon g_{21}^\varepsilon} \left(
\frac{ \partial \tilde{x}^\varepsilon_2}{\partial X_1} \frac{ \partial \tilde{x}^\varepsilon_3}{\partial X_2} -  \frac{ \partial \tilde{x}^\varepsilon_2}{\partial X_2} \frac{ \partial \tilde{x}^\varepsilon_3}{\partial X_1}  \right)^2 .
\end{equation}
Combining \eqref{eq37} and \eqref{eq33}, we find that
\begin{multline}\label{eq312}
\frac{\partial \tilde{x}^\varepsilon_1}{\partial X_\alpha} \frac{\partial \tilde{x}^\varepsilon_1}{\partial X_\beta} g^{\alpha \beta}_\varepsilon = \frac{\partial \tilde{x}^\varepsilon_1}{\partial X_1} \frac{\partial \tilde{x}^\varepsilon_1}{\partial X_1} g^{11}_\varepsilon + 2 \frac{\partial \tilde{x}^\varepsilon_1}{\partial X_1} \frac{\partial \tilde{x}^\varepsilon_1}{\partial X_2} g^{12}_\varepsilon + \frac{\partial \tilde{x}^\varepsilon_1}{\partial X_2} \frac{\partial \tilde{x}^\varepsilon_1}{\partial X_2} g^{22}_\varepsilon\\
= \frac{1}{g_{11}^\varepsilon g_{22}^\varepsilon - g_{12}^\varepsilon g_{21}^\varepsilon} \left( \frac{\partial \tilde{x}^\varepsilon_1}{\partial X_1} \frac{\partial \tilde{x}^\varepsilon_1}{\partial X_1} g_{22}^\varepsilon - 2 \frac{\partial \tilde{x}^\varepsilon_1}{\partial X_1} \frac{\partial \tilde{x}^\varepsilon_1}{\partial X_2} g_{12}^\varepsilon + \frac{\partial \tilde{x}^\varepsilon_1}{\partial X_2} \frac{\partial \tilde{x}^\varepsilon_1}{\partial X_2} g_{11}^\varepsilon \right) . 
\end{multline}
Adding \eqref{eq311} and \eqref{eq312}, then using \eqref{eq34}, \eqref{eq35}, and \eqref{eq36}, we see that
\begin{align*}
n_1^\varepsilon n_1^\varepsilon + \frac{\partial \tilde{x}^\varepsilon_1}{\partial X_\alpha} \frac{\partial \tilde{x}^\varepsilon_1}{\partial X_\beta} g^{\alpha \beta}_\varepsilon  = \frac{g_{11}^\varepsilon g_{22}^\varepsilon - g_{12}^\varepsilon g_{21}^\varepsilon}{g_{11}^\varepsilon g_{22}^\varepsilon - g_{12}^\varepsilon g_{21}^\varepsilon} = 1.
\end{align*}

Next we show \eqref{eq310}. From \eqref{eq38}, we have
\begin{equation*}
n_1^\varepsilon n_2^\varepsilon = \frac{1}{g_{11}^\varepsilon g_{22}^\varepsilon - g_{12}^\varepsilon g_{21}^\varepsilon} \left(
\frac{ \partial \tilde{x}^\varepsilon_2}{\partial X_1} \frac{ \partial \tilde{x}^\varepsilon_3}{\partial X_2} -  \frac{ \partial \tilde{x}^\varepsilon_2}{\partial X_2} \frac{ \partial \tilde{x}^\varepsilon_3}{\partial X_1}  \right) \left( \frac{ \partial \tilde{x}^\varepsilon_3}{\partial X_1} \frac{ \partial \tilde{x}^\varepsilon_1}{\partial X_2} -  \frac{ \partial \tilde{x}^\varepsilon_3}{\partial X_2} \frac{ \partial \tilde{x}^\varepsilon_1}{\partial X_1} \right).
\end{equation*}
By \eqref{eq37}, \eqref{eq33}, and the fact that $g_{j i}^\varepsilon =g_{ij}^\varepsilon$, we check that
\begin{multline*}
\frac{\partial \tilde{x}^\varepsilon_1}{\partial X_\alpha} \frac{\partial \tilde{x}^\varepsilon_2}{\partial X_\beta} g^{\alpha \beta}_\varepsilon = \frac{\partial \tilde{x}^\varepsilon_1}{\partial X_1} \frac{\partial \tilde{x}^\varepsilon_2}{\partial X_1} g^{11}_\varepsilon + \frac{\partial \tilde{x}^\varepsilon_1}{\partial X_1} \frac{\partial \tilde{x}^\varepsilon_2}{\partial X_2} g^{12}_\varepsilon +\frac{\partial \tilde{x}^\varepsilon_1}{\partial X_2} \frac{\partial \tilde{x}^\varepsilon_2}{\partial X_1} g^{21}_\varepsilon + \frac{\partial \tilde{x}^\varepsilon_1}{\partial X_2} \frac{\partial \tilde{x}^\varepsilon_2}{\partial X_2} g^{22}_\varepsilon\\
= \frac{1}{g_{11}^\varepsilon g_{22}^\varepsilon - g_{12}^\varepsilon g_{21}^\varepsilon} \bigg( \frac{\partial \tilde{x}^\varepsilon_1}{\partial X_1} \frac{\partial \tilde{x}^\varepsilon_2}{\partial X_1} g_{22}^\varepsilon - \frac{\partial \tilde{x}^\varepsilon_1}{\partial X_1} \frac{\partial \tilde{x}^\varepsilon_2}{\partial X_2} g_{12}^\varepsilon\\ - \frac{\partial \tilde{x}^\varepsilon_1}{\partial X_2} \frac{\partial \tilde{x}^\varepsilon_2}{\partial X_1} g_{12}^\varepsilon + \frac{\partial \tilde{x}^\varepsilon_1}{\partial X_2} \frac{\partial \tilde{x}^\varepsilon_2}{\partial X_2} g_{11}^\varepsilon \bigg).
\end{multline*}
A direct calculation with \eqref{eq34}, \eqref{eq35}, and \eqref{eq36} shows that
\begin{align*}
n_1^\varepsilon n_2^\varepsilon + \frac{\partial \tilde{x}^\varepsilon_1}{\partial X_\alpha} \frac{\partial \tilde{x}^\varepsilon_2}{\partial X_\beta} g^{\alpha \beta}_\varepsilon  = 0.
\end{align*}

In the same manner, we see that for each $i,j=1,2,3$,
\begin{align*}
& n_i^\varepsilon n_j^\varepsilon + \frac{\partial \tilde{x}^\varepsilon_i}{\partial X_\alpha} \frac{\partial \tilde{x}^\varepsilon_j}{\partial X_\beta} g^{\alpha \beta}_\varepsilon  = \delta_{i j},\\
& n_i n_j + \frac{\partial \tilde{x}_i}{\partial X_\alpha} \frac{\partial \tilde{x}_j}{\partial X_\beta} g^{\alpha \beta} = \delta_{i j}.
\end{align*}
Therefore the lemma follows.
 \end{proof}
\noindent See \cite{KLG} for another proof of Lemma \ref{lem36}.

Using flow maps and the Riemannian metrics, we study surface divergence.
\begin{lemma}\label{lem37}
For each fixed $\Omega_0 \subset \Gamma_0$ and $0<t<T$
\begin{align*}
\Omega (t) :=&  \{ x \in \mathbb{R}^3; { \ }x = \hat{x} ( \xi , t ),{ \ } \xi \in \Omega_0 \},\\
\Omega^\varepsilon (t) := & \{ x \in \mathbb{R}^3; { \ }x = \hat{x}^\varepsilon (\xi , t ),{ \ } \xi \in \Omega_0 \},
\end{align*}
where $x(\xi ,t)$ and $\hat{x}^\varepsilon (\xi ,t)$ are two flow maps on $\Gamma (t)$ and $\Gamma^\varepsilon (t)$, respectively. Then the following two assertions hold:\\
\noindent $ ( \mathrm{i})$ For every $f = { }^t (f_1,f_2,f_3) \in C^1(\mathbb{R}^3 \times \mathbb{R})$,
\begin{multline}
\int_{\Omega (t)} \{ {\rm{div}}_{\Gamma} f \}(x,t) { \ }d \mathcal{H}^2_x \\
= \int_{U} 1_{\Omega_0}(\Phi (X)) \tilde{\Psi} (X) \left\{ g^\alpha \cdot \frac{\partial f}{\partial X_\alpha} \sqrt{J} \right\}(X,t){ \ } d X, \label{eq313}
\end{multline}
\begin{multline}
\int_{\Omega^\varepsilon (t)} \{ {\rm{div}}_{\Gamma^\varepsilon} f \} (x,t) { \ }d \mathcal{H}^2_x \\
= \int_{U} 1_{\Omega_0}(\Phi (X)) \tilde{\Psi} (X) \left\{ g^\alpha_\varepsilon \cdot \frac{\partial f}{\partial X_\alpha} \sqrt{J^\varepsilon} \right\}{ \ } d X .\label{eq314}
\end{multline}
\noindent $(\mathrm{ii})$ For all $f \in C ( \mathbb{R}^3 \times \mathbb{R})$,
\begin{multline}\label{eq315}
\int_{\Omega (t)} \{ f ({\rm{div}}_{\Gamma} v ) \} (x,t) { \ }d \mathcal{H}^2_{x} \\
= \int_{U} 1_{\Omega_0}(\Phi (X)) \tilde{\Psi} (X)f (\tilde{x} (X,t),t) \left( \frac{\partial }{\partial t} \sqrt{J (X,t)} \right) { \ } d X,
\end{multline}
\begin{multline}\label{eq316}
\int_{\Omega^\varepsilon (t)} \{ f ( {\rm{div}}_{\Gamma^\varepsilon} v^\varepsilon ) \} (x,t) { \ }d \mathcal{H}^2_{x} \\
= \int_{U} 1_{\Omega_0}(\Phi (X)) \tilde{\Psi} (X)f (\tilde{x}^\varepsilon(X,t) , t) \left( \frac{\partial }{\partial t} \sqrt{J^\varepsilon (X,t)} \right) { \ } d X,
\end{multline}
Here
\begin{equation*}
1_{\Omega_0} ( \xi ) := 
\begin{cases}
1, { \ }\xi \in \Omega_0 , \\
0, { \ }\xi \in \mathbb{R}^3 \setminus \Omega_0 .
\end{cases}
\end{equation*}
\end{lemma}

\begin{proof}[Proof of Lemma \ref{lem37}]
We first show $\eqref{eq314}$. Fix $\Omega_0 \subset \Gamma_0$ and $0<t<T$. By change of variables, we check that 
\begin{multline*}
\int_{\Omega^\varepsilon (t)} \{ {\rm{div}}_{\Gamma^\varepsilon} f \} (x,t) { \ }d \mathcal{H}_{x}^2 = \int_{\Omega_0} {\rm{div}}_{\Gamma^\varepsilon} f (\hat{x}^\varepsilon (\xi , t),t) \det (\nabla_{\xi} \hat{x}^\varepsilon) { \ }d \mathcal{H}_{\xi}^2\\
= \int_{\Gamma_0} 1_{\Omega_0}( \xi ) {\rm{div}}_{\Gamma^\varepsilon} f (\hat{x}^\varepsilon (\xi , t),t) \det (\nabla_{\xi} \hat{x}^\varepsilon) { \ }d \mathcal{H}_{\xi}^2\\
= \int_U 1_{\Omega_0} (\Phi (X)) \tilde{\Psi}(X) {\rm{div}}_{\Gamma^\varepsilon} f (\tilde{x}^\varepsilon( X , t ),t) \sqrt{J^\varepsilon (X,t)} { \ }d X.
\end{multline*}
A direct calculation shows that
\begin{multline*}
g^\alpha_\varepsilon \cdot \frac{\partial f}{\partial X_\alpha} = g^{\alpha \beta}_\varepsilon g_\beta^\varepsilon \cdot \frac{\partial f}{\partial X_\alpha} =g^{1 \beta}_\varepsilon g_\beta^\varepsilon \cdot \frac{\partial f}{\partial X_1}+g^{2 \beta}_\varepsilon g_\beta^\varepsilon \cdot \frac{\partial f}{\partial X_2} \\
= \left( g^{1 1}_\varepsilon g_1^\varepsilon \cdot \frac{\partial f}{\partial X_1} + g^{1 2}_\varepsilon g_2^\varepsilon \cdot \frac{\partial f}{\partial X_1} \right) + \left( g^{2 1}_\varepsilon g_1^\varepsilon \cdot \frac{\partial f}{\partial X_2}+g^{2 2}_\varepsilon g_2^\varepsilon \cdot \frac{\partial f}{\partial X_2} \right)\\
= \left( g^{11}_\varepsilon \frac{\partial \tilde{x}^\varepsilon_1}{\partial X_1} \frac{\partial f_1}{\partial X_1} + g^{12}_\varepsilon \frac{\partial \tilde{x}^\varepsilon_1}{\partial X_2} \frac{\partial f_1}{\partial X_1} + g^{21}_\varepsilon \frac{\partial \tilde{x}^\varepsilon_1}{\partial X_1} \frac{\partial f_1}{\partial X_2} + g^{22}_\varepsilon \frac{\partial \tilde{x}^\varepsilon_1}{\partial X_2} \frac{\partial f_1}{\partial X_2} \right)\\
+\left( g^{11}_\varepsilon \frac{\partial \tilde{x}^\varepsilon_2}{\partial X_1} \frac{\partial f_2}{\partial X_1} + g^{12}_\varepsilon \frac{\partial \tilde{x}^\varepsilon_2}{\partial X_2} \frac{\partial f_2}{\partial X_1} + g^{21}_\varepsilon \frac{\partial \tilde{x}^\varepsilon_2}{\partial X_1} \frac{\partial f_2}{\partial X_2} + g^{22}_\varepsilon \frac{\partial \tilde{x}^\varepsilon_2}{\partial X_2} \frac{\partial f_2}{\partial X_2} \right)\\
+\left( g^{11}_\varepsilon \frac{\partial \tilde{x}^\varepsilon_3}{\partial X_1} \frac{\partial f_3}{\partial X_1} + g^{12}_\varepsilon \frac{\partial \tilde{x}^\varepsilon_3}{\partial X_2} \frac{\partial f_3}{\partial X_1} + g^{21}_\varepsilon \frac{\partial \tilde{x}^\varepsilon_3}{\partial X_1} \frac{\partial f_3}{\partial X_2} + g^{22}_\varepsilon \frac{\partial \tilde{x}^\varepsilon_3}{\partial X_2} \frac{\partial f_3}{\partial X_2} \right) \\
=: I_1^\varepsilon + I_2^\varepsilon + I_3^\varepsilon .
\end{multline*}
Since
\begin{equation*}
\frac{\partial f_j}{\partial X_\alpha} =\frac{\partial }{\partial X_\alpha}f_j (\tilde{x}^\varepsilon (X,t) , t) = \sum_{i=1}^3 \frac{\partial \tilde{x}^\varepsilon_i}{\partial X_\alpha} \frac{\partial f_j}{\partial \tilde{x}^\varepsilon_i},
\end{equation*}
we observe that
\begin{multline*}
I_1^\varepsilon = \left( g^{11}_\varepsilon \frac{\partial \tilde{x}^\varepsilon_1}{\partial X_1} \frac{\partial \tilde{x}^\varepsilon_1}{\partial X_1} + g^{12}_\varepsilon \frac{\partial \tilde{x}^\varepsilon_1}{\partial X_1} \frac{\partial \tilde{x}^\varepsilon_1}{\partial X_2} + g^{21}_\varepsilon \frac{\partial \tilde{x}^\varepsilon_1}{\partial X_1} \frac{\partial \tilde{x}^\varepsilon_1}{\partial X_2} + g^{22}_\varepsilon \frac{\partial \tilde{x}^\varepsilon_1}{\partial X_2} \frac{\partial \tilde{x}^\varepsilon_1}{\partial X_2}  \right) \frac{\partial f_1}{\partial \tilde{x}^\varepsilon_1}\\
 + \left( g^{11}_\varepsilon \frac{\partial \tilde{x}^\varepsilon_1}{\partial X_1} \frac{\partial \tilde{x}^\varepsilon_2}{\partial X_1} + g^{12}_\varepsilon \frac{\partial \tilde{x}^\varepsilon_1}{\partial X_1} \frac{\partial \tilde{x}^\varepsilon_2}{\partial X_2} + g^{21}_\varepsilon \frac{\partial \tilde{x}^\varepsilon_1}{\partial X_1} \frac{\partial \tilde{x}^\varepsilon_2}{\partial X_2} + g^{22}_\varepsilon \frac{\partial \tilde{x}^\varepsilon_1}{\partial X_2} \frac{\partial \tilde{x}^\varepsilon_2}{\partial X_2}  \right) \frac{\partial f_1}{\partial \tilde{x}^\varepsilon_2}\\
 +\left( g^{11}_\varepsilon \frac{\partial \tilde{x}^\varepsilon_1}{\partial X_1} \frac{\partial \tilde{x}^\varepsilon_3}{\partial X_1} + g^{12}_\varepsilon \frac{\partial \tilde{x}^\varepsilon_1}{\partial X_1} \frac{\partial \tilde{x}^\varepsilon_3}{\partial X_2} + g^{21}_\varepsilon \frac{\partial \tilde{x}^\varepsilon_1}{\partial X_1} \frac{\partial \tilde{x}^\varepsilon_3}{\partial X_2} + g^{22}_\varepsilon \frac{\partial \tilde{x}^\varepsilon_1}{\partial X_2} \frac{\partial \tilde{x}^\varepsilon_3}{\partial X_2}  \right) \frac{\partial f_1}{\partial \tilde{x}^\varepsilon_3}\\
= \left( \frac{\partial \tilde{x}^\varepsilon_1}{\partial X_\alpha} \frac{\partial \tilde{x}^\varepsilon_1}{\partial X_\beta} g^{\alpha \beta}_\varepsilon \frac{\partial}{\partial \tilde{x}^\varepsilon_1} + \frac{\partial \tilde{x}^\varepsilon_1}{\partial X_\alpha} \frac{\partial \tilde{x}^\varepsilon_2}{\partial X_\beta} g^{\alpha \beta}_\varepsilon \frac{\partial}{ \partial \tilde{x}^\varepsilon_2}  + \frac{\partial \tilde{x}^\varepsilon_1}{\partial X_\alpha} \frac{\partial \tilde{x}^\varepsilon_3}{\partial X_\beta} g^{\alpha \beta}_\varepsilon \frac{\partial}{\partial \tilde{x}^\varepsilon_3} \right) f_1.
\end{multline*}
Applying Lemma \ref{lem36}, we see that 
\begin{equation*}
\int_{\Omega^\varepsilon (t)} \{ \partial_1^{\Gamma^\varepsilon} f_1 \} (x,t) { \ }d \mathcal{H}^2_x = \int_{U} 1_{\Omega_0}(\Phi (X)) \tilde{\Psi} (X) I^\varepsilon_1 (X,t) \sqrt{J^\varepsilon (X,t)}{ \ } d X .
\end{equation*}
Similarly, we check that
\begin{equation*}
\int_{\Omega^\varepsilon (t)} \{ {\rm{div}}_{\Gamma^\varepsilon} f \} (x,t) { \ }d \mathcal{H}^2_x = \int_{U} 1_{\Omega_0}(\Phi (X)) \tilde{\Psi} (X) \{ (I^\varepsilon_1 + I^\varepsilon_2 + I^\varepsilon_3 )  \sqrt{J^\varepsilon}\} (X,t){ \ } d X .
\end{equation*}
Therefore we have \eqref{eq314}. Similarly, we obtain \eqref{eq313}.

Next we derive $\eqref{eq316}$. Since $J^\varepsilon = {\rm{det}}\{{ }^t(\nabla_X \tilde{x}^\varepsilon )(\nabla_X \tilde{x}^\varepsilon )\}$, we find that
\begin{multline*}J^\varepsilon = {\rm{det}}
\begin{pmatrix}
\frac{ \partial \tilde{x}^\varepsilon_1}{\partial X_1}  \frac{ \partial \tilde{x}^\varepsilon_1}{\partial X_1} + \frac{ \partial \tilde{x}^\varepsilon_2}{\partial X_1} \frac{ \partial \tilde{x}^\varepsilon_2}{\partial X_1} + \frac{ \partial \tilde{x}^\varepsilon_3}{\partial X_1}  \frac{ \partial \tilde{x}^\varepsilon_3}{\partial X_1} & \frac{ \partial \tilde{x}^\varepsilon_1}{\partial X_1}  \frac{ \partial \tilde{x}^\varepsilon_1}{\partial X_2} + \frac{ \partial \tilde{x}^\varepsilon_2}{\partial X_1} \frac{ \partial \tilde{x}^\varepsilon_2}{\partial X_2} + \frac{ \partial \tilde{x}^\varepsilon_3}{\partial X_1}  \frac{ \partial \tilde{x}^\varepsilon_3}{\partial X_2}   \\
\frac{ \partial \tilde{x}^\varepsilon_1}{\partial X_2}  \frac{ \partial \tilde{x}^\varepsilon_1}{\partial X_1} + \frac{ \partial \tilde{x}^\varepsilon_2}{\partial X_2} \frac{ \partial \tilde{x}^\varepsilon_2}{\partial X_1} + \frac{ \partial \tilde{x}^\varepsilon_3}{\partial X_2}  \frac{ \partial \tilde{x}^\varepsilon_3}{\partial X_1} & \frac{ \partial \tilde{x}^\varepsilon_1}{\partial X_2}  \frac{ \partial \tilde{x}^\varepsilon_1}{\partial X_2} + \frac{ \partial \tilde{x}^\varepsilon_2}{\partial X_2} \frac{ \partial \tilde{x}^\varepsilon_2}{\partial X_2} + \frac{ \partial \tilde{x}^\varepsilon_3}{\partial X_2}  \frac{ \partial \tilde{x}^\varepsilon_3}{\partial X_2}   
\end{pmatrix}\\
= {\rm{det}} \begin{pmatrix}
g_1^\varepsilon \cdot g_1^\varepsilon & g_1^\varepsilon \cdot g_2^\varepsilon\\
g_2^\varepsilon \cdot g_1^\varepsilon & g_2^\varepsilon \cdot g_2^\varepsilon
\end{pmatrix}
= {\rm{det}} \begin{pmatrix}
g_{11}^\varepsilon & g_{12}^\varepsilon \\
g_{21}^\varepsilon & g_{22}^\varepsilon
\end{pmatrix}\\
 = (g_1^\varepsilon \cdot g_1^\varepsilon )  (g_2^\varepsilon \cdot g_2^\varepsilon ) - (g_1^\varepsilon \cdot g_2^\varepsilon )(g_2^\varepsilon \cdot g_1^\varepsilon ) = g^\varepsilon_{11}g^\varepsilon_{22} - g^\varepsilon_{12} g^\varepsilon_{21}.
\end{multline*}
From the fact that
\begin{equation*}
\frac{\partial}{\partial t} (g^\varepsilon_\alpha \cdot g^\varepsilon_\beta ) = g^\varepsilon_\beta \cdot \frac{\partial v^\varepsilon}{\partial X_\alpha} + g^\varepsilon_\alpha \cdot \frac{\partial v^\varepsilon}{\partial X_\beta},
\end{equation*}
we check that
\begin{multline*}
\frac{\partial}{\partial t} J^\varepsilon = 2 g^\varepsilon_{22} \left( g^\varepsilon_1 \cdot \frac{\partial v^\varepsilon}{\partial X_1}  \right) + 2 g^\varepsilon_{11} \left( g^\varepsilon_2 \cdot \frac{\partial v^\varepsilon}{\partial X_2}  \right) \\
- 2 g^\varepsilon_{12} \left( g^\varepsilon_2 \cdot \frac{\partial v^\varepsilon}{\partial X_1}  \right) - 2 g^\varepsilon_{21} \left( g^\varepsilon_1 \cdot \frac{\partial v^\varepsilon}{\partial X_2}  \right).
\end{multline*}
Dividing both sides of the above equality by $J^\varepsilon =  g^\varepsilon_{11}g^\varepsilon_{22} - g^\varepsilon_{12} g^\varepsilon_{21}$, and then using \eqref{eq33}, we have
\begin{multline*}
\frac{1}{J^\varepsilon}\frac{\partial}{\partial t} J^\varepsilon \\
= 2 g_\varepsilon^{11} \left( g^\varepsilon_1 \cdot \frac{\partial v^\varepsilon}{\partial X_1}  \right) + 2 g_\varepsilon^{22} \left( g^\varepsilon_2 \cdot \frac{\partial v^\varepsilon}{\partial X_2}  \right) + 2 g_\varepsilon^{12} \left( g^\varepsilon_2 \cdot \frac{\partial v^\varepsilon}{\partial X_1}  \right) + 2 g_\varepsilon^{21} \left( g^\varepsilon_1 \cdot \frac{\partial v^\varepsilon}{\partial X_2}  \right)\\
=2 \left( ( g_\varepsilon^{11} g_1^\varepsilon + g_\varepsilon^{12} g_2^\varepsilon) \cdot \frac{\partial v^\varepsilon}{\partial X_1}  \right) + 2 \left( ( g_\varepsilon^{21} g_1^\varepsilon + g_\varepsilon^{22} g_2^\varepsilon) \cdot \frac{\partial v^\varepsilon}{\partial X_2}  \right)\\
=2 g^1_\varepsilon \cdot \frac{\partial v^\varepsilon}{\partial X_1}  + 2 g^2_\varepsilon \cdot \frac{\partial v^\varepsilon}{\partial X_2} . 
\end{multline*}
As a result, we have
\begin{equation*}
\frac{\partial}{\partial t} J^\varepsilon = 2 \left(g_\varepsilon^\alpha \cdot \frac{\partial v^\varepsilon}{\partial X_\alpha}  \right) J^\varepsilon . 
\end{equation*}
Therefore we find that
\begin{align*}
\frac{\partial }{\partial t} \sqrt{J^\varepsilon (X,t)} = & \frac{1}{2 \sqrt{ J^\varepsilon }} \frac{\partial}{\partial t}J^\varepsilon\\
= & g_\varepsilon^\alpha \cdot \frac{\partial v^\varepsilon}{\partial X_\alpha} \sqrt{J^\varepsilon} .
\end{align*}
Applying \eqref{eq314}, we see \eqref{eq316}. Similarly, we see \eqref{eq315}. Therefore the lemma follows.
 \end{proof}

\subsection{Continuity Equation on an Evolving Surface}\label{subsec33}
In this subsection we study the continuity equation on an evolving surface and a variation of the evolving surface. We first prepare two lemmas, and then we prove Proposition \ref{prop13}.

\begin{lemma}\label{lem38}
For every $f \in C^{1,1} (\mathcal{S}^\varepsilon_T)$,
\begin{align*}
D_t^\varepsilon f & := \partial_t f + (v^\varepsilon , \nabla ) f,\\
D_t^{N^\varepsilon} f & := \partial_t f + (v^\varepsilon \cdot n^\varepsilon )(n^\varepsilon, \nabla ) f ,\\
D_t^{\Gamma^\varepsilon} f & := \partial_t f + (v^\varepsilon , \nabla_{\Gamma^\varepsilon}) f . 
\end{align*}
Then for every $f \in C^{1,1}( \mathcal{S}^\varepsilon_T)$,
\begin{align*}
D_t^{N^\varepsilon} f + {\rm{div}}_{\Gamma^\varepsilon} (f v^\varepsilon) & = D_t f + ({\rm{div}}_{\Gamma^\varepsilon} v^\varepsilon ) f,\\
D_t^{N^\varepsilon} ( f v^\varepsilon) + {\rm{div}}_{\Gamma^\varepsilon} ( f v^\varepsilon \otimes v^\varepsilon) & = \{ D_t^\varepsilon f + ({\rm{div}}_{\Gamma^\varepsilon} v^\varepsilon ) f \} v^\varepsilon + f D_t^\varepsilon v^\varepsilon.
\end{align*}
\end{lemma}
\noindent Since we can prove Lemma \ref{lem38} by using the same argument in the proof of Lemma \ref{lem27}, the proof of Lemma \ref{lem38} is left to the reader.

\begin{lemma}\label{lem39}
For each fixed $\Omega_0 \subset \Gamma_0$ and $0<t<T$,
\begin{align*}
\Omega (t) := & \{ x \in \mathbb{R}^3; { \ }x = \hat{x} ( \xi , t ),{ \ } \xi \in \Omega_0 \},\\
\Omega^\varepsilon (t) := & \{ x \in \mathbb{R}^3; { \ }x = \hat{x}^\varepsilon (\xi , t ),{ \ } \xi \in \Omega_0 \},
\end{align*}
where $\hat{x}^\varepsilon$ and $\hat{x}^\varepsilon$ are two flow maps on $\Gamma (t)$ and $\Gamma^\varepsilon (t)$, respectively. Then for each $\Omega_0 \subset \Gamma_0$, $0<t<T$, and $f \in C (\mathbb{R}^3 \times \mathbb{R} )$,
\begin{align}
\frac{d }{d t}\int_{\Omega (t)} f (x,t) { \ }d \mathcal{H}_{x}^2 = \int_{\Omega (t)} \{ D_t f + ( {\rm{div}}_{\Gamma} v) f \} (x,t){ \ } d \mathcal{H}_{x}^2,\label{eq317}\\
\frac{d }{d t}\int_{\Omega^\varepsilon (t)} f (x,t) { \ }d \mathcal{H}_{x}^2 = \int_{\Omega^\varepsilon (t)} \{ D_t^\varepsilon f + ( {\rm{div}}_{\Gamma^\varepsilon} v^\varepsilon ) f \}{ \ } d \mathcal{H}_{x}^2.\label{eq318}
\end{align}
\end{lemma}

\begin{proof}[Proof of Lemma \ref{lem39}]
We first derive \eqref{eq318}. Fix $\Omega_0 \subset \Gamma_0$, $0<t<T$ and $f \in C(\mathbb{R}^3 \times \mathbb{R} )$. By definition, we observe that
\begin{multline*}
\int_{\Omega^\varepsilon (t)} f (x,t) { \ }d \mathcal{H}_{x}^2 = \int_{\Omega_0} f (\hat{x}^\varepsilon (\xi , t),t) \det (\nabla_{\xi} \hat{x}^\varepsilon) { \ }d \mathcal{H}_{\xi}^2\\
= \int_{\Gamma_0} 1_{\Omega_0}( \xi ) f (\hat{x}^\varepsilon (\xi , t),t) \det ( \nabla_{\xi} \hat{x}^\varepsilon) { \ }d \mathcal{H}_{\xi}^2\\
= \int_U 1_{\Omega_0} (\Phi (X)) \tilde{\Psi}(X) f (\tilde{x}^\varepsilon( X , t ),t) \sqrt{J^\varepsilon (X,t)} { \ }d X.
\end{multline*}
Here
\begin{equation*}
1_{\Omega_0} ( \xi ) = 
\begin{cases}
1, { \ }\xi \in \Omega_0 , \\
0, { \ }\xi \in \mathbb{R}^3 \setminus \Omega_0 .
\end{cases}
\end{equation*}
Using the assertion $(\mathrm{ii})$ of Lemma \ref{lem37}, we check that
\begin{multline*}
\frac{d }{d t}\int_{\Omega^\varepsilon (t)} f (x,t) { \ }d \mathcal{H}_{x}^2 \\
= \int_U 1_{\Omega_0} (\Phi (X)) \tilde{\Psi}(X) \left( \frac{d}{d t} f (\tilde{x}^\varepsilon( X , t ),t) \right) \sqrt{J^\varepsilon (X,t)} { \ }d X\\
+ \int_U 1_{\Omega_0} (\Phi (X)) \tilde{\Psi}(X) f (\tilde{x}^\varepsilon( X , t ),t) \left( \frac{ \partial }{\partial t} \sqrt{J^\varepsilon (X,t)} \right) { \ }d X\\
= \int_{\Omega^\varepsilon (t)} \{ D_t^\varepsilon f + ({\rm{div}}_{\Gamma^\varepsilon} v^\varepsilon ) f  \} (x,t){ \ }d \mathcal{H}_{x}^2.
\end{multline*}
Therefore we have \eqref{eq318}. Similarly, we see \eqref{eq317}.
 \end{proof}
Lemma \ref{lem39} attacks Proposition \ref{prop13}.
\begin{proof}[Proof of Proposition \ref{prop13}]
We only prove the assertion $(\mathrm{ii})$ of Proposition \ref{prop13}. Fix $t \in (0,T)$ and $\Omega^\varepsilon ( t) \subset \Gamma^\varepsilon (t)$. Since the flow map $\hat{x}^\varepsilon (\xi, t)$ is bijective, there is $\Omega_0 \subset \Gamma_0$ such that
\begin{equation*}
\Omega^\varepsilon (t) = \{ x \in \mathbb{R}^3; { \ }x = \hat{x}^\varepsilon (\xi , t ),{ \ } \xi \in \Omega_0 \} .
\end{equation*}
From Lemma \ref{lem39}, we see that
\begin{equation*}
\frac{d }{d t}\int_{\Omega^\varepsilon (t)} \rho^\varepsilon (x,t) { \ }d \mathcal{H}_{x}^2 = \int_{\Omega^\varepsilon (t)} \{ D_t^\varepsilon \rho^\varepsilon + ( {\rm{div}}_{\Gamma^\varepsilon} v^\varepsilon ) \rho^\varepsilon \} (x,t) { \ } d \mathcal{H}^2_{x}.
\end{equation*}
Since $\Omega^\varepsilon (t)$ is arbitrary, we conclude that
\begin{equation*}
D_t^\varepsilon \rho^\varepsilon + ({\rm{div}}_{\Gamma^\varepsilon} v^\varepsilon) \rho^\varepsilon = 0 \text{ on } \mathcal{S}^\varepsilon_T .
\end{equation*}
Therefore we have the continuity equation for the fluid on $\Gamma^\varepsilon (t)$. Similarly, we can prove Proposition \ref{prop13}.
 \end{proof}
See also \cite{Bet86}, \cite{GSW89}, \cite{DE07}, \cite{KLG}, and \cite{SSO07} for the proof of Proposition \ref{prop13}.

\subsection{Representation of Density, Kinetic Energy, Dissipation Energies, Work for the Fluid on an Evolving Surface}\label{subsec34}

In this subsection we study the representation of the kinetic energy, dissipation energies, and work for the fluid on an evolving surface by applying the Riemannian metric induced by a flow map. From now we follow the convention:
\begin{convention}\label{conv310}
Assume that $\rho = \rho (x,t)$, $\rho^\varepsilon = \rho^\varepsilon (x,t)$, $\rho_0 = \rho_0 (x)$, $\sigma = \sigma (x,t)$, $e = e (x,t)$, $\theta = \theta (x,t)$, $h = h (x,t)$, $s = s (x,t)$, $e_F = e_F (x,t)$, $F = F (x,t) = { }^t (F_1 , F_2, F_3)$, $\mu = \mu (x,t)$, $\lambda = \lambda (x,t)$, $\kappa = \kappa (x,t)$, $\nu = \nu (x,t)$, $C = C (x,t)$, $C_\theta = C_\theta (x,t)$, $Q_\theta = Q_\theta (x,t)$, $Q_C = Q_C (x,t)$, $\mathcal{F}_1 = \mathcal{F}_1(x,t) $, $\mathcal{F}_2 = \mathcal{F}_2 (x,t)$ are smooth functions. Moreover, $p \in C^1 ( (0, \infty ))$ or $p \in C^1 ( [0, \infty ))$.
\end{convention}

Based on Proposition \ref{prop13}, we assume that $ \rho $ and $\rho^\varepsilon$ satisfies
\begin{equation*}
\begin{cases}
D_t \rho + ({\rm{div}}_\Gamma v) \rho =0 & \text{ on } \mathcal{S}_T,\\
\rho |_{t = 0} = \rho_0 & \text{ on }\Gamma (0),
\end{cases}
\end{equation*}
\begin{equation*}
\begin{cases}
D_t^\varepsilon \rho^\varepsilon + ({\rm{div}}_{\Gamma^\varepsilon} v^\varepsilon ) \rho^\varepsilon =0 & \text{ on } \mathcal{S}^\varepsilon_T,\\
\rho^\varepsilon |_{t = 0} = \rho_0 & \text{ on }\Gamma (0) .
\end{cases}
\end{equation*}

The aim of this subsection is to prove the following three propositions.
\begin{proposition}[Representation of Energy densities $(\mathrm{I})$]\label{prop311}
Set
\begin{align*}
\tilde{\rho}_0 (X) = & \rho_0 (\tilde{x}(X,0)) \sqrt{J(X,0)} ,\\
\mathcal{K}(e_K)  =& \frac{1}{2} \frac{\tilde{\rho}_0 (X) }{\sqrt{J (X,t)}} \tilde{x}_t (X, t) \cdot \tilde{x}_t (X,t) ,\\
\mathcal{K}(e_A) =& \frac{\tilde{\rho}_0 (X) }{\sqrt{J (X,t)}} \left\{ \frac{1}{2} \tilde{x}_t (X, t) \cdot \tilde{x}_t (X,t) + e (\tilde{x} (X,t), t) \right\},\\
\mathcal{K}(e_B) =&\frac{1}{2}  \frac{\tilde{\rho}_0 (X) }{\sqrt{J (X,t)}} \{ \tilde{x}_t (X, t) \cdot \tilde{x}_t (X,t)\} - p \left( \frac{\tilde{\rho}_0 (X) }{\sqrt{J (X,t)}} \right),\\
\mathcal{K}(e_{W_2}) =&\frac{\tilde{\rho}_0 (X) }{\sqrt{J (X,t)}}  \tilde{x}_t (X,t) \cdot F(\tilde{x} (X,t), t) .
\end{align*}
Then
\begin{align*}
\int_{\Gamma (t)} \left\{ \frac{1}{2} \rho | v |^2 \right\} (x,t) { \ }d \mathcal{H}_{x}^2 = \int_U \tilde{\Psi} (X) \mathcal{K}(e_K) \sqrt{ J (X,t)  } { \ }d X,\\
\int_{\Gamma (t)} \left\{ \frac{1}{2} \rho | v |^2 + \rho e \right\} (x,t) { \ }d \mathcal{H}_{x}^2 = \int_U \tilde{\Psi} (X) \mathcal{K}(e_A)\sqrt{ J (X,t)  } { \ }d X,\\
\int_{\Gamma (t)} \left\{ \frac{1}{2} \rho | v |^2 - p (\rho) \right\} (x,t) { \ }d \mathcal{H}_{x}^2 = \int_U \tilde{\Psi} (X) \mathcal{K}(e_B)\sqrt{ J (X,t)  } { \ }d X,\\
\int_{\Gamma (t)} \{ \rho F \cdot v \} (x,t) { \ }d \mathcal{H}_{x}^2 = \int_U \tilde{\Psi} (X) \mathcal{K}(e_{W_2}) \sqrt{ J (X,t)  } { \ }d X.
\end{align*}
Here
\begin{equation*}
 \int_U \tilde{\Psi} (X) f ( X ,t  ){ \ } d X = \sum_{m=1}^N \int_{U_m} \Psi_m (\Phi_m (X)) f ( X ,t  ) { \ } d X .
\end{equation*}
\end{proposition}

\begin{proposition}[Representation of Energy densities $(\mathrm{II})$]\label{prop312}
Set
\begin{align*}
\mathcal{K}(e_{W_1}) = & \frac{1}{2}\sigma (\tilde{x} (X,t) , t ) (\acute{g}_{\alpha \beta} g^{\alpha \beta})(X,t),\\
\mathcal{K} (e_D) =& \frac{1}{2} \left\{ \frac{1}{2} \mu (\tilde{x} (X,t) , t) ( \acute{g}_{\alpha \beta}\acute{g}_{\zeta \eta} g^{\alpha \zeta}g^{\beta \eta} ) + \frac{1}{4} \lambda (\tilde{x} (X,t) , t) (\acute{g}_{\alpha \beta} \acute{g}_{\zeta \eta} g^{\alpha \beta} g^{\zeta \eta}) \right\} ,\\
\mathcal{K}(e_{TD}) =& \frac{1}{2} \kappa (\tilde{x}(X,t), t)\left\{ g^{\alpha \beta}  \frac{\partial \theta }{\partial X_\alpha} \frac{\partial \theta }{\partial X_\beta} \right\} (X,t),\\
\mathcal{K}(e_{SD}) =& \frac{1}{2} \nu (\tilde{x}(X,t), t) \left\{ g^{\alpha \beta} \frac{\partial C}{\partial X_\alpha} \frac{\partial C }{\partial X_\beta} \right\} (X,t).
\end{align*}
Then
\begin{align*}
\int_{\Gamma (t)} \{ \sigma ({\rm{div}_\Gamma} v )\} (x,t) { \ }d \mathcal{H}_{x}^2 = \int_U \tilde{\Psi} (X)\mathcal{K}(e_{W_1})\sqrt{ J (X,t)  } { \ }d X,\\
\int_{\Gamma (t)} \frac{1}{2} \{ 2 \mu | P_\Gamma D (v) P_\Gamma |^2 + \lambda |{\rm{div}}_\Gamma v |^2 \} { \ }d \mathcal{H}_{x}^2 = \int_U \tilde{\Psi} (X) \mathcal{K}(e_D) \sqrt{ J (X,t)  } { \ }d X,\\
\int_{\Gamma (t)} \left\{ \frac{1}{2}\kappa | {\rm{grad}}_\Gamma \theta |^2 \right\}(x,t)  { \ }d \mathcal{H}_{x}^2 = \int_U \tilde{\Psi} (X)\mathcal{K}(e_{TD}) \sqrt{ J (X,t)  } { \ }d X,\\
\int_{\Gamma (t)} \left\{ \frac{1}{2} \nu |{\rm{grad}_\Gamma } C |^2 \right\} (x,t) { \ }d \mathcal{H}_{x}^2 = \int_U \tilde{\Psi} (X)\mathcal{K}(e_{SD}) \sqrt{ J (X,t)  } { \ }d X.
\end{align*}
\end{proposition}

\begin{proposition}[Representation of Energy density $(\mathrm{III})$]\label{prop313}
Assume that $e_{\mathcal{J}} \in C^1 ((0, \infty ))$ or $e_{\mathcal{J}} \in C^1 ([0, \infty ))$. Set
\begin{equation*}
\mathcal{K}(e_{\mathcal{J}}) = \frac{1}{2} e_{\mathcal{J}} \left( \left\{ g^{\alpha \beta} \frac{\partial f}{\partial X_\alpha} \frac{\partial f}{\partial X_\beta} \right\} (X,t)  \right).
\end{equation*}
Then
\begin{equation*}
\int_{\Gamma (t)} \frac{1}{2} e_{\mathcal{J}} ( | {\rm{grad}}_\Gamma f (x,t) |^2 ) { \ } d \mathcal{H}^2_x = \int_U \tilde{\Psi}(X) \mathcal{K} (e_{\mathcal{J}}) \sqrt{J (X,t)} d X .
\end{equation*}
\end{proposition}
In order to prove the three propositions, we prepare several lemmas. Let us first study the representation of the density of the fluid on an evolving surface.
\begin{lemma}[Representation of density]\label{lem314}
Let $f_0 = f_0 (x)$, $f = f (x,t)$ and $\mathcal{F} = \mathcal{F} (x,t)$ be three smooth functions. Then two assertions hold:\\
$(\mathrm{i})$ Suppose that
\begin{equation*}
\begin{cases}
D_t f + ({ \rm{div} }_\Gamma v) f = \mathcal{F} \text{ on } \mathcal{S}_T,\\
f |_{t = 0} = f_0.
\end{cases}
\end{equation*}
Then for each $X$ and $t$,
\begin{multline*}
f (\tilde{x} (X, t ) , t )  \\
= \frac{f_0 (\tilde{x}(X,0) )\sqrt{J(X,0)}}{\sqrt{J (X,t )}} + \frac{1}{\sqrt{J (X, t)}} \int_0^t \mathcal{F} (\tilde{x}(X, \tau) , \tau) \sqrt{J (X , \tau )} { \ }d \tau .
\end{multline*}
\noindent $(\mathrm{ii})$ Suppose that
\begin{equation*}
\begin{cases}
D_t^\varepsilon f + ({ \rm{div} }_{\Gamma^\varepsilon} v^\varepsilon) f = \mathcal{F} \text{ on } \mathcal{S}^\varepsilon_T,\\
f |_{t = 0} = f_0.
\end{cases}
\end{equation*}
Then for each $X$ and $t$,
\begin{multline}\label{eq319}
f (\tilde{x}^\varepsilon (X, t ) , t ) = \frac{f_0 (\tilde{x} (X,0) )\sqrt{J(X,0)}}{\sqrt{J^\varepsilon (X,t )}}\\
 + \frac{1}{\sqrt{J^\varepsilon (X, t)}} \int_0^t \mathcal{F} (\tilde{x}^\varepsilon (X, \tau) , \tau) \sqrt{J^\varepsilon (X , \tau )} { \ }d \tau.
\end{multline}
\end{lemma}

\begin{proof}[Proof of Lemma \ref{lem314}]
We only prove $(\mathrm{ii})$. Fix $X$. Set
\begin{equation*}
G^\varepsilon (X, t) = f (\tilde{x}^\varepsilon (X, t ) , t ) \sqrt{J^\varepsilon (X,t )}.
\end{equation*}
Applying Lemma \ref{lem37} and the assumption, we check that
\begin{align*}
\frac{d G^\varepsilon }{d t} = & \left\{ \left( \frac{d \tilde{x}^\varepsilon}{d t} , \nabla_{\tilde{x}_\varepsilon} \right) f + \partial_t f \right\} \sqrt{J^\varepsilon} + f \frac{d}{d t} \sqrt{J^\varepsilon}\\
 = & \{ D_t^\varepsilon f + ({\rm{div}}_{\Gamma^\varepsilon} v^\varepsilon ) f \} \sqrt{J^\varepsilon} = \mathcal{F} \sqrt{J^\varepsilon}. 
\end{align*}
Integrating with respect to time, we have
\begin{equation*}
G^\varepsilon (t)  = G^\varepsilon (0) + \int_0^t \mathcal{F} (\tilde{x}^\varepsilon (X, \tau ) , \tau )\sqrt{J^\varepsilon (X, \tau)} { \ }d \tau.
\end{equation*}
This implies that
\begin{multline*}
f (\tilde{x}^\varepsilon (X, t ) , t ) \sqrt{J^\varepsilon (X,t )}\\
 = f_0 (\tilde{x}^\varepsilon (X,0) )\sqrt{J^\varepsilon (X,0)} + \int_0^t \mathcal{F} (\tilde{x}^\varepsilon (X, \tau) , \tau) \sqrt{J^\varepsilon (X , \tau )} { \ }d \tau.
\end{multline*}
Since $\tilde{x}^\varepsilon (X, 0) = \tilde{x} (X, 0) = \Phi (X) = \xi$, we have \eqref{eq319}. Similarly, we see $(\mathrm{i})$. Therefore the lemma follows.
 \end{proof}

Let us now study the representation of several energies for the fluid on an evolving surface. 
\begin{lemma}[Representation of energies (I)]\label{lem315}{ \ }\\
\noindent $(\mathrm{i})$ 
\begin{align}
\int_{\Gamma (t)} \rho (x,t) | v (x,t) |^2  { \ }d \mathcal{H}_{x}^2 =& \int_U \tilde{\Psi} (X) \tilde{\rho}_0 (X) | \tilde{x}_t (X,t) |^2 { \ }d X, \label{eq320}\\
\int_{\Gamma (t)} \rho (x,t) e (x,t) { \ }d \mathcal{H}_{x}^2 =& \int_U \tilde{\Psi} (X) \tilde{\rho}_0 (X) e( \tilde{x} (X,t) ,t){ \ }d X,\label{eq321}
\end{align}
\begin{multline}
\int_{\Gamma (t)} \rho (x,t) F (x,t) \cdot v (x,t) { \ }d \mathcal{H}_{x}^2 \\
= \int_U \tilde{\Psi} (X) \tilde{\rho}_0 (X) F (\tilde{x}(X,t) , t) \cdot  \tilde{x}_t (X,t) { \ }d X,
\end{multline}
\begin{equation}
\int_{\Gamma (t)} p( \rho (x,t) ) { \ }d \mathcal{H}_{x}^2 = \int_U \tilde{\Psi} (X) p \left( \frac{ \tilde{\rho}_0 (X)}{\sqrt{J (X,t)}} \right) \sqrt{J (X,t)}{ \ }d X.
\end{equation}
\noindent $(\mathrm{ii})$ 
\begin{align}
\int_{\Gamma^\varepsilon (t)} \rho^\varepsilon (x,t) | v^\varepsilon (x,t) |^2  { \ }d \mathcal{H}_{x}^2 = \int_U \tilde{\Psi} (X) \tilde{\rho}_0 (X) | \tilde{x}^\varepsilon_t (X, t) |^2 { \ }d X,\label{eq324}\\
\int_{\Gamma^\varepsilon (t)} p( \rho^\varepsilon (x,t) ) { \ }d \mathcal{H}_{x}^2 = \int_U \tilde{\Psi} (X) p \left( \frac{ \tilde{\rho}_0 (X)}{\sqrt{J^\varepsilon (X,t)}} \right) \sqrt{J^\varepsilon (X,t)}{ \ }d X.\label{eq325}
\end{align}
Here
\begin{equation*}
\tilde{\rho}_0 (X) = \rho_0 (\tilde{x}(X,0)) \sqrt{J(X,0)}.
\end{equation*}
\end{lemma}

\begin{proof}[Proof of Lemma \ref{lem315}]
We first show \eqref{eq324}. From Lemma \ref{lem314}, we check that
\begin{multline*}
\int_{\Gamma^\varepsilon (t)} \rho^\varepsilon (x,t) | v^\varepsilon ( x, t )|^2 { \ }d \mathcal{H}^2_x \\
= \int_{\Gamma (0)} \rho^\varepsilon (\hat{x}^\varepsilon (\xi , t),t) | v^\varepsilon ( \hat{x}^\varepsilon (\xi, t), t )|^2 {\rm{det}}(\nabla_\xi \hat{x}^\varepsilon) { \ }d \mathcal{H}^2_\xi\\
=\int_U \tilde{\Psi} (X) \rho^\varepsilon (\tilde{x}^\varepsilon (X,t), t) \tilde{x}_t^\varepsilon (X,t) \cdot \tilde{x}^\varepsilon_t (X,t) \sqrt{ J^\varepsilon (X,t)  } { \ }d X\\
=\int_U \tilde{\Psi} (X) \tilde{\rho}_0 (X) \left\{ \tilde{x}^\varepsilon_t (X,t) \cdot \tilde{x}^\varepsilon_t (X,t) \right\} { \ }d X.
\end{multline*}
Therefore we have \eqref{eq324}. Similarly, we see \eqref{eq320}-\eqref{eq325}.
 \end{proof}

\begin{lemma}[Representation of energies (II)]\label{lem316}{ \ }\\
Let $f = f (x,t) \in C^{1,0}(\mathcal{S}_T)$. Then
\begin{multline}
\int_{\Gamma (t)} \{ \sigma ({\rm{div}_\Gamma} v ) \} (x,t) { \ }d \mathcal{H}_{x}^2 \\
= \int_U \tilde{\Psi} (X) \sigma (\tilde{x}(X,t) , t) \left\{ \frac{1}{2}  \acute{g}_{\alpha \beta} g^{\alpha \beta} \right\} (X,t) \sqrt{ J (X,t)  } { \ }d X,\label{eq326}
\end{multline}
\begin{multline}
\int_{\Gamma (t)} \{ \lambda |{\rm{div}}_\Gamma v |^2 \} (x,t) { \ }d \mathcal{H}_{x}^2\\
 = \int_U \tilde{\Psi} (X) \lambda (\tilde{x} (X,t) , t) \left\{ \frac{1}{4} \acute{g}_{\alpha \beta}\acute{g}_{\zeta \eta} g^{\alpha \beta}g^{\zeta \eta} \right\} (X,t) \sqrt{ J (X,t)  } { \ }d X,\label{eq327}
\end{multline}
\begin{multline}\label{eq328}
\int_{\Gamma (t)} \{ \mu | P_\Gamma D (v ) P_\Gamma |^2 \} (x,t) { \ }d \mathcal{H}_{x}^2 \\
=  \int_U \tilde{\Psi} (X) \mu (\tilde{x} (X,t) , t) \left\{ \frac{1}{4}\acute{g}_{\alpha \beta}\acute{g}_{\zeta \eta} g^{\alpha \zeta}g^{\beta \eta} \right\} (X,t) \sqrt{ J (X,t)  } { \ }d X,
\end{multline}
\begin{multline}\label{eq329}
\int_{\Gamma (t)} \{ | {\rm{grad}}_\Gamma f |^2 \} (x,t) { \ }d \mathcal{H}_{x}^2 \\
=  \int_U \tilde{\Psi} (X) \left\{ g^{\alpha \beta} \frac{\partial f }{\partial X_\alpha} \frac{\partial f }{\partial X_\beta} \right\} (X,t) \sqrt{ J (X,t)  } { \ }d X.
\end{multline}

\end{lemma}

\begin{proof}[Proof of lemma \ref{lem316}] A direct calculation shows that
\begin{align*}
\acute{g}_{\alpha \beta} g^{\alpha \beta} = & (\acute{g}_\alpha \cdot g_\beta + \acute{g}_\beta \cdot g_\alpha) g^{\alpha \beta} \\
= & \acute{g}_\alpha \cdot g^\alpha + \acute{g}_\beta \cdot g^{\beta}\\
= & \frac{\partial v}{\partial X_\alpha}  \cdot g^\alpha + \frac{\partial v}{\partial X_\beta} \cdot g^{\beta} (= 2 {\rm{div}}_\Gamma v) .
\end{align*}
From Lemma \ref{lem37}, we see that
\begin{equation*}
\int_U \tilde{\Psi} (X)\sigma (\tilde{x} (X,t) ,t) \left( \frac{1}{2} \acute{g}_{\alpha \beta} g^{\alpha \beta} \right) \sqrt{ J (X,t)  } { \ }d X = \int_{\Gamma (t)} \{ \sigma ({\rm{div}_\Gamma} v )\} (x,t) { \ }d \mathcal{H}_{x}^2
\end{equation*}
and that
\begin{multline*}
\int_U \tilde{\Psi} (X)\lambda (\tilde{x} (X,t),t) \left( \frac{1}{2} \acute{g}_{\alpha \beta} g^{\alpha \beta} \right) \left( \frac{1}{2} \acute{g}_{\zeta \eta} g^{\zeta \eta} \right) \sqrt{ J (X,t)  } { \ }d X \\
= \int_{\Gamma (t)} \{ \lambda |{\rm{div}_\Gamma} v |^2 \} (x,t) { \ }d \mathcal{H}_{x}^2 .
\end{multline*}
Thus, we have \eqref{eq326} and \eqref{eq327}. Since
\begin{align*}
\frac{1}{2} \acute{g}_{\alpha \beta} = \frac{\partial \tilde{x}_i}{\partial X_\alpha} [D (v) ]_{i j} \frac{\partial \tilde{x}_j}{\partial X_\beta} \text{ and }[P_{\Gamma}]_{i j } = \frac{\partial \tilde{x}_i}{\partial X_\alpha} \frac{\partial \tilde{x}_j}{\partial X_\beta} g^{\alpha \beta},
\end{align*}
we observe that
\begin{multline*}
\frac{1}{4} \acute{g}_{\alpha \beta} \acute{g}_{\zeta \eta} g^{\alpha \zeta} g^{\beta \eta} = \frac{\partial \tilde{x}_i}{\partial X_\alpha} [D (v) ]_{i j} \frac{\partial \tilde{x}_j}{\partial X_\beta}g^{\alpha \zeta} g^{\beta \eta} \frac{\partial \tilde{x}_k}{\partial X_\zeta} [D (v) ]_{k \ell} \frac{\partial \tilde{x}_\ell}{\partial X_\eta}\\
= [D (v) ]_{i j}  \frac{\partial \tilde{x}_i}{\partial X_\alpha} \frac{\partial \tilde{x}_k}{\partial X_\zeta} g^{\alpha \zeta} [D (v) ]_{k \ell} \frac{\partial \tilde{x}_j}{\partial X_\beta} \frac{\partial \tilde{x}_\ell}{\partial X_\eta} g^{\beta \eta}\\
= [D (v)]_{i j} [P_\Gamma ]_{i k} [D (v)]_{k \ell} [P_\Gamma ]_{j \ell} = {\text{Tr}} \{ D (v) P_\Gamma D (v) P_\Gamma \}.
\end{multline*}
From the fact that $[P_\Gamma D (v) P_\Gamma ]_{i j} = \{ \partial_i^{\Gamma} v_j + \partial_j^{\Gamma} v_i - n_i (n \cdot \partial_j^{\Gamma} v ) - n_j ( n \cdot \partial_i^{\Gamma} v )\}/2$, we find that $P_\Gamma D (v) P_\Gamma$ is a symmetric matrix. Since $P_\Gamma D (v) P_\Gamma$ is a symmetric matrix and $P_\Gamma P_\Gamma = P_\Gamma$, we check that
\begin{align*}
(P_\Gamma D (v) P_\Gamma ) : (P_\Gamma D (v) P_\Gamma ) = & {\text{Tr}} \{ (P_\Gamma D (v) P_\Gamma ) (P_\Gamma D (v) P_\Gamma) \}\\
= & {\text{Tr}} \{ P_\Gamma D (v) P_\Gamma D (v) P_\Gamma \}\\
= & {\text{Tr}} \{ D (v) P_\Gamma D (v) P_\Gamma \}.
\end{align*}
Thus, we see \eqref{eq328}. Since
\begin{equation*}
\frac{\partial}{\partial X_\beta} f (\tilde{x} (X,t) , t) = \frac{\partial \tilde{x}_j }{\partial X_\beta } \frac{\partial f }{\partial \tilde{x}_j},
\end{equation*}
we apply Lemma \ref{lem36} to see that for each $i =1,2,3$,
\begin{align*}
\partial_i^\Gamma f & = \frac{\partial \tilde{x}_i}{\partial X_\alpha} \frac{\partial \tilde{x}_j}{\partial X_\beta} g^{\alpha \beta} \frac{\partial f}{\partial \tilde{x}_j}\\
& = g^{\alpha \beta}  \frac{\partial \tilde{x}_i}{\partial X_\alpha} \frac{\partial f}{\partial X_\beta}.
\end{align*}
By definition, we observe that
\begin{align*}
| {\rm{grad}}_\Gamma f|^2 = (\partial_i^\Gamma f ) (\partial_i^\Gamma f) & = \left( g^{\alpha \beta}  \frac{\partial \tilde{x}_i}{\partial X_\alpha} \frac{\partial f}{\partial X_\beta} \right) \left( g^{\zeta \eta}  \frac{\partial \tilde{x}_i}{\partial X_\zeta} \frac{\partial f}{\partial X_\eta} \right)\\
& = g^{\alpha \beta} (g^{\zeta} \cdot g^{\eta}) ( g_{\alpha} \cdot g_{\zeta} )\frac{\partial f}{\partial X_\beta} \frac{\partial f}{\partial X_\eta}\\
& = g^{\alpha \beta} (g^{\eta} \cdot g_{\alpha} )\frac{\partial f}{\partial X_\beta} \frac{\partial f}{\partial X_\eta} .
\end{align*}
Moreover, we see that
\begin{align*}
g^{\alpha \beta} (g^{\eta} \cdot g_{\alpha} )\frac{\partial f}{\partial X_\beta} \frac{\partial f}{\partial X_\eta}& = g^{\alpha \beta} \frac{\partial f}{\partial X_\beta}  \delta_{\eta \alpha} \frac{\partial \tilde{x}_j }{\partial X_\eta} \frac{\partial f}{\partial \tilde{x}_j}\\
& = g^{\alpha \beta} \frac{\partial f}{\partial X_\beta} \frac{\partial \tilde{x}_j }{\partial X_\alpha} \frac{\partial f}{\partial \tilde{x}_j}=  g^{\alpha \beta} \frac{\partial f}{\partial X_\beta} \frac{\partial f}{\partial X_\alpha}.
\end{align*}
Here $\delta_{\eta \alpha}$ is Kronecker's delta. Therefore we have \eqref{eq329}.
 \end{proof}

\begin{proof}[Proof of Propositions \ref{prop311}-\ref{prop313}]
From Lemmas \ref{lem315} and \ref{lem316}, we have Propositions \ref{prop311}-\ref{prop313}.
 \end{proof}

\section{Variations of the Kinetic, Dissipation Energies, and Work}\label{sect4}

In this section we study variations of several energies for compressible fluid on an evolving surface. Throughout Section \ref{sect4} we follow Conventions \ref{conv35} and \ref{conv310}.

In subsection \ref{subsec41} we consider variation of the flow map to the action integral to prove Theorems \ref{thm14} and \ref{thm19}. In subsection \ref{subsec42} we calculate variation of the velocity to the dissipation energies and work for the fluid on the evolving surface to prove Theorems \ref{thm15}, \ref{thm16}, and \ref{thm17}.

\subsection{Variation of the Flow Maps to Action Integral}\label{subsec41}
Let us study variation of the flow map to the action integral. We call $\rho = \rho (x,t)$ and $\rho^\varepsilon = \rho^\varepsilon (x,t)$ the two \emph{densities} of the fluid on $\Gamma (t)$ and $\Gamma^\varepsilon (t)$, respectively, if $\rho$ and $\rho^\varepsilon$ satisfy
\begin{equation*}
\begin{cases}
D_t \rho + ( {\rm{div}}_\Gamma v ) \rho  = 0 &\text{ on } \mathcal{S}_T,\\
\rho |_{t = 0} = \rho_0 &\text{ on }  \Gamma ( 0 ),
\end{cases}
\begin{cases}
D_t^\varepsilon \rho^\varepsilon + ( {\rm{div}}_{\Gamma^\varepsilon} v^\varepsilon ) \rho^\varepsilon  = 0 &\text{ on } \mathcal{S}^\varepsilon_T,\\
\rho^\varepsilon |_{t = 0} = \rho_0 &\text{ on }  \Gamma ( 0 ).
\end{cases}
\end{equation*}
Suppose there are $\hat{y} \in [ C^\infty ( \mathbb{R}^3 \times \mathbb{R})]^3$ and $z \in [C^\infty (\mathcal{S}_T)]^3$ such that for every $\xi \in \Gamma_0 (=\Gamma(0))$ and $0 \leq t < T$,
\begin{align*} 
&  \hat{x}^\varepsilon ( \xi ,t) \bigg|_{\varepsilon = 0}= \hat{x} ( \xi , t ),\\
& v^\varepsilon (\hat{x}^\varepsilon ( \xi ,t ) ,t )\bigg|_{\varepsilon = 0} = v (\hat{x} ( \xi , t ),t),\\
& \frac{d}{d \varepsilon} \bigg|_{\varepsilon = 0} \hat{x}^\varepsilon ( \xi ,t) = \hat{y} ( \xi , t ),\\
& z (\hat{x} ( \xi , t ) , t) = \hat{y} ( \xi, t ).
\end{align*}
For each variation $\hat{x}^\varepsilon = \hat{x}^\varepsilon (\xi , t)$ of the flow map $x = \hat{x} ( \xi , t )$, we define the action integral $A_B [\hat{x}^\varepsilon]$ by
\begin{equation*}
A_B [ \hat{x}^\varepsilon ] = - \int_0^T \int_{\Gamma^\varepsilon (t)} \bigg( \frac{1}{2} \rho^\varepsilon (x,t) |v^\varepsilon (x , t)|^2 - p(\rho^\varepsilon (x,t) ) \bigg) { \ }d \mathcal{H}_x^2 d t,
\end{equation*}
where $\rho^\varepsilon$ is the density of the fluid on $\Gamma^\varepsilon (t)$ and $v^\varepsilon$ is the determined by the flow map $\hat{x}^\varepsilon (\xi , t)$. Note that 
\begin{equation*}
\Gamma^\varepsilon (t) = \{  x \in \mathbb{R}^3;{ \ }x = \hat{x}^\varepsilon (\xi , t ) , { \ }\xi \in \Gamma_0 \} .
\end{equation*}

We begin by discussing some properties of $\hat{y} = \hat{y} (\xi , t)$.
\begin{lemma}\label{lem41}
Set 
\begin{equation*}
\tilde{y} (X,t) = \hat{y} ( \Phi (X),t) = \hat{y} (\xi , t)
\end{equation*}
for $\xi \in \Gamma_0$ and $0 \leq t < T$. Then
\begin{align}
& \tilde{y} (X,0) = \hat{y} (\xi, 0)  = 0,\label{eq41}\\
& \frac{1}{J}\frac{d}{d \varepsilon} \bigg|_{\varepsilon =0} J^\varepsilon = 2 g^\alpha \cdot \frac{ \partial \tilde{y}}{\partial X_\alpha}. \label{eq42}
\end{align}
\end{lemma}

\begin{proof}[Proof of Lemma \ref{lem41}]
We first show \eqref{eq41}. Since
\begin{equation*}
\hat{x}^\varepsilon (\xi, 0) - \hat{x} (\xi , 0) = \xi - \xi =0,
\end{equation*}
we find that
\begin{equation*}
\frac{d}{d \varepsilon }\bigg|_{\varepsilon = 0}\hat{x}^\varepsilon (\xi, 0) = 0 = \hat{y} (\xi ,0 ) = \tilde{y} (X,0).
\end{equation*}

Next we prove \eqref{eq42}. From $J^\varepsilon = g_{11}^\varepsilon g_{22}^\varepsilon - g_{12}^\varepsilon g_{21}^\varepsilon$ and $g^\varepsilon_{\alpha \beta} = g^\varepsilon_\alpha \cdot g^\varepsilon_\beta$, we have
\begin{multline*}
\frac{d}{d \varepsilon} \bigg|_{\varepsilon =0} J^\varepsilon = 2 \left( g_1 \cdot \frac{\partial \tilde{y}}{\partial X_1} \right) g_{22} + 2 \left( g_2 \cdot \frac{\partial \tilde{y}}{\partial X_2} \right) g_{11} \\- 2 \left( g_1 \cdot \frac{\partial \tilde{y}}{\partial X_2} \right) g_{21} - 2 \left( g_2 \cdot \frac{\partial \tilde{y}}{\partial X_1} \right) g_{12}.
\end{multline*}
Since $(g^{\alpha \beta})_{2 \times 2} = (g_{\alpha \beta})_{2 \times 2}^{-1}$ and $g^{\beta \alpha} = g^{\alpha \beta}$, we see that
\begin{multline*}
\frac{1}{ J}\frac{d}{d \varepsilon} \bigg|_{\varepsilon =0} J^\varepsilon = 2 \left( g_1 \cdot \frac{\partial \tilde{y}}{\partial X_1} \right) g^{11} + 2 \left( g_2 \cdot \frac{\partial \tilde{y}}{\partial X_2} \right) g^{22} \\+ 2 \left( g_1 \cdot \frac{\partial \tilde{y}}{\partial X_2} \right) g^{21} + 2 \left( g_2 \cdot \frac{\partial \tilde{y}}{\partial X_1} \right) g^{12}\\ 
= 2 \left( (g_1 g^{11} + g_2 g^{12}) \cdot \frac{\partial \tilde{y}}{\partial X_1} \right) + 2 \left( (g_1 g^{21} + g_2 g^{22}) \cdot \frac{\partial \tilde{y}}{\partial X_2} \right)\\
=  2 \left( g^1 \cdot \frac{\partial \tilde{y}}{\partial X_1} \right) + 2 \left( g^2 \cdot \frac{\partial \tilde{y}}{\partial X_2} \right).
\end{multline*}
Thus, we have \eqref{eq42}.
 \end{proof}

Let us attack variation of the flow map to the action integral.
\begin{proposition}\label{prop42}
Assume for every $\xi \in \Gamma_0$ and $0 \leq t < T$,
\begin{equation*}
\rho^\varepsilon(\hat{x}^\varepsilon (\xi,t), t)|_{\varepsilon =0} = \rho (x( \xi, t), t) .
\end{equation*}
Then
\begin{equation}\label{eq43}
\frac{d }{ d \varepsilon } \bigg|_{\varepsilon = 0} A_B [ \hat{x}^\varepsilon ] = \int_0^T \int_{\Gamma (t)} \{ \rho D_t v + {\rm{grad}}_\Gamma \mathfrak{p} + \mathfrak{p} H_\Gamma n \} (x,t) \cdot z (x,t){ \ }d \mathcal{H}^2_{x} d t
\end{equation}
where $\mathfrak{p} = \rho p'( \rho ) - p (\rho)$.
\end{proposition}

\begin{proof}[Proof of Proposition \ref{prop42}] Set $\tilde{\rho}_0 (X) = \rho_0 (\tilde{x}(X,0)) \sqrt{J(X,0)}$. From Lemma \ref{lem315}, we find that
\begin{multline*}
\int_{\Gamma^\varepsilon (t)} \left\{ \frac{1}{2} \rho^\varepsilon (x,t) | v (x,t) |^2 - p (\rho^\varepsilon (x,t)) \right\} { \ }d \mathcal{H}_{x}^2\\
 = \int_U \tilde{\Psi} (X) \frac{1}{2} \tilde{\rho}_0 (X) | \tilde{x}^\varepsilon_t (X, t) |^2 d X - \int_U \tilde{\Psi} (X)  p \left( \frac{\tilde{\rho}_0 (X)}{\sqrt{J^\varepsilon (X,t)}} \right) \sqrt{J^\varepsilon (X,t)} d X.
\end{multline*}
A direct calculation yields
\begin{multline*}
\frac{d}{ d \varepsilon} \bigg|_{\varepsilon = 0} \int_0^T \int_U \tilde{\Psi} (X) \frac{1}{2} \tilde{\rho}_0 (X) | \tilde{x}^\varepsilon_t (X, t) |^2 d X d t\\
= \int_0^T \int_U \tilde{\Psi} (X) \tilde{\rho}_0 (X)  \tilde{y}_t (X, t) \cdot \tilde{x}_t (X,t)  d X d t = \text{ (R.H.S.)}.
\end{multline*}
Integrating by parts with \eqref{eq41} and using Proposition \ref{prop311}, we observe that
\begin{multline*}
\text{ (R.H.S.)}= \int_0^T \int_U \tilde{\Psi} (X) \tilde{\rho}_0 (X)  \tilde{y}_t (X, t) \cdot v (\tilde{x}(X,t) , t ) d X d t \\
= - \int_0^T \int_U \tilde{\Psi} (X) \tilde{\rho}_0 (X) \tilde{y} (X,t) \cdot \frac{d}{dt} [v (\tilde{x}(X,t) , t ) ] d X d t\\
= - \int_0^T \int_U \tilde{\Psi} (X) \frac{\tilde{\rho}_0 (X)}{\sqrt{J (X,t)}} [\{ D_t v \} ( \tilde{x}(X,t) , t ) ] \cdot \tilde{y} (X,t) \sqrt{J (X,t)} d X d t\\
= - \int_0^T \int_{\Gamma (t)} \rho (x,t) D_t v (x,t) \cdot z (x,t) { \ }d \mathcal{H}^2_x d t.
\end{multline*}
Therefore we see that
\begin{equation*}
\frac{d}{d \varepsilon}\bigg|_{\varepsilon = 0} \int_0^T \int_{\Gamma^\varepsilon (t)} \left\{ \frac{1}{2} \rho^\varepsilon |v^\varepsilon |^2 \right\} (x,t) { \ } d \mathcal{H}^2_x d t = - \int_0^T \int_{\Gamma (t)} \{ \rho D_t v \cdot z \} (x,t) { \ }d \mathcal{H}^2_x d t .
\end{equation*}
On the other hand, a direct calculation shows that
\begin{multline*}
\frac{d}{ d \varepsilon} \int_U \tilde{\Psi} (X)  p \left( \frac{\tilde{\rho}_0 (X)}{\sqrt{J^\varepsilon (X,t)}} \right) \sqrt{J^\varepsilon (X,t)} d X\\
= \int_U \tilde{\Psi} (X)  p' \left( \frac{\tilde{\rho}_0 (X)}{\sqrt{J^\varepsilon (X,t)}} \right) \left( \frac{d}{d \varepsilon} \frac{\tilde{\rho}_0 (X)}{\sqrt{J^\varepsilon (X,t)}}  \right) \sqrt{J^\varepsilon (X,t)} d X\\
+ \int_U \tilde{\Psi} (X)  p \left( \frac{\tilde{\rho}_0 (X)}{\sqrt{J^\varepsilon (X,t)}} \right) \left( \frac{d}{d \varepsilon}\sqrt{J^\varepsilon (X,t)} \right) d X\\
= - \frac{1}{2} \int_U \tilde{\Psi} (X)  p' \left( \frac{\tilde{\rho}_0 (X)}{\sqrt{J^\varepsilon (X,t)}} \right) \left( \frac{\tilde{\rho}_0 (X)}{\sqrt{J^\varepsilon (X,t)}}  \right) \left( \frac{1}{J^\varepsilon}\frac{d}{d \varepsilon} J^\varepsilon \right) \sqrt{J^\varepsilon (X,t)} d X\\
+ \frac{1}{2} \int_U \tilde{\Psi} (X)  p \left( \frac{\tilde{\rho}_0 (X)}{\sqrt{J^\varepsilon (X,t)}} \right) \left( \frac{1}{J^\varepsilon}\frac{d}{d \varepsilon} J^\varepsilon \right) \sqrt{J^\varepsilon (X,t)}  d X.
\end{multline*}
Applying Lemma \ref{lem37} and \eqref{eq42}, we see that
\begin{multline*}
\frac{d}{ d \varepsilon} \bigg|_{\varepsilon = 0} \int_U \tilde{\Psi} (X)  p \left( \frac{\tilde{\rho}_0 (X)}{\sqrt{J^\varepsilon (X,t)}} \right) \sqrt{J^\varepsilon (X,t)} d X\\
= - \int_U \tilde{\Psi} (X)  p' \left( \frac{\tilde{\rho}_0 (X)}{\sqrt{J (X,t)}} \right) \left( \frac{\tilde{\rho}_0 (X)}{\sqrt{J (X,t)}}  \right) \left( g^\alpha \cdot \frac{\partial \tilde{y}}{\partial X_\alpha} \right) \sqrt{J (X,t)} d X\\
+ \int_U \tilde{\Psi} (X)  p \left( \frac{\tilde{\rho}_0 (X)}{\sqrt{J (X,t)}} \right) \left( g^\alpha \cdot \frac{\partial \tilde{y}}{\partial X_\alpha} \right) \sqrt{J (X,t)}  d X\\
= \int_{\Gamma (t)} \{ - p'(\rho) \rho + p (\rho) \} (x,t) \{ {\rm{div}}_\Gamma z \} (x,t){ \ }d \mathcal{H}^2_x. 
\end{multline*}
Note that $z ( \hat{x} ( \xi , t ) , t) = \hat{y} (\xi , t )$. Since $\Gamma (t)$ is a closed surface, we use integration by parts (Lemma \ref{lem28}) to check that
\begin{equation*}
\int_{\Gamma (t)} \{ (- p'(\rho ) \rho + p) {\rm{div}}_\Gamma z \}(x,t) { \ }d \mathcal{H}^2_x = \int_{\Gamma (t)} \{ ({\rm{grad}}_\Gamma \mathfrak{p} + \mathfrak{p} H_\Gamma n) \cdot z \} (x,t){ \ }d \mathcal{H}^2_x,
\end{equation*}
where $\mathfrak{p} = \rho p' (\rho) - p (\rho)$. Thus, we see that
\begin{multline*}
\frac{d }{ d \varepsilon } \bigg|_{\varepsilon = 0} \int_0^T \int_{\Gamma^\varepsilon (t)} - p (\rho^\varepsilon (x,t) ) { \ }d \mathcal{H}^2_x d t\\
= - \int_0^T \int_{\Gamma (t)} \{ {\rm{grad}}_\Gamma \mathfrak{p} + \mathfrak{p} H_\Gamma n \} (x,t) \cdot z (x,t){ \ }d \mathcal{H}^2_{x} d t.
\end{multline*}
Therefore Proposition \ref{prop42} is proved.
 \end{proof}
Let us complete the proof of Theorems \ref{thm14} and \ref{thm19}.
\begin{proof}[Proof of Theorems \ref{thm14} and \ref{thm19}]
Since the former part of Theorems \ref{thm14} and \ref{thm19} has been already proved by Proposition \ref{prop42}, we give the proof of the assertion $(\mathrm{ii})$ of Theorem \ref{thm19}. Assume that for each $z \in [C_0^\infty ( \mathcal{S}_T )]^3$ satisfying $z \cdot n =0$,
\begin{equation*}
\int_0^T \int_{\Gamma (t)} \{ \rho D_t v+ {\rm{grad}}_\Gamma \mathfrak{p} + \mathfrak{p} H_\Gamma n \} ( x ,t ) \cdot z (x,t){ \ }d \mathcal{H}^2_{x} d t = 0.
\end{equation*}
From the fact that for $f = f (x,t) = { }^t (f_1 , f_2, f_3)$,
\begin{align*}
\int_0^T \int_{\Gamma (t)} f (x,t) \cdot z (x,t){ \ }d \mathcal{H}^2_{x} d t & = \int_0^T \int_{\Gamma (t)}  f (x,t) \cdot \{ P_\Gamma z \} (x ,t){ \ }d \mathcal{H}^2_{x} d t\\ 
& = \int_0^T \int_{\Gamma (t)} \{ P_\Gamma f \} (x,t) \cdot z (x ,t){ \ }d \mathcal{H}^2_{x} d t,
\end{align*}
we conclude that
\begin{equation*}
P_\Gamma \rho D_t v + {\rm{grad}}_\Gamma \mathfrak{p}  = 0.
\end{equation*}
Note that $P_\Gamma ( \mathfrak{p} H_\Gamma n ) = { }^t ( 0 , 0 , 0 )$. Therefore Theorems \ref{thm14} and \ref{thm19} are proved. 
 \end{proof}

\subsection{Variation of the Velocity to Dissipation Energies and Work}\label{subsec42}
In this subsection we consider variation of the velocity to the dissipation energies and work for the fluid on the evolving surface $\Gamma (t)$ to prove Theorems \ref{thm15}, \ref{thm16}, and \ref{thm17}. We first attack the following lemma.
\begin{lemma}\label{lem43}
For each $t \in (0, T)$ and $V = V (x,t) = { }^t (V_1 , V_2, V_3) \in [ C^\infty (\mathcal{S}_T) ]^3$,
\begin{align*}
E_1 [V] (t) := & \int_{\Gamma (t)} \{ ({\rm{div}}_\Gamma V) \sigma \} (x,t) { \ }d \mathcal{H}_x^2,\\
E_2 [V] (t) := & \int_{\Gamma (t)}  \{ \rho F \cdot V \} (x,t) { \ }d \mathcal{H}_x^2,\\
E_3 [V] (t) := & \int_{\Gamma (t)} \left\{ \frac{1}{2}\lambda | {\rm{div}}_\Gamma V |^2 \right\} (x,t){ \ }d \mathcal{H}_x^2,\\
E_4 [V] (t) := & \int_{\Gamma (t)} \left\{ \frac{1}{2} \mu | D_\Gamma (V ) |^2 \right\} (x,t) { \ }d \mathcal{H}_x^2.
\end{align*}
Then for all $0 < t < T$ and $\varphi = { }^t (\varphi_1 , \varphi_2 , \varphi_3) \in [ C_0^\infty ( \Gamma (t) )]^3$
\begin{align}
\frac{d}{d \varepsilon} \bigg|_{ \varepsilon = 0} E_1 [ v + \varepsilon \varphi] (t) = & - \int_{\Gamma ( t )} \{ {\rm{div}}_\Gamma (P_\Gamma \sigma ) \} (x,t) \cdot \varphi (x) { \ }d \mathcal{H}_x^2,\label{eq44}\\
\frac{d}{d \varepsilon} \bigg|_{ \varepsilon = 0} E_2 [ v + \varepsilon \varphi] (t) = & \int_{\Gamma ( t )} \{ \rho F \} (x,t) \cdot \varphi (x) { \ }d \mathcal{H}_x^2,\label{eq45}\\
\frac{d}{d \varepsilon} \bigg|_{ \varepsilon = 0} E_3 [ v + \varepsilon \varphi] (t) = & - \int_{\Gamma ( t )} \{ {\rm{div}}_\Gamma ( \lambda P_\Gamma ({\rm{div}}_\Gamma v ) ) \} (x,t) \cdot \varphi (x){ \ }d \mathcal{H}_x^2,\label{eq46}\\
\frac{d}{d \varepsilon} \bigg|_{ \varepsilon = 0} E_4 [ v + \varepsilon \varphi] (t) = & - \int_{\Gamma ( t )} \{ {\rm{div}}_\Gamma ( \mu D_\Gamma (v) ) \} (x,t) \cdot \varphi (x) { \ }d \mathcal{H}_x^2,\label{eq47}
\end{align}
Here $D_\Gamma ( V ) = P_\Gamma D ( V ) P_\Gamma$ and $D ( V ) = \{ (\nabla V ) + { }^t (\nabla V )\}/2$.
\end{lemma}

\begin{proof}[Proof of Lemma \ref{lem43}]
Fix $t \in (0 , T )$ and $\varphi = { }^t ( \varphi_1, \varphi_2 , \varphi_3) \in [ C_0^\infty ( \Gamma ( t ) ) ]^3$. We first show \eqref{eq44} and \eqref{eq45}. Using integration by parts (Lemma \ref{lem28}) and Lemma \ref{lem25}, we see that
\begin{align*}
\frac{d}{d \varepsilon} \bigg|_{ \varepsilon = 0} E_1 [ v + \varepsilon \varphi] (t) =& \int_{\Gamma (t)} \sigma (x,t) {\rm{div}}_\Gamma \varphi (x) { \ }d \mathcal{H}_x^2,\\
= & - \int_{\Gamma (t)} \{ {\rm{grad}}_\Gamma \sigma + \sigma H_\Gamma n \} (x,t) \cdot \varphi (x){ \ }d \mathcal{H}_x^2,\\
= & - \int_{\Gamma (t)} \{ {\rm{div}}_\Gamma (P_\Gamma \sigma ) \} (x,t) \cdot \varphi (x) { \ }d \mathcal{H}_x^2.
\end{align*}
Thus, we have \eqref{eq44}. A direct calculation gives
\begin{equation*}
\frac{d}{d \varepsilon} \bigg|_{ \varepsilon = 0} E_2 [ v + \varepsilon \varphi] (t)= \int_{\Gamma ( t )} \rho (x,t) F(x,t) \cdot \varphi (x) { \ }d \mathcal{H}_x^2,
\end{equation*}
which is \eqref{eq45}.

Next we prove \eqref{eq46}. Using integration by parts and \eqref{eq27}, we check that
\begin{multline*}
\frac{d}{d \varepsilon} \bigg|_{ \varepsilon = 0} E_3 [ v + \varepsilon \varphi] (t) = \int_{\Gamma (t)} \{ \lambda( {\rm{div}}_\Gamma v ) ( {\rm{div}}_\Gamma \varphi ) \} (x,t){ \ } d \mathcal{H}^2_x \\
= - \int_{\Gamma (t)} (\lambda {\rm{div}}_\Gamma v) H_\Gamma n \cdot \varphi { \ } d \mathcal{H}^2_x - \int_{\Gamma (t)} {\rm{grad}}_\Gamma ( \lambda {\rm{div}}_\Gamma v) \cdot \varphi { \ } d \mathcal{H}_x^2\\
= - \int_{\Gamma (t)} \{ {\rm{div}}_\Gamma ( \lambda P_\Gamma ({\rm{div}}_\Gamma v ) ) \} (x,t) \cdot \varphi (x) { \ }d \mathcal{H}^2_x.   
\end{multline*}
Therefore we see \eqref{eq46}.

Finally we prove \eqref{eq47}. It is clear that
\begin{equation*}
\frac{d}{d \varepsilon} \bigg|_{ \varepsilon = 0} E_4 [ v + \varepsilon \varphi] (t)= \int_{\Gamma ( t )} \mu D_\Gamma (v) : D_\Gamma ( \varphi ) { \ }d \mathcal{H}_x^2.
\end{equation*}
Since
\begin{equation*}
{\rm{div}}_\Gamma \{ \mu D_\Gamma (v) \} \cdot \varphi = ( \partial_1^\Gamma \{ \mu [D_\Gamma (v)]_{i 1} \} + \partial_\Gamma^2 \{ \mu [D_\Gamma (v)]_{i 2}\} + \partial_3^\Gamma \{ \mu [D_\Gamma (v)]_{i 3} \} )\varphi_i ,
\end{equation*}
we apply Lemma \ref{lem28} to see that
\begin{multline*}
\int_{\Gamma (t)} {\rm{div}}_\Gamma \{ \mu D_\Gamma (v) \} \cdot \varphi { \ }d \mathcal{H}^2_x = - \int_{\Gamma (t)} \mu D_\Gamma (v) : \mathbb{D}_\Gamma (\varphi) { \ }d \mathcal{H}^2_x \\
- \int_{\Gamma (t)} (H_\Gamma n_1 [\mu D_\Gamma (v)]_{i 1} + H_\Gamma n_2 [\mu D_\Gamma (v)]_{i 2} + H_\Gamma n_3 [\mu D_\Gamma (v)]_{i3}) \varphi_i { \ }d \mathcal{H}^2_x.
\end{multline*}
By definition, we observe that
\begin{align*}
n_j [D_\Gamma (v)]_{i j} & = n_1 \{\partial^\Gamma_i v_1 + \partial^\Gamma_1 v_i - n_i (n \cdot \partial_1^\Gamma v) - n_1 (n \cdot \partial_i^\Gamma v )\}\\
&{ \ \ \ } + n_2 \{\partial^\Gamma_i v_2 + \partial^\Gamma_2 v_i - n_i (n \cdot \partial_2^\Gamma v) - n_2 (n \cdot \partial_i^\Gamma v )\}\\
&{ \ \ \ } + n_3 \{\partial^\Gamma_i v_3 + \partial^\Gamma_3 v_i - n_i (n \cdot \partial_3^\Gamma v) - n_3 (n \cdot \partial_i^\Gamma v )\}\\
& = (n \cdot \partial_i^\Gamma v) - (n_1^2+ n_2^2 + n_3^2) (n \cdot \partial_i^\Gamma v ) = 0 .
\end{align*}
Since $D_\Gamma (v) : \mathbb{D}_\Gamma (v) = D_\Gamma (v) : D_\Gamma (\varphi)$ by Lemma \ref{lem26}, we have
\begin{equation*}
\int_{\Gamma (t)} {\rm{div}}_\Gamma \{ \mu D_\Gamma (v) \} \cdot \varphi { \ }d \mathcal{H}^2_x = - \int_{\Gamma (t)} \mu D_\Gamma (v) : D_\Gamma (\varphi) { \ }d \mathcal{H}^2_x.
\end{equation*}
Therefore we see \eqref{eq47}.
 \end{proof}

\begin{proof}[Proof of Theorem \ref{thm15}]
Applying Lemma \ref{lem43}, we have Theorem \ref{thm15}. Remark that if $\varphi \cdot n =0$ then $P_\Gamma \varphi = \varphi$.
 \end{proof}

Finally, we prove Theorems \ref{thm16} and \ref{thm17}. Since we can prove Theorem \ref{thm16} by applying Theorem \ref{thm17}, we only prove Theorem \ref{thm17}.
\begin{proof}[Proof of Theorem \ref{thm17}]
Let $e_{\mathcal{J}} \in C^1 ([0, \infty ))$ or $e_{\mathcal{J}} \in C^1 ((0,\infty))$. Fix $f \in C^{2,0} (\mathcal{S}_T)$ with $| {\rm{grad}}_\Gamma f | \neq 0$. Suppose that $e_{\mathcal{J}}$ is a non-negative function. Since
\begin{equation*}
E_{GD}[f + \varepsilon \varphi ] (t) = - \int_{\Gamma (t)} \frac{1}{2} e_{\mathcal{J}} (| {\rm{grad}}_\Gamma ( f + \varepsilon \varphi ) |^2 ) { \ }d \mathcal{H}^2_x,
\end{equation*}
we use integration by parts to see that
\begin{align*}
\frac{d}{d \varepsilon} \bigg|_{\varepsilon = 0} E_{GD}[f + \varepsilon \varphi ] (t) & = - \int_{\Gamma (t)} e_{\mathcal{J}}' (| {\rm{grad}}_\Gamma f |^2 ){\rm{grad}}_\Gamma f \cdot {\rm{grad}}_\Gamma \varphi { \ }d \mathcal{H}^2_x\\
& = \int_{\Gamma (t)} {\rm{div}}_\Gamma \{ e_{\mathcal{J}}' (| {\rm{grad}}_\Gamma f |^2 ){\rm{grad}}_\Gamma f \} \varphi { \ }d \mathcal{H}^2_x.
\end{align*}
Since
\begin{equation*}{ }^t
\begin{pmatrix}
\frac{\partial \mathcal{E}_{GD}}{\partial \vartheta_1 } ,\frac{\partial \mathcal{E}_{GD}}{\partial \vartheta_2 }, \frac{\partial \mathcal{E}_{GD}}{\partial \vartheta_3 }
\end{pmatrix} = - e_{\mathcal{J}}' ( \vartheta_1^2 + \vartheta_2^2 + \vartheta_3^2 ) { }^t (\vartheta_1 , \vartheta_2 , \vartheta_3 ),
\end{equation*}
we find that
\begin{equation*}{ }^t
\begin{pmatrix}
\frac{\partial \mathcal{E}_{GD}}{\partial \vartheta_1 } ,\frac{\partial \mathcal{E}_{GD}}{\partial \vartheta_2 }, \frac{\partial \mathcal{E}_{GD}}{\partial \vartheta_3 }
\end{pmatrix}\bigg|_{ ( \vartheta_1 = \partial_1^\Gamma f, \vartheta_2 = \partial_2^\Gamma f, \vartheta_3 = \partial_3^\Gamma f )} = - e_{\mathcal{J}}' (| {\rm{grad}}_\Gamma f |^2) {\rm{grad}}_\Gamma f.
\end{equation*}
Therefore Theorem \ref{thm17} is proved.
 \end{proof}

\section{Energetic Variational Approaches for Compressible Fluid Systems}\label{sect5}{ \ }\\

In this section we apply our energetic variational approaches, the thermodynamic theory, Proposition \ref{prop13}, and Theorems \ref{thm14}-\ref{thm19} to make mathematical models of compressible fluid flow on an evolving surface. Moreover, we derive the two generalized heat and diffusion systems on an evolving surface from an energetic point of view. Throughout Section \ref{sect5} we follow Conventions \ref{conv35} and \ref{conv310}.

The outline of this section is as follows: We apply our energetic variational approaches to derive the full compressible fluid system \eqref{eq11}, the tangential compressible fluid system \eqref{eq111}, the non-canonical compressible fluid system \eqref{eq112}, the barotropic compressible fluid systems \eqref{eq113} and \eqref{eq114}, and the generalized heat and diffusion systems \eqref{eq115} and \eqref{eq116} in subsections \ref{subsec51}, \ref{subsec54}, \ref{subsec55}, \ref{subsec56}, and \ref{subsec57}, respectively. In subsection \ref{subsec52}, we study the enthalpy, entropy, free energy, and conservative form of the system \eqref{eq11}. In subsection \ref{subsec53}, we study conservation laws of the system \eqref{eq11} to prove Theorem \ref{thm18}.

\subsection{Energetic Variational Approach for Full Compressible Fluid System}\label{subsec51}

Let us apply our energetic variational approaches to derive the full compressible fluid system \eqref{eq11} on the evolving surface $\Gamma (t)$. We assume that $\Gamma (t)$ is flowed by the total velocity $v$. We set the energy densities for compressible fluid as in Assumption \ref{ass12}. Based on Proposition \ref{prop13}, we set the continuity equation on the evolving surface as follows:
\begin{equation}\label{eq51}
D_t \rho + ({\rm{div}}_\Gamma v) \rho = 0 \text{ on }\mathcal{S}_T.
\end{equation}

We first derive the momentum equation of the full compressible fluid system. Applying an energetic variational approach with Theorems \ref{thm14} and \ref{thm15}, we have the following momentum equation:
\begin{equation}\label{eq52}
\rho D_t v = {\rm{div}}_\Gamma S_\Gamma (v , \sigma , \mu , \lambda ) + \rho F \text{ on }\mathcal{S}_T,
\end{equation}
where $S_\Gamma (v , \sigma , \mu , \lambda) = 2 \mu D_\Gamma (v) + \lambda P_\Gamma ({\rm{div}}_\Gamma v ) - P_\Gamma \sigma$. Here we consider variations on the total velocity $v$. More precisely, we set the action integral $A[ \hat{x}]$, the dissipation energy $E_D[v]$, and the work $E_W[v]$ as follows:
\begin{align}
A [ \hat{x}] & = - \int_0^T \int_{\Gamma (t)} \left\{ \frac{1}{2} \rho |v|^2 \right\} (x,t) { \ } d \mathcal{H}^2_x d t, \label{eq53}\\
E_D[ v ] &= - \int_{\Gamma (t)} \left\{ \frac{1}{2} ( 2 \mu |D_\Gamma (v)|^2 + \lambda |{\rm{div}}_\Gamma v |^2 ) \right\} (x,t) { \ }d \mathcal{H}^2_x,\label{eq54}\\
E_W [v] & = \int_{\Gamma (t)} \{  ({\rm{div}}_\Gamma v ) \sigma + \rho F \cdot v \} (x,t) { \ }d \mathcal{H}^2_x.\label{eq55}
\end{align}
Applying Theorems \ref{thm14} and \ref{thm15}, we consider variations ${d }/{d \varepsilon}|_{\varepsilon = 0} A[\hat{x}^\varepsilon]$ with \\$d \hat{x}^\varepsilon /d t = v^\varepsilon$, ${d}/{d \varepsilon}|_{\varepsilon = 0} E_D [v + \varepsilon \varphi]$, and ${d}/{d \varepsilon}|_{\varepsilon = 0} E_{W}[v + \varepsilon \varphi]$ to have
\begin{align}
\frac{\delta A}{\delta \hat{x}} = & \rho D_t v ,\label{eq56}\\
\frac{\delta E_D}{\delta v} = & {\rm{div}}_\Gamma (2 \mu D_\Gamma (v) + \lambda P_\Gamma ({\rm{div}}_\Gamma v)), \label{eq57}\\
\frac{\delta E_W}{\delta v} = &- {\rm{div}}_\Gamma (P_\Gamma \sigma ) + \rho F. \label{eq58}
\end{align}
We assume the following energetic variational principle:
\begin{equation*}
\frac{\delta A}{\delta \hat{x}} = \frac{\delta E_{D+W}}{\delta v} \left( = \frac{\delta E_D}{\delta v} + \frac{\delta E_W}{\delta v} \right).
\end{equation*}
This is \eqref{eq52}.

Secondly, we study the internal energy of compressible fluid. From Theorem \ref{thm16}, we have the following forces:
\begin{align*}
\frac{\delta E_{TD}}{\delta \theta} = {\rm{div}}_\Gamma (\kappa {\rm{grad}}_\Gamma \theta ),\\
\frac{\delta E_{SD}}{\delta C} = {\rm{div}}_\Gamma (\nu {\rm{grad}}_\Gamma C).
\end{align*}
Set
\begin{align*}
q_\theta & = \kappa {\rm{grad}}_\Gamma \theta ,\\
q_C &  = \nu {\rm{grad}}_\Gamma C.
\end{align*}
In order to use the the first law of thermodynamics, we consider the energy dissipation due to the viscosities and the work done by the pressure. Using \eqref{eq52} and integration by parts (Lemma \ref{lem28}), we check that
\begin{multline*}
\frac{d}{d t}\int_{\Gamma (t)} \left\{ \frac{1}{2} \rho |v|^2 \right\} (x,t) { \ }d \mathcal{H}^2_x =  \int_{\Gamma (t)} \rho D_t v \cdot v { \ } d\mathcal{H}^2_x\\
= \int_{\Gamma (t)} ({\rm{div}}_\Gamma S_\Gamma (v,\sigma , \mu , \lambda ) + \rho F) \cdot v { \ } d\mathcal{H}^2_x\\
= \int_{\Gamma (t)} \{ - (2 \mu |D_\Gamma (v)|^2 + \lambda |{\rm{div}}_\Gamma v |^2 ) + ({\rm{div}}_\Gamma v ) \sigma + \rho F \cdot v \} (x,t) { \ } d\mathcal{H}^2_x.
\end{multline*}
Integrating with respect to time, we find that for $0 < t_1 < t_2 < T$,
\begin{multline*}
\int_{\Gamma (t_2)} \frac{1}{2} \rho | v|^2 { \ }d \mathcal{H}^2_x + \int_{t_1}^{t_2} \int_{\Gamma ( \tau )} \{ (2 \mu |D_\Gamma (v)|^2 + \lambda |{\rm{div}}_\Gamma v |^2 ) - ({\rm{div}}_\Gamma v ) \sigma \} { \ }d \mathcal{H}^2_x d \tau \\
=  \int_{\Gamma (t_1)} \frac{1}{2} \rho | v|^2 { \ }d \mathcal{H}^2_x + \int_{t_1}^{t_2} \int_{\Gamma ( \tau )} \rho F \cdot v { \ }d \mathcal{H}^2_x d \tau .
\end{multline*}
This shows that $(2 \mu |D_\Gamma (v)|^2 + \lambda |{\rm{div}}_\Gamma v |^2 )$ is the density for the energy dissipation due to the viscosities and that $ ({\rm{div}}_\Gamma v ) \sigma $ is the work done by the pressure of our compressible fluid system. Set
\begin{equation*}
\tilde{e}_D = 2 \mu |D_\Gamma (v)|^2 + \lambda |{\rm{div}}_\Gamma v |^2 .
\end{equation*}
Since $\rho$ satisfies \eqref{eq51}, we use Lemma \ref{lem39} and an argument in the proof of Proposition \ref{prop13} to see that for $\Omega (t) \subset \Gamma (t)$ flowed by the total velocity $v$,
\begin{align*}
\frac{d}{d t} \int_{\Omega (t)} \{ \rho e \} (x,t) { \ }d \mathcal{H}^2_x = & \int_{\Omega (t)} \{ D_t \left( \rho e \right) + \rho e({\rm{div}}_\Gamma v ) \} (x,t) { \ }d \mathcal{H}^2_x\\
= & \int_{\Omega (t)} \{ \rho D_t e \} (x,t) { \ }d \mathcal{H}^2_x.
\end{align*}
Now we apply the first law of thermodynamics, that is, we assume that for arbitrary $\Omega ( t ) \subset \Gamma (t)$ flowed by the total velocity $v$,
\begin{equation*}
\frac{d}{d t} \int_{\Omega (t)} \{ \rho e \} (x,t) { \ }d \mathcal{H}^2_x = \int_{\Omega (t)} \left\{ \frac{\delta E_{TD}}{\delta \theta} + \rho Q_\theta + \tilde{e}_D - ({\rm{div}}_\Gamma v ) \sigma \right\} (x,t) { \ }d \mathcal{H}^2_x.
\end{equation*}
Then we have
\begin{equation*}
\rho D_t e = {\rm{div}}_\Gamma q_\theta+ \rho Q_\theta + e_D - ({\rm{div}}_\Gamma v) \sigma \text{ on } \mathcal{S}_T.
\end{equation*}
This is equivalent to
\begin{equation}\label{eq59}
\rho D_t e + ({\rm{div}}_\Gamma v ) \sigma = {\rm{div}}_\Gamma q_\theta + \rho Q_\theta + \tilde{e}_D \text{ on } \mathcal{S}_T.
\end{equation}

Finally we derive the diffusion system of our compressible fluid system. We assume that the change of rate of the concentration $C$ equals to the force derived from a variation of the energy dissipation due to surface diffusion, that is for arbitrary $\Omega ( t ) \subset \Gamma (t)$ flowed by the velocity $v$,
\begin{equation*}
\frac{d}{d t} \int_{\Omega (t)} C { \ } d \mathcal{H}^2_x = \int_{\Omega (t)} \left( \frac{ \delta E_{SD}}{\delta C} + Q_C \right) { \ } d \mathcal{H}_x^2.
\end{equation*}
From Proposition \ref{prop13}, we obtain
\begin{equation}\label{eq510}
D_t C + ({\rm{div}}_\Gamma v) C = {\rm{div}}_\Gamma q_C + Q_C \text{ on } \mathcal{S}_T.
\end{equation}
Combining \eqref{eq51}, \eqref{eq52}, \eqref{eq59}, and \eqref{eq510}, we have the full compressible fluid system \eqref{eq11}.

\subsection{On Full Compressible Fluid System}\label{subsec52}

In this section we study conservative form of the full compressible fluid system \eqref{eq11}, and investigate the enthalpy, entropy, and free energy of the system \eqref{eq11}.

We first consider the total energy. Set $e_A = \{ \rho |v|^2 \}/2 + \rho e$. We apply Lemmas \ref{lem26} and \ref{lem27} to observe that
\begin{multline*}
D_t^N e_A + {\rm{div}}_\Gamma (e_A v ) = D_t e_A + ({\rm{div}}_\Gamma v) e_A \\
= \left( \frac{1}{2}|v|^2 D_t \rho + \rho D_t v \cdot v + e D_t \rho + \rho D_t e \right) + ({\rm{div}}_\Gamma v ) \left( \frac{1}{2} \rho |v|^2  + \rho e \right)\\
=  \rho D_t v \cdot v + \rho D_t e\\
= \{ {\rm{div}}_\Gamma \{ S_\Gamma (v , \sigma , \mu , \lambda ) \} \cdot v + \rho F \cdot v \} + \{ {\rm{div}}_\Gamma q_\theta+ \rho Q_\theta + \tilde{e}_D - ({\rm{div}}_\Gamma v ) \sigma \} \\
= {\rm{div}}_\Gamma q_\theta+ \rho Q_\theta + \rho F \cdot v + {\rm{div}}_\Gamma \{ S_\Gamma (v,\sigma , \mu , \lambda ) v \}.
\end{multline*}
Therefore we have 
\begin{equation*}
D_t^N e_A + {\rm{div}}_\Gamma ( e_A v - q_\theta - S_\Gamma (v , \sigma , \mu, \lambda ) v ) = \rho Q_\theta + \rho F \cdot v .
\end{equation*}
Similarly, we see that the system \eqref{eq11} satisfies the conservative form \eqref{eq12}.

Next we investigate the enthalpy, entropy, and free energy of the system \eqref{eq11}. Assume that $\rho , \theta , \mu , \lambda , \kappa $ are positive functions. Suppose that
\begin{equation}\label{eq511}
D_t e = \theta D_t s - \sigma D_t \left( \frac{1}{\rho} \right) \text{ on } \mathcal{S}_T.
\end{equation}
Set $h = e + \sigma / \rho$ and $e_F = e - \theta s $. By the definition of $h$, we find that
\begin{align*}
\rho D_t h = & \rho D_t e + \rho D_t \left( \frac{\sigma}{\rho} \right)\\
= & \{ {\rm{div}}_\Gamma q_\theta+ \rho Q_\theta + \tilde{e}_D - ({\rm{div}}_\Gamma v ) \sigma \} + \{ ({\rm{div}}_\Gamma v ) \sigma + D_t \sigma \}\\
= & {\rm{div}}_\Gamma q_\theta + \rho Q_\theta + \tilde{e}_D + D_t \sigma .
\end{align*}
It is easy to check that
\begin{equation*}
D_t^N ( \rho h ) + {\rm{div}}_\Gamma ( \rho h v - q_\theta ) - \rho Q_\theta = \tilde{e}_D + D_t \sigma .
\end{equation*}
Since $\rho$ satisfies \eqref{eq51}, we see that
\begin{equation*}
D_t \left( \frac{1}{\rho} \right) = - \frac{D_t \rho }{\rho^2} = \frac{{\rm{div}}_\Gamma v}{\rho}.
\end{equation*}
By \eqref{eq511}, we check that
\begin{align*}
\theta \rho D_t s = & \rho D_t e + \sigma \rho D_t \left( \frac{1}{\rho} \right) \\
= & ( {\rm{div}}_\Gamma q_\theta + \rho Q_\theta + \tilde{e}_D - ({\rm{div}}_\Gamma v ) \sigma ) + ({\rm{div}}_\Gamma v ) \sigma \\
= & {\rm{div}}_\Gamma q_\theta + \rho Q_\theta + \tilde{e}_D  . 
\end{align*}
Therefore we have
\begin{equation*}
\theta \rho D_t s = {\rm{div}}_\Gamma q_\theta+ \rho Q_\theta + \tilde{e}_D. 
\end{equation*}
Using Lemma \ref{lem27}, we observe that
\begin{align*}
D_t^N (\rho s) +{\rm{div}}_\Gamma (\rho s v ) & = D_t (\rho s) + ({\rm{div}}_\Gamma v) (\rho s )\\
& = \rho D_t s + \{ D_t \rho + ({\rm{div}}_\Gamma v ) \rho \} s\\
 & = \rho D_t s = \frac{1}{\theta} \left( {\rm{div}}_\Gamma q_\theta+ \rho Q_\theta + \tilde{e}_D \right) .
\end{align*}
A direct calculation gives
\begin{align*}
{\rm{div}}_\Gamma \left( - \frac{q_\theta }{\theta} \right) = {\rm{div}}_\Gamma \left( - \frac{ \kappa {\rm{grad}}_\Gamma \theta }{\theta} \right) = - \frac{ {\rm{div}}_\Gamma ( \kappa  {\rm{grad}}_\Gamma \theta )}{\theta} + \frac{\kappa |{\rm{grad}}_\Gamma \theta |^2}{\theta^2}.
\end{align*}
Consequently, we see that
\begin{equation*}
D_t^N (\rho s) +{\rm{div}}_\Gamma (\rho s v ) - {\rm{div}}_\Gamma \left( \frac{ q_\theta }{\theta} \right) - \frac{ \rho Q_\theta }{\theta}  = \frac{ \tilde{e}_D }{\theta} + \frac{\kappa |{\rm{grad}}_\Gamma \theta |^2}{\theta^2} \geq 0.
\end{equation*}
Since $e_F = e - \theta s$, we check that
\begin{align*}
\rho D_t e_F + \rho s D_t \theta = & (\rho D_t e - \rho s D_t \theta - \rho \theta D_t s ) + \rho s D_t \theta\\
= & \rho D_t e - \rho \theta D_t s\\
= & - ({\rm{div}}_\Gamma v ) \sigma .
\end{align*}
By \eqref{eq216}, we see that
\begin{equation*}
\rho D_t e_F + \rho s D_t \theta - S_\Gamma (v , \sigma , \mu, \lambda) : D_\Gamma (v) = - \tilde{e}_D \leq 0.
\end{equation*}
Therefore we have the enthalpy, entropy, and free energy of the system \eqref{eq11}.

\subsection{Conservation Laws of Full Compressible Fluid System}\label{subsec53}

In this section, we investigate conservation laws of the system \eqref{eq11} to prove Theorem \ref{thm18}.
\begin{proof}[Proof of Theorem \ref{thm18}] 
Assume that $\Gamma (t)$ is flowed by the total velocity $v$. We first show \eqref{eq117}. Since $\rho $ satisfies
\begin{equation*}
D_t^N \rho +{\rm{div}}_\Gamma (\rho v) = 0,
\end{equation*}
we use Lemmas \ref{lem27} and \ref{lem39} to see that
\begin{equation*}
\frac{d}{d t} \int_{\Gamma (t)} \rho (x,t) {  \ } d\mathcal{H}^2_x = 0.
\end{equation*}
Integrating with respect to time, we see that for $0 <t_1 < t_2 < T$,
\begin{equation*}
\int_{\Gamma (t_2)} \rho (x, t_2) { \ } d \mathcal{H}^2_x = \int_{\Gamma (t_1)} \rho (x, t_1) { \ } d \mathcal{H}^2_x.
\end{equation*}

Next we show \eqref{eq118}. From
\begin{equation*}
D_t^N (\rho v) + {\rm{div}}_\Gamma (\rho v \otimes v) = {\rm{div}}_\Gamma S_\Gamma (v,\sigma , \mu , \lambda  ) + \rho F,
\end{equation*}
we check that
\begin{align*}
\frac{d}{d t} \int_{\Gamma (t)} \{ \rho v \} (x,t) { \ } d\mathcal{H}^2_x = & \int_{\Gamma (t)} \{  D_t^N (\rho v) + {\rm{div}}_\Gamma (\rho v \otimes v) \}(x,t) { \ } d\mathcal{H}^2_x\\
= & \int_{\Gamma (t)} \{ {\rm{div}}_\Gamma S_\Gamma (v , \sigma , \mu , \lambda ) + \rho F \} (x,t) { \ }d \mathcal{H}^2_x.
\end{align*}
Using integration by parts (Lemma \ref{lem28}), we find that
\begin{align*}
\frac{d}{d t} \int_{\Gamma (t)} \{ \rho v \} (x,t) { \ } d\mathcal{H}^2_x = \int_{\Gamma (t)} \{\rho F \} (x,t) { \ }d \mathcal{H}^2_x.
\end{align*}
Remark that $\Gamma (t)$ is a closed surface and that $S_\Gamma (v , \sigma , \mu , \lambda ) \cdot n =0$. Integrating with respect to time, we see the law of conservation of momentum.

Next we derive the law of conservation of angular momentum \eqref{eq121}. Since $D_t x = 2 v$ and $v \times v = 0$, we find that
\begin{align*}
\frac{d}{d t} \int_{\Gamma (t)} x \times (\{ \rho v \} (x,t)) { \ } d \mathcal{H}^2_x & = \int_{\Gamma (t)} x \times \rho D_t v { \ }d \mathcal{H}^2_x\\
& = \int_{\Gamma (t)} x \times \{ {\rm{div}}_\Gamma S_\Gamma (v , \sigma , \mu , \lambda ) + \rho F \}{ \ }d \mathcal{H}^2_x.  
\end{align*}
Set $M = S_\Gamma (v , \sigma , \mu , \lambda )$. It is clear that
\footnotesize
\begin{multline*}
x \times M = \\
\begin{pmatrix}
x_2 ( \partial_1^\Gamma [M]_{31} + \partial_2^\Gamma [M]_{32} + \partial_3^\Gamma [M]_{33} ) - x_3 ( \partial_1^\Gamma [M]_{21} + \partial_2^\Gamma [M]_{22} + \partial_3^\Gamma [M]_{23} ) \\
x_3 ( \partial_1^\Gamma [M]_{11} + \partial_2^\Gamma [M]_{12} + \partial_3^\Gamma [M]_{13} ) - x_1 ( \partial_1^\Gamma [M]_{31} + \partial_2^\Gamma [M]_{32} + \partial_3^\Gamma [M]_{33} ) \\
x_1 ( \partial_1^\Gamma [M]_{21} + \partial_2^\Gamma [M]_{22} + \partial_3^\Gamma [M]_{23} ) - x_2 ( \partial_1^\Gamma [M]_{11} + \partial_2^\Gamma [M]_{12} + \partial_3^\Gamma [M]_{13} ) 
\end{pmatrix}.
\end{multline*}\normalsize
We now prove that for each $i,j =1,2,3,$
\begin{multline*}
\int_{\Gamma (t)} \{ x_i ( \partial_1^\Gamma [M]_{j 1} + \partial_2^\Gamma [M]_{j 2} + \partial_3^\Gamma [M]_{j 3} ) \\- x_j ( \partial_1^\Gamma [M]_{i 1} + \partial_2^\Gamma [M]_{i 2} + \partial_3^\Gamma [M]_{i 3} ) \} { \ }d \mathcal{H}^2_x = 0.
\end{multline*}
Fix $i$ and $j$. Using the integration by parts and the fact that $n \cdot M = { }^t (0,0,0)$, we see that
\begin{multline*}
\int_{\Gamma (t)} \{ x_i ( \partial_1^\Gamma [M]_{j 1} + \partial_2^\Gamma [M]_{j 2} + \partial_3^\Gamma [M]_{j 3} ) { \ }d \mathcal{H}^2_x\\
=  - \int_{\Gamma (t)} \{ (n_i + H x_i) ( n_1 [M]_{j 1} + n_2 [M]_{j 2} + n_3 [M]_{j 3} ) + [M]_{j i} \} { \ }d \mathcal{H}^2_x\\
= - \int_{\Gamma (t)} [M]_{j i} { \ }d \mathcal{H}^2_x.
\end{multline*}
Since $[M]_{j i } = [M]_{i j}$, we check that
\begin{multline*}
\int_{\Gamma (t)} \{ x_i ( \partial_1^\Gamma [M]_{j 1} + \partial_2^\Gamma [M]_{j 2} + \partial_3^\Gamma [M]_{j 3} ) \\
- x_j ( \partial_1^\Gamma [M]_{i 1} + \partial_2^\Gamma [M]_{i 2} + \partial_3^\Gamma [M]_{i 3} ) \} { \ }d \mathcal{H}^2_x\\
= - \int_{\Gamma (t)} [M]_{j i} { \ }d \mathcal{H}^2_x + \int_{\Gamma (t)} [M]_{i j} { \ }d \mathcal{H}^2_x = 0.
\end{multline*}
Therefore we conclude that
\begin{equation*}
\int_{\Gamma (t)} x \times {\rm{div}}_\Gamma S_\Gamma (v , \sigma , \mu , \lambda ) { \ }d \mathcal{H}^2_x = { }^t ( 0 , 0 , 0) .
\end{equation*}
As a result, we have
\begin{equation*}
\frac{d}{d t} \int_{\Gamma (t)} x \times ( \rho v ) { \ }d \mathcal{H}^2_x = \int_{\Gamma (t)} x \times (\rho F) { \ } d \mathcal{H}^2_x. 
\end{equation*}
Integrating with respect to time, we find that
\begin{multline*}
\int_{\Gamma (t_2)} x \times \{ \rho v \} (x, t_2) { \ } d \mathcal{H}^2_x\\
 = \int_{\Gamma (t_1)} x \times \{ \rho v \} (x, t_1) { \ } d \mathcal{H}^2_x + \int_{t_1}^{t_2} \int_{\Gamma (t_1)} x \times \{ \rho F \} (x, \tau ) { \ } d \mathcal{H}^2_x  d \tau .
\end{multline*}

Finally, we prove \eqref{eq119} and \eqref{eq120}. Since $e_A$ and $C$ satisfy
\begin{align*}
D_t^N e_A + {\rm{div}}_\Gamma (e_A v) & = {\rm{div}} (S_\Gamma (v,\sigma , \mu , \lambda ) v ) + \rho Q_\theta +  \rho F \cdot v,\\
D_t^N C + {\rm{div}}_\Gamma (C v) & = {\rm{div}}_\Gamma q_C + Q_C,
\end{align*}
we use the previous argument to deduce \eqref{eq119} and \eqref{eq120}. Therefore Theorem \ref{thm18} is proved.
 \end{proof}

\subsection{Energetic Variational Approach for Tangential Compressible Fluid System}\label{subsec54}

Let us apply our energetic variational approaches to derive the tangential compressible fluid system \eqref{eq111} on the evolving surface $\Gamma (t)$. We assume that $\Gamma (t)$ is flowed by the total velocity $v$. We set the energy densities for compressible fluid as in Assumption \ref{ass12}. Based on Proposition \ref{prop13}, we set the continuity equation on the evolving surface as follows:
\begin{equation}\label{eq512}
D_t \rho + ({\rm{div}}_\Gamma v) \rho = 0 \text{ on }\mathcal{S}_T.
\end{equation}

We first derive the momentum equation of the tangential compressible fluid system. Applying an energetic variational principle with Theorems \ref{thm14} and \ref{thm15}, we derive the following momentum equation:
\begin{equation}\label{eq513}
P_\Gamma \rho D_t v = P_\Gamma {\rm{div}}_\Gamma S_\Gamma (v , \sigma , \mu , \lambda ) + P_\Gamma \rho F \text{ on }\mathcal{S}_T.
\end{equation}
Here we consider variations on the tangential part of the total velocity $v$. More precisely, we assume the following energetic variational principle:
\begin{equation*}
\frac{\delta A}{\delta \hat{x}} \bigg|_{z \cdot n = 0} = \frac{\delta E_{D+W}}{\delta v} \bigg|_{\varphi \cdot n =0}.
\end{equation*}
That is,
\begin{equation*}
P_\Gamma \rho D_t v = P_\Gamma {\rm{div}}_\Gamma S_\Gamma (v , \sigma , \mu , \lambda ) + P_\Gamma \rho F,
\end{equation*}
where $A$, $E_D$, and $E_W$ are the action integral, dissipation energy, and work defined as in Subsection \ref{subsec51}. On the basis of Propositions \ref{prop312}, \ref{prop313}, and Theorem \ref{thm16}, we set $q_\theta = \kappa {\rm{grad}}_\Gamma \theta$ and $q_C = \nu {\rm{grad}}_\Gamma C$. 

Next we study the internal energy. Assume that $v \cdot n =0$. Applying \eqref{eq513} and integration by parts (Lemma \ref{lem28}), we check that
\begin{multline*}
\frac{d}{d t}\int_{\Gamma (t)} \left\{ \frac{1}{2} \rho |v|^2 \right\} (x,t) { \ }d \mathcal{H}^2_x =  \int_{\Gamma (t)} P_\Gamma \rho D_t v \cdot v { \ } d\mathcal{H}^2_x\\
= \int_{\Gamma (t)} ({\rm{div}}_\Gamma S_\Gamma (v,\sigma , \mu , \lambda ) + P_\Gamma \rho F) \cdot v { \ } d\mathcal{H}^2_x\\
= \int_{\Gamma (t)} \{ - (2 \mu |D_\Gamma (v)|^2 + \lambda |{\rm{div}}_\Gamma v |^2 ) + ({\rm{div}}_\Gamma v ) \sigma + P_\Gamma \rho F \cdot v \} (x,t) { \ } d\mathcal{H}^2_x.
\end{multline*}
Note that $P_\Gamma v = v$. Integrating with respect to time, we find that for $0 < t_1 < t_2 < T$,
\begin{multline*}
\int_{\Gamma (t_2)} \frac{1}{2} \rho | v|^2 { \ }d \mathcal{H}^2_x + \int_{t_1}^{t_2} \int_{\Gamma ( \tau )} \{ (2 \mu |D_\Gamma (v)|^2 + \lambda |{\rm{div}}_\Gamma v |^2 ) - ({\rm{div}}_\Gamma v ) \sigma \} { \ }d \mathcal{H}^2_x d \tau \\
=  \int_{\Gamma (t_1)} \frac{1}{2} \rho | v|^2 { \ }d \mathcal{H}^2_x + \int_{t_1}^{t_2} \int_{\Gamma ( \tau )} P_\Gamma \rho F \cdot v { \ }d \mathcal{H}^2_x d \tau .
\end{multline*}
This shows that $(2 \mu |D_\Gamma (v)|^2 + \lambda |{\rm{div}}_\Gamma v |^2 )$ is the density for the energy dissipation due to the viscosities and that $\sigma ({\rm{div}}_\Gamma v ) $ is the work done by the pressure of our compressible fluid system. Set
\begin{equation*}
\tilde{e}_D = 2 \mu |D_\Gamma (v)|^2 + \lambda |{\rm{div}}_\Gamma v |^2 .
\end{equation*}
Applying the first law of thermodynamics, we obtain
\begin{equation}\label{eq514}
\rho D_t e + ({\rm{div}}_\Gamma v ) \sigma = {\rm{div}}_\Gamma q_\theta+ \rho Q_\theta + \tilde{e}_D \text{ on } \mathcal{S}_T.
\end{equation}
By an similar argument to derive the diffusion system in subsection \ref{subsec51}, we have
\begin{equation}\label{eq515}
D_t C + ({\rm{div}}_\Gamma v) C = {\rm{div}}_\Gamma q_C + Q_C.
\end{equation}
Therefore we have the tangential compressible fluid system \eqref{eq111} on the evolving surface by combining \eqref{eq512}-\eqref{eq515}.

\subsection{Derivation of Non-canonical Compressible Fluid System}\label{subsec55}

Let us consider compressible fluid flow on the evolving surface $\Gamma (t)$ from a different point of view. Based on Proposition \ref{prop13} we admit \eqref{eq51}. We set the action integral $A[\hat{x}]$ and the work $E_{W}[v]$ defined by \eqref{eq53} and \eqref{eq55}, respectively. We set the dissipation energy $E_D[u]$ as follows:
\begin{align*}
E_D[ u ] = - \int_{\Gamma (t)} \frac{1}{2} \{ 2 \mu |D_\Gamma (u)|^2 + \lambda |{\rm{div}}_\Gamma u |^2 \} (x,t) { \ }d \mathcal{H}^2_x.
\end{align*}
Using arguments similar to those in the proof of Theorems \ref{thm14} and \ref{thm15}, we have \eqref{eq56}, \eqref{eq58}, and
\begin{equation*}
\frac{\delta E_D}{\delta u} = P_\Gamma {\rm{div}}_\Gamma (2 \mu D_\Gamma (v) + \lambda P_\Gamma ({\rm{div}}_\Gamma v ) ) = P_\Gamma {\rm{div}}_\Gamma S_{\Gamma } (u,0, \mu , \lambda ).
\end{equation*}
We assume the following energetic variational principle : 
\begin{equation*}
\frac{\delta A [ \hat{x} ]}{\delta \hat{x}} = \frac{\delta E_D[u]}{\delta u} + \frac{\delta E_{W}[v]}{\delta v}
\end{equation*}
to have
\begin{equation*}
\rho D_t v = P_\Gamma {\rm{div}}_\Gamma S_{\Gamma } (u,0, \mu , \lambda ) - {\rm{grad}}_\Gamma \sigma  - \sigma H_\Gamma n + \rho F.
\end{equation*}
Therefore we have the non-canonical compressible fluid system \eqref{eq112} on the evolving surface.

\subsection{Derivation of Barotropic Compressible Fluid Systems}\label{subsec56}

Let us consider the barotropic compressible fluid flow on the evolving surface $\Gamma (t)$. Based on Proposition \ref{prop13}, we set the continuity equation on an evolving surface as follows:
\begin{equation*}
D_t \rho + ({\rm{div}}_\Gamma v) \rho = 0 \text{ on }\mathcal{S}_T.
\end{equation*}
On the basis of Proposition \ref{prop311}, we set the total energy $e_B$ as follows:
\begin{equation*}
e_B = \frac{1}{2} \rho | v |^2 - p (\rho).
\end{equation*}
From Theorem \ref{thm19}, we obtain the two barotropic compressible fluid systems \eqref{eq113} and \eqref{eq114}.

\subsection{Derivation of Generalized Heat and Diffusion Systems}\label{subsec57}

Let us derive the generalized heat and diffusion systems on the evolving surface $\Gamma (t)$. From Proposition \ref{prop13}, we set the continuity equation on the evolving surface as follows:
\begin{equation*}
D_t \rho + ({\rm{div}}_\Gamma v) \rho = 0 \text{ on }\mathcal{S}_T.
\end{equation*}
Let $e_{\mathcal{J}_1} , e_{\mathcal{J}_2} \in C^1 ([0, \infty ))$ or $e_{\mathcal{J}_1}, e_{\mathcal{J}_2} \in C^1 ((0, \infty))$. Suppose that $e_{\mathcal{J}_1}$, $e_{\mathcal{J}_2}$ are two non-negative function. Based on Proposition \ref{prop313} we set the following energy densities:
\begin{equation*}
e_{TD} = e_{ \mathcal{J}_1 } (|{\rm{grad}}_\Gamma \theta |^2) \text{ and } e_{SD} = e_{ \mathcal{J}_2} (|{\rm{grad}}_\Gamma C |^2).
\end{equation*}
From Theorem \ref{thm17}, we have the following forces:
\begin{align*}
\frac{\delta E_{TD}}{\delta \theta} = & {\rm{div}}_\Gamma \{ e_{\mathcal{J}_1}' (| {\rm{grad}}_\Gamma \theta |^2) {\rm{grad}}_\Gamma \theta \},\\
\frac{\delta E_{SD}}{\delta C} =& {\rm{div}}_\Gamma \{ e_{\mathcal{J}_2}' (| {\rm{grad}}_\Gamma C |^2) {\rm{grad}}_\Gamma C \}.
\end{align*}
Let $\mathcal{F}_1 = \mathcal{F}_1 (x,t)$ and $\mathcal{F}_2 = \mathcal{F}_2 ( x, t) $ be two smooth functions. Assume that for every $\Omega (t) \subset \Gamma (t)$ flowed by the total velocity $v$,
\begin{align*}
& \frac{d}{d t} \int_{\Omega (t)} (\rho C_\theta \theta ) (x,t) { \ }d \mathcal{H}^2_x = \int_{\Omega (t)} \left\{  \frac{\delta E_{TD}}{\delta \theta} + \rho Q_\theta + \mathcal{F}_1 \right\} { \ } d \mathcal{H}^2_x,\\
& \frac{d}{d t} \int_{\Omega (t)} C (x,t) { \ }d \mathcal{H}^2_x = \int_{\Omega (t)} \left\{ \frac{\delta E_{SD}}{\delta C} + Q_C + \mathcal{F}_2 \right\} { \ } d \mathcal{H}^2_x.
\end{align*}
Then we have
\begin{align*}
\rho D_t (C_\theta \theta ) & = {\rm{div}}_\Gamma \{ e_{\mathcal{J}_1}' (| {\rm{grad}}_\Gamma \theta |^2) {\rm{grad}}_\Gamma \theta \} + \rho Q_\theta + \mathcal{F}_1 \text{ on } \mathcal{S}_T,\\
D_t C + ({\rm{div}}_\Gamma v ) C & = {\rm{div}}_\Gamma \{ e_{\mathcal{J}_2}' (| {\rm{grad}}_\Gamma C |^2) {\rm{grad}}_\Gamma C \} + Q_C + \mathcal{F}_2 \text{ on } \mathcal{S}_T.
\end{align*}
Therefore we have the generalized heat system \eqref{eq115} and diffusion system \eqref{eq116} on an evolving surface.

\begin{flushleft}
{\bf{Acknowledgments}}
The author would like to thank Professor Chun Liu for valuable discussions about energetic variational approaches. The author would like to acknowledge the hospitality and supports of Penn State University (U.S.A.) during his visits. Part of this work was done when the author was an assistant professor at Waseda University (Japan), whose hospitality and kindness are gratefully acknowledged. This work was partly supported by the Japan Society for the Promotion of Science (JSPS) KAKENHI Grant Numbers JP25887048 and JP15K17580.
\end{flushleft}

\begin{flushleft}
{\bf{Conflict of interest :}} The author declares that they have no conflict of interest.
\end{flushleft}

\end{document}